\documentclass[11pt]{article}

\usepackage{xspace}
\usepackage{amsmath}
\usepackage{amssymb}
\usepackage{amsthm}
\usepackage{fullpage,times}

\newtheorem{theorem}{Theorem}

\newtheorem{claim}[theorem]{Claim}

\newtheorem{corollary}[theorem]{Corollary}

\newtheorem{definition}[theorem]{Definition}

\newtheorem{property}[theorem]{Property}
\newtheorem{lemma}[theorem]{Lemma}

\newtheorem{proposition}[theorem]{Proposition}
\newtheorem{remark}[theorem]{Remark}

\newcommand{\Xomit}[1]{ }
\newcommand{\algo}[1]{\ensuremath{\mbox{\textsc{#1}}}\xspace}

\newcommand{\opt}{\algo{opt}}
\newcommand{\sol}{\algo{Sol}}

\newcommand{\R}{{\cal R}}
\newcommand{\I}{{\cal I}}
\newcommand{\C}{{\cal C}}
\newcommand{\cO}{{\cal O}}

\newcommand{\J}{{\cal J}}
\newcommand{\s}{{\cal S}}

\newcommand{\ak}{A^{(k)}}
\newcommand{\akk}{A^{(k')}}
\newcommand{\akkk}{A^{(k'')}}
\newcommand{\solk}{\algo{Sol}^{(k)}}
\newcommand{\solka}{\algo{Sol'}^{(k)}}

\newcommand{\solkb}{\algo{Sol''}^{(k)}}

\newcommand{\solkab}{{Sol}^{(k)}}
\newcommand{\solkkab}{{Sol}^{(k')}}
\newcommand{\solkkkab}{{Sol}^{(k'')}}

\newcommand{\del}{\delta}

\newcommand{\eps}{\varepsilon}

\makeatletter

\newcommand{\Rmnum}[1]{\expandafter\@slowromancap\romannumeral #1@}
\makeatother

\begin{document}

\title{Minimum total weighted completion time: \\ Faster approximation schemes}

\author {Leah Epstein\thanks{Department of Mathematics, University
of Haifa,  Haifa, Israel. {\tt lea@math.haifa.ac.il}.} \and
  Asaf Levin\thanks{Faculty of Industrial Engineering and Management, The Technion,  Haifa, Israel. {\tt levinas@ie.technion.ac.il.}}}\date{}

\maketitle

\begin{abstract} We study classic scheduling problems
on {\it uniformly related machines}. Efficient polynomial time
approximation schemes (EPTAS's) are fast and practical
approximation schemes.  New methods and techniques are essential
in developing such improved approximation schemes, and their
design is a primary goal of this research agenda. We present
EPTAS's for the scheduling problem of a set of jobs on uniformly
related machines so as to minimize the total weighted completion
time, both for the case with release dates and its special case
without release dates. These problems are NP-hard in the strong
sense, and therefore EPTAS's are the best possible approximation
schemes unless P=NP. Previously, only PTAS's were known for these
two problems, while an EPTAS was known only for the special case
of identical machines without release dates.
\end{abstract}

\section{Introduction}
We consider one of the most basic multiprocessor scheduling
problems: scheduling on uniformly related machines with the goal
of minimizing the total weighted completion times of jobs. More
precisely, we are given a set of $n$ jobs, where each job has a
positive size, a positive weight, and a non-negative release date
(also called release time) associated with it. We are also given a
set of $m$ machines for their processing, such that each machine
has a positive speed. Processing job $j$ on machine $i$ requires
the allocation of a time interval on this machine, where its
length is precisely the size of $j$ divided by the speed of $i$.
We consider non-preemptive schedules and thus every job is
assigned to a machine and to a single time slot on that machine,
such that the following conditions hold.  First, the length of the
time slot allocated for job $j$ (on one of the machines) is the
time it takes to process $j$ on that machine. Second, the time
interval assigned to job $j$ starts no earlier than the release
date of $j$.  Finally, a machine can process at most one job at
each time, and therefore we require that the time intervals
assigned to two jobs that are assigned to a common machine do not
intersect in an inner point. Given such a schedule, the completion
time of job $j$ is the ending point of the time interval assigned
to job $j$, and the weighted completion time of $j$ is the product
of the weight of $j$ and its completion time. The goal is to find
a schedule that minimizes the sum of the weighted completion times
of all jobs.

Before we state our main result, we define the notions of
approximation algorithms and the different types of approximation
schemes.  An $\R$-approximation algorithm for a minimization
problem is a polynomial time algorithm that always finds a
feasible solution of cost at most $\R$ times the cost of an
optimal solution. The infimum value of $\R$ for which an algorithm
is an $\R$-approximation is called the approximation ratio or the
performance guarantee of the algorithm. A polynomial time
approximation scheme (PTAS) is a family of approximation
algorithms such that the family has a $(1+\eps)$-approximation
algorithm for any $\eps >0$. An efficient polynomial time
approximation scheme (EPTAS) is a PTAS whose time complexity is of
the form $f(\frac{1}{\eps}) \cdot poly(n)$ where $f$ is some (not
necessarily polynomial) function and $poly(n)$ is a polynomial
function of the length of the (binary) encoding of the input. A
fully polynomial time approximation scheme (FPTAS), is a stronger
concept, defined like an EPTAS, but the function $f$ must be a
polynomial in $\frac 1{\eps}$. In this paper, we are interested in
EPTAS's and we say that an algorithm (for some problem) has a
polynomial running time complexity if its time complexity is of
the form $f(\frac{1}{\eps}) \cdot poly(n)$. Note that while a PTAS
may have time complexity of the form $n^{g(\frac{1}{\eps})}$,
where $g$ can be linear or even exponential, this cannot be the
case for an EPTAS.  The notion of an EPTAS is modern and finds its
roots in the FPT (fixed parameter tractable) literature (see
\cite{CT97,DF99,FG06,Marx08}). It was introduced in order to
distinguish practical from impractical running times of PTAS's,
for cases where a fully polynomial time approximation scheme
(FPTAS) does not exist (unless P=NP). In this work we design an
EPTAS for the scheduling problem above.

The seminal work of Smith \cite{Smith} established the existence
of polynomial time algorithm for solving the problem of minimizing
the total weighted completion time without release dates on a
single machine. His algorithm can be described as follows. The
jobs are scheduled according to a non-decreasing order of their
densities, starting at time zero, and without any idle time (where
the density of job $j$ is the ratio between its weight and its
size). The correctness of this algorithm follows by a simple
exchange argument.  This algorithm generalizes the SPT
(shortest-processing-time) approach for the case of equal weights.
In our settings, we can conclude the following property for the
problem {\it without} release dates. Once the jobs have been
assigned to machines (but not to time slots), the permutation of
jobs assigned to a given machine is fixed according to Smith's
algorithm. For a constant number of machines (at least two
machines), the problem without release dates is NP-hard in the
ordinary sense, but it is solvable in pseudo-polynomial time and
has an FPTAS \cite{Sa76}. For the case where the number of
machines is a part of the input, the problem is strongly NP-hard
(see e.g. problem SS13 in \cite{GJ79}). The case of equal weights
(and no release dates) is polynomially solvable even for the more
general case of unrelated machines \cite{H73}. The problems with
release dates are known to be much harder, and they are strongly
NP-hard for any constant number of machines, even with equal
weights.  In fact, even the problem with release dates on a single
machine and equal weights is strongly NP-hard \cite{LRB77}.  The
property that our problems are strongly NP-hard excludes the
possibility to design FPTAS's for the problems that we consider in
this work, and thus EPTAS's are the best possible results (unless
P=NP). The problem in the more general setting of unrelated
machines is APX-hard \cite{Hoogeven} both for the case with
release dates and unit weight jobs, and for the case without
release dates (where the jobs have arbitrary weights).

With respect to approximation algorithms, the development of good
approximation algorithms for these problems was fairly slow. Till
the late nineties, only constant approximation algorithms were
developed for min-sum scheduling problems such as the ones we
study. We refer to the papers cited in \cite{SW00,Afrati+99,CK01}
for a restricted survey of such results.  Here we elaborate on the
approximation schemes in the literature.

The first approximation scheme for a special case of our problem
was introduced by Skutella and Woeginger \cite{SW00} and it was
designed for the special case of identical machines and without
release dates. Their scheme is in fact an EPTAS for this problem.
Shortly afterwards, Afrati et al. \cite{Afrati+99} presented
PTAS's for the special case of the problem on identical machines
with release dates (and also for the problem on a constant number
of unrelated machines).  Their approach was generalized by Chekuri
and Khanna \cite{CK01} who showed the existence of a PTAS for the
problem studied here (namely, related machines with release
dates). See also \cite{AM06} for a description of the methods used
in these schemes.

Before explaining the methods and techniques of these last
schemes, as well as the limitations of those approaches, we
mention the state of the art of approximation schemes for
 load balancing problems on identical machines and
uniformly related machines.  The relation between those families
of problems will become clear later. In these load balancing
problems, the goal is to minimize the vector of {\em machines}
completion times (and not job completion times as we study here).
Hochbaum and Shmoys presented the dual approximation framework and
used it to show that the makespan minimization problem has a PTAS
for identical machines \cite{HS87} and for uniformly related
machines \cite{HS88}.  It was noted in \cite{HocBook} that the
PTAS of \cite{HS87} for identical machines can be converted into
an EPTAS by using integer program in fixed dimension instead of
dynamic programming. Recently, Jansen \cite{Ja10} developed an
EPTAS for the makespan minimization problem on uniformly related
machines (see \cite{JR11} for an alternative EPTAS for this
problem).  The $\ell_p$ minimization problem (of the vector of
machine loads) has an EPTAS on identical machines
\cite{AAWY97,AAWY98}, and a PTAS and an EPTAS on uniformly related
machines \cite{ES04,EL14}.

Next, we elaborate on the known approximation schemes for the
special cases of our problems.  Skutella and Woeginger \cite{SW00}
observed that the $\ell_2$ norm minimization of the vector of
machine loads is equivalent to minimizing the total weighted
completion times (without release dates) if the jobs have equal
densities.  Thus, they showed that if the ratio between the
maximum density and the minimum density of jobs is bounded by a
constant, then one can adapt the ideas of Alon et
al.~\cite{AAWY97,AAWY98} and obtain an EPTAS for this restricted
setting.  Their scheme for the problem without release dates on
identical machines follows the following ideas.  First, round all
the job sizes and job weights to integer powers of $1+\eps$. Next,
apply randomized geometric partitioning of the jobs based on their
(rounded) density, solve each sub-instance consisting of all jobs
of one partition using the scheme (which is similar to the one of
\cite{AAWY97,AAWY98}) and schedule the partial solutions for every
machine sorted by non-decreasing job densities.  This last step of
uniting the solutions for the different parts in the partition was
more delicate in \cite{SW00}, but as noted in \cite{Afrati+99},
the last step could be made much simpler. Afrati et al.~also noted
that this approach can be extended to obtain an approximation
scheme (PTAS) for the problem without release dates in other
machines environments (see the last remark in \cite{Afrati+99})
however they wrote that {\em ``unfortunately, this approach breaks
down completely when jobs have release dates"}. More precisely,
major difficulties arise even for one machine (with release dates)
as it is no longer correct that simply  scheduling all the jobs of
one part of the geometric partition before the remaining jobs has
minor affect on the performance of the algorithm.  Afrati et
al.~\cite{Afrati+99} managed to overcome these difficulties by
rounding the input parameters and then structuring the input in a
way that a job is processed relatively quickly after its release
date. Using this structure, each machine has a constant number of
states (where a state of a machine is the next time in which it
becomes available for processing another job), and one can apply
dynamic programming in the time-horizon to schedule the jobs while
recalling at each time the number of machines in each state, and
the number of unscheduled jobs of each large size that were
released at any given time in the last few release dates (and
similarly, for the total volume of jobs seen as small that were
released in those last few release dates).  Chekuri and Khanna
\cite{CK01} extended these techniques of \cite{Afrati+99} to the
setting of uniformly related machines (with release dates).
Observe that the time complexity of the time-horizon dynamic
programming is too large if one seeks for an EPTAS, as the number
of possible states of machines is clearly a constant depending on
$\eps$ and we need to recall the number of machines in each state,
and thus the degree of $m$ in the polynomial of the running time
depends on $\eps$ and this violates the conditions on the running
time of an EPTAS. Thus, the methods of \cite{Afrati+99,CK01} fail
completely when one tries to obtain an EPTAS for these problems
with release dates (even for identical machines).

\paragraph{Paper outline.}
We start our study in Section \ref{sec:no_release_dates}, where we
present an EPTAS for the special case without release dates. Such
a scheme was not known prior to our work (it was only known for
identical machines). This first scheme is simpler than the one for
the general case that we present later, and thus it can serve as
an introduction of a part of the methods which we will use
afterwards for the design of the EPTAS for the variant with
release dates. In Section \ref{sec:release_dates}, we present our
main result, that is an EPTAS for the general case of the problem,
which is the non-preemptive scheduling problem with release dates
on uniformly related machines so as to minimize the total weighted
completion time of the jobs.

In our schemes, all parameters are rounded to powers of a
parameter $1+\del$ (where $\del=\frac{\eps}{\Upsilon}$ for a
constant $\Upsilon>0$ independent of $\eps$)
\cite{SW00,Afrati+99}. In the presence of release dates, rounding
is done carefully so that the resulting release dates do not
include zero. In the case without release dates, we use the
shifting technique \cite{HM85} to separate the input into {\it
bounded} instances, where for each instance, all the jobs have
sufficiently similar densities, and show that each instance can be
solved separately and independently from other instances, and the
solutions can be combined without further modifications. To solve
each instance (due to shifting, there is a polynomial number of
such instances), we use configuration mixed-integer linear
programs (MILPs), where we require a variable of a machine
configuration to be integral if the machine is heavily loaded,
that is, it receives a large total size of jobs. By splitting the
machine set into slow and fast machines, we show that a slow
machine is never heavily loaded in optimal solutions. Moreover,
there is a limit to how much a fast machine can be loaded (in
terms of the load of the first fastest machine that we guess
approximately), and as a result, the number of heavy
configurations is constant. A configuration consists of a complete
list of relatively large jobs together with their densities
(compared to the total size of jobs assigned to this machine), and
an approximate total size of smaller jobs of each density.  The
assignment of jobs that are scheduled as small jobs is carried out
by another set of assignment variables.  These assignment
variables can be fractional in a feasible solution to the MILP. As
in \cite{SW00}, we approximate the increase of cost due to the
contribution of slightly larger set of small jobs using the total
sizes of jobs of the configuration, and thus fractional assignment
becomes possible, and can be converted into an integral one by
slightly increasing the cost of each configuration. As for
fractional configurations, we round down the number of machines
with each such configuration. The total size of jobs that remain
unassigned is sufficiently small to be combined into the schedule
of a heavily loaded machine (at least one such fastest machine
must exist). Our scheme of the similar densities instance is of
independent interest as it extends the methods of \cite{Ja10} for
the $\ell_2$-norm minimization problem (the case of one density
for all jobs) to obtain an EPTAS for a constant number of
densities (where the constant depends on $\del$), whereas
extending the original EPTAS for this $\ell_2$-norm minimization
problem \cite{EL14} results in only a PTAS for the case of
constant number of densities.  Thus, presenting a new EPTAS for
the $\ell_2$-norm minimization problem is indeed necessary for our
generalization.

The general structure of the scheme for  the case with release
dates is as follows. In this case, we are also interested in
reducing the problem into a polynomial set of subproblems. It is
possible to split time into intervals \cite{Afrati+99}, but
obviously those time intervals cannot be treated separately. Each
subproblem which we create is such that it can be solved using a
mixed-integer linear program. Our scheme starts with a similar
flavor to the schemes of \cite{Afrati+99,CK01}.  Namely, we start
by rounding the input parameters and by an iterative procedure
that we call {\it job shifting}, where we delay the release dates
of some of the jobs. More specifically, if the total number (for
relatively large jobs of equal properties) or total size (for
sub-classes of relatively small jobs with similar properties) that
are being released at a given release date is too large, it is
impossible to schedule such a large number or total size of those
jobs by the next release date, so the release date of some jobs
can be altered.  Then, we deviate dramatically from their
approach, and we apply two kinds of shifting. The first is on the
release dates. In this step we increase release dates of a small
portion of the jobs (based on contribution to cost). The goal of
that is to create large time differences between groups of release
dates, so that sub-inputs of very different release dates could be
solved almost independently. It can happen that jobs of much
smaller release dates should be combined into a solution for given
release dates, but in these cases the assignment of such jobs can
be restricted to gaps of the schedule (a kind of idle time), and
by slightly stretching time \cite{Afrati+99}, there are suitable
gaps with sufficient frequency in any schedule. Afterwards, we
apply shifting on densities. The goal is once again to create
inputs with small numbers of possible parameters, such that there
is a major difference between parameters of subproblems. Unlike
the variant without release dates, combining solutions of very
different densities is challenging; jobs cannot run before their
release date, and if one solution is sparse in some time interval,
we would like another solution to take over, in order to exploit
that time, and to avoid postponing all jobs (and thus increasing
their costs significantly). Yet, that solution is still not
sufficiently good in some cases, as the schedule of jobs with very
high densities must be done carefully, and they cannot be delayed
(compared to their schedule in optimal solutions) in any solution.
On the other hand, it harms the solution to delay multiple jobs
(even if their densities are not very large) due to a sparse time
slot (with high density jobs).  Thus, the schedule of different
subproblems needs to be coordinated and these subproblems need to
be taken care of delicately.   To overcome this notoriously
difficult task, we apply a series of guessing steps and
transformations on the solutions to ensure that these bad cases
simply do not happen. Once again, we employ gaps to combine
unassigned jobs and to allow the modifications to the solutions.
Finally, we use the mixed-integer linear program paradigm to help
us to approximate each subproblem, this time, with time-dependent
configurations.

\section{EPTAS for the special case without release dates\label{sec:no_release_dates}}
Here we consider the special case without release dates.  That is,
we assume that all jobs are available for processing at time $0$
 and all release dates are $0$.  Later, we will use some of
the techniques which we develop for this case in the development
of the EPTAS for the general case with release dates.

\paragraph{Properties} Obviously, since any job can start running at any time, an optimal solution (or schedule) never introduces
idle time, and moreover, every machine runs its assigned jobs in
an optimal order, that is, the jobs are sorted by Smith's ratio.
As any tie-breaking policy leads to the same objective function
value, we use a specific tie-breaking policy in this section. More
specifically, we will always break ties in favor of running larger
jobs first, and in the case of equal size jobs (of the same
weight), jobs of smaller indices are scheduled earlier. We call
this ordering {\it the natural ordering}. Thus, we can define a
{\it solution} or a {schedule} as an ordered partition of the jobs
to $m$ subsets, one subset for each machine. In some cases, we
will compute the total weighted completion time of another
permutation (not of the natural ordering), this will be done when
this is easier, and we are only interested in an upper bound on
the objective value.

For an input $X$ and a solution $B$, we let $B(X)$ denote the
output and the objective value of $B$ for the input $X$. Recall
that for a job $j$, the density of $j$ is the ratio between its
weight and its size. Thus, running the jobs sorted according to
the natural ordering is equivalent to first sorting them by
non-increasing density, and for each density, the jobs are sorted
by non-increasing size, breaking ties (among equal size jobs) in
favor of scheduling jobs of smaller indices earlier.

For a fixed schedule, the work of machine $i$ is the total size of
jobs assigned to it, and its load is its completion time, that is,
the work divided by the speed (as we only consider schedules without any idle time). Let the original instance be
denoted by $A$, where job $j$ has size $a_j>0$ and
weight $\omega_j>0$, and machine $i$ has speed $v_i>0$. 
Let $C_j$ denote the completion time of $j$ under a given
schedule, that is, the total size of jobs that run before $j$ on
the same machine (including $j$), divided by the speed of this
machine. Let $\Gamma_j=\omega_j(C_j-\frac {a_j}{2v_i})$ where
$C_j-\frac {a_j}{2v_i}$ is the time when half of job $j$ is
completed. We call these values $\Gamma$-values, and obviously,
the cost of a solution is at least the sum of $\Gamma$-values.
Moreover, for identical machines, the difference between the cost
of a solution and the sum of the $\Gamma$-values is independent of
the solution whereas for related machines, this difference depends
on the speeds.

\begin{lemma}
\label{lboncost}
Consider a set of jobs $\tilde{I} \subseteq A$ assigned to machine $i$. Let $\Phi =\sum_{j \in \tilde{I}} a_j$ and let
$\phi=\min_{j \in \tilde{I}} \frac{\omega_j}{a_j}$ be the minimum density of any job of $\tilde{I}$. The sum of $\Gamma$-values (and the
cost) of any solution that selects to run all jobs of $\tilde{I}$ on machine $i$ (possibly with other jobs) is at least $\frac {\phi \Phi^2}{2v_i}$.
\end{lemma}
\begin{proof}
We prove the claim by induction on $|\tilde{I}|$. The claim obviously holds if $\tilde{I}=\emptyset$. Consider the job $j' \in \tilde{I}$ that is scheduled last (among the jobs of $\tilde{I}$) to run on machine $i$. The completion time of this job is at least $\frac{\Phi}{v_i}$, the time that half of it is completed, is at least $\frac{\Phi-a_{j'}/2}{v_i}$, and its weight satisfies $\omega_j \geq \phi a_{j'}$, thus $\Gamma_{j'} \geq  \frac{\phi a_{j'}(2\Phi-a_{j'})}{2v_i}$. By the inductive hypothesis, the sum of $\Gamma$-values of the jobs of $\tilde{I}\setminus\{j'\}$ is at least $\frac {\phi (\Phi-a_{j'})^2}{2v_i}$. In total, we get at least $\frac {\phi \Phi^2}{2v_i}$.
\end{proof}

\subsection{Rounding and shifting}
Given the original instance $A$, where job $j$ has size $a_j$ and
weight $\omega_j$, and machine $i$ has speed $v_i$, we create an
instance $A'$ as follows. Let $0<\del \leq 1/8$ be an accuracy
factor, that is a function of $\eps$ (where $\del<\eps$), and
such that $\frac{1}{\del}$ is an integer. The sets of jobs and
machines are the same as in $A$. Let $s_i=(1+\del)^{\lfloor
\log_{1+\del} v_i \rfloor}$ be the speed of machines $i$ in
instance $A'$. Let $w_j=(1+\del)^{\lceil \log_{1+\del} \omega_j
\rceil}$, and $p_j=(1+\del)^{\lceil \log_{1+\del} a_j \rceil}$ be
the weight and size (respectively) of job $j$ in instance $A'$.
Let $\cO$ and $\cO'$ denote optimal solutions for $A$ and $A'$
respectively. Let $SOL$ denote a given solution for both
instances. Using the new notation, for a given schedule, the
completion time of $j$ is still denoted by $C_j$, and its weighted
completion time is $w_j\cdot C_j$.

\begin{proposition}
We have $SOL(A)  \leq SOL(A') \leq (1+\del)^3 SOL(A)$ and $\cO(A)  \leq \cO'(A') \leq (1+\del)^3 \cO(A)$.
\end{proposition}
\begin{proof}
We start with the first claim. For each job, we compare its
completion time in $SOL$ for the two inputs. For each job $j$, its
size in $A'$ is at least its size in $A$. Thus, when we consider
the permutation of jobs on each machine as it is in $SOL(A')$ the
total size of jobs that are completed before $j$ together with the
size of $j$ cannot be smaller in $A'$ than this value in $A$. The
speed of the machine running $j$ in $A'$ is at most the speed of
the same machine in $A$. Thus, the completion time of $j$ in $A'$
is at least its completion time in $A$. The weight of $j$ in $A'$
is at least its weight in $A$, and thus its contribution to the
objective function in $A$ (that is, $\omega_j$ times the
completion time in $SOL$ for the input $A$ with the permutations
of jobs as defined by $SOL(A')$) is no larger than this value in
$A'$ (that is, $w_j$ times the completion time in $SOL$ for the
input $A'$), and as the actual permutation of $SOL(A)$ cannot have
a larger cost, we find $SOL(A) \leq SOL(A')$. On the other hand,
when we consider the permutation of jobs as it is in $SOL(A)$ the
total size of jobs that are completed before $j$ together with the
size of $j$ in $A'$ is at most $1+\del$ times the corresponding
value in $A$. The speed of the machine running $j$ in $A'$ is at
least $\frac{1}{1+\del}$ times the speed of the same machine in
$A$. Thus, the completion time of $j$ in $A'$ when we consider the
permutation of jobs on each machine as defined by $SOL(A)$ is at
most its completion time in $A$ times $(1+\del)^2$. The weights of
$j$ in $A$ and in $A'$ satisfy $w_j \leq (1+\del) \omega_j$, and
thus its contribution to the objective function in $A'$ is at most
$(1+\del)^3$ times this value in $A$, and we find $SOL(A') \leq
(1+\del)^3 SOL(A)$.

The second claim follows from the first one; since $\cO$ and
$\cO'$ are optimal solutions for $A$ and $A'$ respectively, we
have $\cO(A) \leq \cO'(A)$ and $\cO'(A') \leq \cO(A')$. Letting
$SOL=\cO'$ we get by the first claim that $\cO'(A) \leq  \cO'(A')$ and letting $SOL=\cO$
we get $\cO(A') \leq (1+\del)^3 \cO(A)$, which proves the claim.
\end{proof}

\begin{corollary}
If a solution $SOL$ satisfies $SOL(A') \leq (1+k\del)\cO'(A')$ for some $k>0$,
then $SOL(A) \leq (1+(2k+4)\del) \cO(A)$.
\end{corollary}
\begin{proof}
We  have $SOL(A) \leq SOL(A') \leq (1+k\del)\cO'(A') \leq (1+k\del)(1+\del)^3 \cO(A)$.
Using $\del \leq \frac 18$ we get $(1+k\del)(1+\del)^3 = 1+(k+3)\del+\del(3\del+\del^2+k(3\del+3\del^2+\del^3)) \leq 1+(k+3)\del+\del(\frac 38+\frac1{64}+k(\frac 38+\frac{3}{64}+\frac{1}{512})) \leq 1+(k+3)\del+\del(k+1)$.
\end{proof}

Consider the instance $A'$. Recall that if $j$ is executed on machine $i$, then
$\Gamma_j=w_j(C_j-\frac {p_j}{2s_i})$ where $C_j-\frac
{p_j}{2s_i}$ is the time when half of the job is completed. We say that for a given schedule, a block of jobs is a set of jobs of equal density that are assigned to one machine to run consecutively on that machine (there may be additional jobs of the same density, each assigned to run later or earlier than these jobs).
The next claim shows that computing the sum of $\Gamma$-values for a block of jobs is a function of their total size, their commom density, and the starting time of the block, and it is independent of the other properties of these jobs.

\begin{claim}\label{gammaval}
Let $\bar{I}\subseteq A'$ be a set of jobs, all having densities equal to some $\Delta>0$, and assume that these jobs are scheduled to run consecutively on machine $i$ starting at time $\tau$. Then $$\sum_{j \in \bar{I}} \Gamma_j=\Delta\cdot(\tau+\frac{1}{2s_i}{\sum_{j \in \bar{I}}p_j}) \cdot ( \sum_{j \in \bar{I}}p_j) \ . $$
\end{claim}
\begin{proof}
We prove the claim by induction on $|\bar{I}|$. Obviously, it holds if
$\bar{I}=\emptyset$. Let $j'$ be the job of $\bar{I}$ assigned to
run last. We have $\Gamma_{j'}= w_{j'}(\tau+\frac{\sum_{j \in
\bar{I}}p_j-p_{j'}/2}{s_i})=\Delta p_{j'} (\tau+\frac{\sum_{j \in
\bar{I}}p_j-p_{j'}/2}{s_i})$. By the inductive hypothesis, $\sum_{j
\in \bar{I}, j \neq j'}
\Gamma_j=\Delta\cdot(\tau+\frac{1}{2s_i}\sum_{j \in \bar{I}, j
\neq j'}p_j)( \sum_{j \in \bar{I}, j \neq j'}p_j)$, and we get
$$\sum_{j \in \bar{I}} \Gamma_j=\Delta p_{j'} (\tau+\frac{\sum_{j \in
\bar{I}}p_j-p_{j'}/2}{s_i})+\Delta\cdot(\tau+\frac{1}{2s_i}\sum_{j
\in \bar{I}, j \neq j'}p_j)( \sum_{j \in \bar{I}, j \neq
j'}p_j)=\Delta\cdot(\tau+\frac{1}{2s_i}{\sum_{j \in \bar{I}}p_j})(
\sum_{j \in \bar{I}}p_j) \ , $$ as required.
\end{proof}

Let $\xi= \lceil \ell \log_{1+\del}\frac 1{\del} \rceil$ for a
fixed integer $3 \leq \ell \leq 5$ (and thus $\frac{1}{\del^3} \leq \frac{1}{\del^{\ell}} \leq
(1+\del)^{\xi} <
\frac{1+\del}{\del^{\ell}}<\frac{2}{\del^{\ell}}<\frac{1}{\del^{\ell+1}}\leq \frac{1}{\del^6}$. Since $(1+\del)^{\frac{1}{\del^2}}>\frac{1}{\del}$, by the integrality of $\frac{\ell}{\del^2}$, we have $\xi \leq \frac{\ell}{\del^2}<\frac{1}{\del^3}$. However, $(\frac 98)^{\xi} \geq  (1+\del)^{\xi} \geq \frac{1}{\del^{\ell}} \geq {8^3}$, so $\xi >50$.

For an integer $c \in \mathbb{Z}$, let
$\Omega_c=\{c\xi+1,\ldots,(c+1)\xi\}$. Let $0 \leq \zeta\leq
\frac{1}{\del^{\ell+1}}-1$ be an integer. We define the instance
$A_{\zeta}$ by modifying the weights of jobs in $A'$. For a job
$j$, if for some integer $v$, $\log_{1+\del} \frac{w_j}{p_j} \in
\Omega_{v/\del^{\ell+1}+\zeta}$ (recall that the density of $j$,
$\frac{w_j}{p_j}$, is an integer power of $1+\del$), then
$w^{\zeta}_j=w_j \cdot (1+\del)^{\xi}$, and otherwise
$w^{\zeta}_j=w_j$. In the first case, we have $w^{\zeta}_j =
(1+\del)^\xi w_j \leq \frac{1+\del}{\del^{\ell}} w_j$. Let
$\cO^{\zeta}$ be an optimal solution for $A_{\zeta}$. As the set
of jobs and machines is the same in $A$, $A'$, and $A_{\zeta}$,
the sets of feasible solutions for the three instances are the
same.  Next, we bound the increase of the cost due to the
transformation from $A'$ to $A_{\zeta}$. As a result, no job  $j
\in A_{\zeta}$ has a density such that $\log_{1+\del}
\frac{w^{\zeta}_j}{p_j} \in \Omega_{v'/\del^{\ell+1}+\zeta}$ for
any integer $v'$ (since by $\log_{1+\del} \frac{w_j}{p_j} \in
\Omega_{v/\del^{\ell+1}+\zeta}$ we have $\log_{1+\del}
\frac{w^{\zeta}_j}{p_j} \in \Omega_{v/\del^{\ell+1}+\zeta+1}$ and
$v/\del^{\ell+1}+\zeta<v/\del^{\ell+1}+\zeta+1<(v+1)/\del^{\ell+1}+\zeta$).
Any value $(1+\del)^{\beta}$ where $\beta \in
\Omega_{v/\del^{\ell+1}+\zeta}$ for an integer $v$ is called a
forbidden density for $\zeta$, and other values $(1+\del)^{\beta}$
for integer $\beta$ are called allowed density for $\zeta$.

\begin{claim}\label{brshiftnorelease}
Given a solution $SOL$, any $0 \leq {\zeta} \leq \frac{1}{\del^{\ell+1}}-1$ satisfies $SOL(A') \leq
SOL(A_{{\zeta}})$, and there exists a value $0 \leq \bar{\zeta} \leq \frac{1}{\del^{\ell+1}}-1$ such that we have $SOL(A_{\bar{\zeta}}) \leq  (1+2\del) SOL(A')$. Additionally, there exists a value $0 \leq \zeta' \leq \frac{1}{\del^{\ell+1}}-1$
such that $\cO'(A') \leq \cO^{\zeta'}(A_{\zeta'})  \leq  (1+2\del) \cO'(A')$, and if a solution $SOL_1$ satisfies $SOL_1(A_{\zeta'}) \leq (1+k'\del))\cO^{\zeta'}(A_{\zeta'})$ for some $k'>0$, then
$SOL_1(A') \leq (1+(2k'+2)\del) \cO'(A')$.
\end{claim}
\begin{proof}
We start with the first claim. Since for any $\zeta$, the weight
of a job in $A_{\zeta}$ is no smaller than the weight of the
corresponding job in $A'$, for any $\zeta$, if we consider the
permutation of jobs as defined by $SOL(A_{\zeta})$, then
$SOL(A')\leq SOL(A_{\zeta})$ holds, and by the optimality of the
permutation of jobs on each machine, we have $SOL(A')\leq
SOL(A_{\zeta})$ where for each instance we use its permutation of
jobs. Seeing $SOL$ as a solution for $A'$, let $SOL_{\zeta}$
denote the total weighted completion time in $SOL$ (for $A'$) of
jobs for which $w_j\neq w^{\zeta}_j$, that is, the contribution of
these jobs to the objective function value for the instance $A'$.
We have $SOL(A')=\sum_{\zeta=0}^{\frac 1{\del^{\ell+1}}-1}
SOL_{\zeta}$, and $SOL(A_{\zeta}) \leq \sum_{0 \leq \eta \leq
\frac 1{\del^{\ell+1}}-1, \eta \neq \zeta} SOL_{\eta}+
(1+\del)^{\xi}SOL_{\zeta} \leq SOL(A')+(1+\del)^{\xi}SOL_{\zeta}$,
where the first inequality holds by computing an upper bound on
$SOL(A_{\zeta})$ by considering the permutation of the jobs on
each machine according to the densities of the jobs in $A'$.

Let $\bar{\zeta}$ be such that
$SOL_{\bar{\zeta}}$ is minimal. Then, $SOL_{\bar{\zeta}} \leq
{\del^{\ell+1}}SOL(A')$. We get $SOL(A_{\bar{\zeta}}) \leq
(1+\del^{\ell+1}(1+\del)^{\xi}) SOL(A')  \leq (1+2\del)SOL(A') $.

The second part will follow from the first one.  Let $SOL'=\cO'$. By the second claim of the first part, there exists a value $\zeta'$ such that $\cO'(A_{\zeta'}) \leq (1+2\del)\cO'(A')$. Since $\cO'$ and
$\cO^{\zeta'}$ are optimal solutions for $A'$ and $A_{\zeta'}$
respectively, we have $\cO'(A') \leq \cO^{\zeta'}(A')$ and
$\cO^{\zeta'}(A_{\zeta'}) \leq \cO'(A_{\zeta'})$.
Letting
$SOL''=\cO^{\zeta'}$ we get (using the first part of the first claim) $\cO^{\zeta'}(A') \leq
\cO^{\zeta'}(A_{\zeta'})$, which proves $\cO'(A') \leq \cO^{\zeta'}(A_{\zeta'})  \leq  (1+2\del) \cO'(A')$. We get $SOL_1(A') \leq SOL_1(A_{\zeta'})\leq (1+k'\del)O^{\zeta'}(A_{\zeta'}) \leq (1+2\del)(1+k'\del) \cO'(A') = (1+(k'+2)\del+2k'\del^2) \cO'(A')\leq (1+(2k'+2)\del) \cO'(A')$.
\end{proof}

\paragraph{The algorithm.}
In the next section (Section \ref{brsimp}) we present an algorithm that receives an input where the ratio between the maximum density of any job and the minimum density of any job is at most $(1+\del)^{{y}}$, where ${y} =\frac{\xi}{\del^{\ell+1}}-\xi-1$, and outputs a solution of cost at most $1+\del$ times the cost of an optimal solution for this input.   
We will use this algorithm as a black box in this section. For
every value of $\zeta$, we apply the following process and create
a schedule for the input $A_{\zeta}$. Afterwards, we choose a
solution of minimum cost among the $\frac{1}{\del^{\ell+1}}$
resulting solutions. Let $0 \leq \zeta \leq
\frac{1}{\del^{\ell+1}}-1$. Decompose the allowed densities for
$\zeta$ into collections of consecutive densities that are
separated by intervals of forbidden densities for $\zeta$ (two
densities are consecutive if the large one is larger by a factor
of exactly $1+\del$ from the smaller one). By our construction,
this results in subsets of allowed densities with very different
densities, such that each subset has allowed densities in an
interval of the form $[(1+\del)^{-{y}} \rho,\rho]$, where
$\rho=(1+\del)^{(v/\del^{\ell+1}+\zeta)\xi}$ for some integer $v$,
$y$ is as defined above (${y}=\frac{\xi}{\del^{\ell+1}}-\xi-1$)
and additionally there is a gap between the allowed densities of
one set and another set. More precisely, if $I$ and $I'$ are two
such subsets and the allowed densities for $I$ are smaller than
those of $I'$, then the largest allowed density for $I$ is smaller
by a factor of $(1+\del)^{\xi+1}\geq \frac{1+\del}{\del^{\ell}}$
than the smallest allowed density of $I'$. A sub-instance is
defined to be a non-empty subset of the jobs corresponding to an
interval of allowed densities together with the complete set of
machines. Let $q$ denote the number of such (sub-)instances (where
$q\leq n$). Let the instances be denoted by $I_1,\ldots,I_q$, such
that the maximum allowed density in $I_p$ is $\rho_p$, and for
$p>1$, $\rho_p>\rho_{p-1}$ and in fact $\rho_{p} \geq
(1+\del)^{{y}+\xi+1}\rho_{p-1} =
(1+\del)^{\frac{\xi}{\del^{\ell+1}}}\rho_{p-1} \geq
(\frac{1}{{\del^{\ell}}})^{\frac 1{\del^{\ell+1}}} \rho_{p-1}$.
Given a set of solutions $SOL_p$, for $1 \leq p \leq q$ (where
$SOL_p$ is a solution for $I_p$), define a combined solution
$SOL$, where the jobs assigned to machine $i$ in $SOL$ are all
jobs assigned to this machine in all the solutions. Obviously, in
$SOL$ the jobs are scheduled sorted by non-increasing indices of
their sets $I_p$.

\begin{lemma}
We have $OPT(A_{\zeta}) \geq \sum_{p=1}^q OPT(I_p)$.
\end{lemma}
\begin{proof}
Given a solution $S$ for the complete set of jobs $A_{\zeta}$,
define its pseudo-cost $S_{pc}$ as follows. The cost of each job
$j \in I_p$ is its weight multiplied by the following amount: the
total size of jobs of $I_p$ that are scheduled to run before $j$
on the same machine (i.e., out of jobs assigned to the same
machine, those are jobs of $I_p$ of strictly larger densities,
larger jobs of $I_p$ with the same density, and jobs of $I_p$ of
the same size and density, but of smaller indices) plus $p_j$, and
divided by the speed of the machine that runs $j$. We let
$OPT_{pc}$ denote the cost of an optimal solution with respect to
pseudo-cost. Obviously, for any solution $S$ and set of jobs $X$,
$S(X)\geq S_{pc}(X)$, and thus $OPT_{pc}(A_{\zeta})\leq
OPT(A_{\zeta})$. Since in an optimal solution for $A_{\zeta}$,
every subset of jobs $I_p$ is assigned independently of other such
subsets, we find $OPT_{pc}(A_{\zeta}) = \sum_{p=1}^q OPT(I_p)$.
\end{proof}

\begin{lemma}
If for any $1 \leq p \leq q$ it holds that $SOL_p(I_p) \leq (1+\nu) OPT(I_p)$ for some $\nu>0$, then the combined solution $SOL$ satisfies $SOL(A_{\zeta}) \leq (1+\nu)(1+8\del) OPT(A_{\zeta})$.
\end{lemma}
\begin{proof}
Let $1 \leq p\leq q$. For $0\leq i \leq {y}$, let $\Pi^i_p=\sum_{j \in I_p: w_j/p_j=(1+\del)^{-i}\rho_p}p_j$ (the total size of jobs of the instance $I_p$ with a given density) and $\Pi_p=\sum_{j \in I_p} p_j=\sum_{i=0}^{{y}} \Pi^i_p$ (the total size of the jobs of $I_p$).
We will prove for each machine separately that the cost for this machine does not exceed $(1+\nu)(1+8\del)$ times the total cost of the solutions $SOL_p$ ($1 \leq p \leq q$) for this machine. Moreover, it is sufficient to prove the property for a machine of speed $1$ (since for a given solution, the cost for a machine of speed $s$ is simply equal to the cost for a unit speed machine divided by $s$, and since the costs of different machines are independent). Therefore, we assume without loss of generality that there is just one machine, and its speed is $1$.
By definition, $SOL(A_{\zeta})=\sum_{p=1}^q SOL_p(I_p)+\sum_{p,p':1\leq p <p'\leq q} C_{p,p'}$, where $C_{p,p'} = \rho_p \cdot \Pi_{p'} \cdot \sum_{i=0}^{{y}} (1+\del)^{-i} \Pi^i_p$. Let $C_{p,p',i}=\rho_p \cdot \Pi_{p'}\Pi^i_p \cdot  (1+\del)^{-i}$ (and we have $C_{p,p'}=\sum_{i=0}^{{y}} C_{p,p',i}$).

By Lemma \ref{lboncost}, we find $$SOL_{p'}(I_{p'}) \geq (1+\del)^{-{y}} \rho_{p'} (\Pi_{p'})^2/2 = (1+\del)^{-{y}} \cdot \frac{\rho_{p'}}{\rho_p} \cdot \frac{\rho_p}2(\Pi_{p'})^2 \geq { { \del^{\frac{-\ell(p'-p-1)}{\del^{\ell+1}}-\ell}}} \cdot \frac{\rho_p}2 {(\Pi_{p'})^2} (1+\del) \ , $$ using $\rho_{p'} \geq (1+\del)^{{y}+\xi+1}\cdot (\frac{1}{{\del^{\ell}}})^{\frac {p'-p-1}{\del^{\ell+1}}} \rho_{p} \geq (1+\del)^{{y}+1}\cdot (\frac{1}{{\del^{\ell}}})^{\frac {p'-p-1}{\del^{\ell+1}}+1} \rho_{p}$.

Let $\iota$ be the set of indices $i$ such that $\Pi_{p}^i \leq \frac{\Pi_{p'}}{\del^{p'-p}}$, and let $\iota'$ be the complement set of indices (such that $\iota\cup \iota'=\{0,1,\ldots,{y}\}$). For any $i \in \iota$,
$$C_{p,p',i} =\rho_p \cdot \Pi_{p'}\Pi^i_p \cdot  (1+\del)^{-i} \leq \rho_p \cdot ( \Pi_{p'})^2 \cdot  (1+\del)^{-i} / \del^{p'-p} \leq 2(1+\del)^{-i-1} SOL_{p'}(I_{p'}){ { \del^{\frac{\ell(p'-p-1)}{\del^{\ell+1}}+\ell-p'+p}}} \ . $$

Thus, $$\sum_{i \in \iota} C_{p,p',i} \leq  2{ { \del^{(p'-p-1)(\frac{\ell}{\del^{\ell+1}}-1)+\ell-1}}} SOL_{p'}(I_{p'})\sum_{i=0}^{{y}} (1+\del)^{-i-1} \leq 2{ { \del^{(p'-p-1)(\frac{\ell}{\del^{\ell+1}}-1)+\ell-2}}} SOL_{p'}(I_{p'}) \  , $$ where the last inequality is since   $\sum_{i=0}^{{y}} (1+\del)^{-i-1} < \frac{1}{\del}$.
We show that we have
${\del^{(p'-p-1)(\frac{\ell}{\del^{\ell+1}}-1)+\ell-2}} \leq \del^{p'-p}$. The last inequality holds since
$(p'-p-1)(\frac{\ell}{\del^{\ell+1}}-1)+\ell-2 \geq (p'-p-1)(8\ell-1)+\ell-2=p'-p+(8\ell-2)(p'-p)-7\ell-1\geq p'-p+\ell-3 \geq p'-p$, as $p'\geq p+1$, $\del \leq \frac 18$, and $\ell \geq 3$.

Since any $i \in \iota'$ satisfies $\Pi_{p}^i > \frac{\Pi_{p'}}{\del^{p'-p}}$, we have $$\sum_{i \in \iota'} C_{p,p',i} \leq \sum_{i \in \iota'} \rho_p \cdot \Pi_{p'}\Pi^i_p \cdot  (1+\del)^{-i} \leq \del^{p'-p}\sum_{i=0}^{{y}} \rho_p \cdot (\Pi^i_p)^2 \cdot  (1+\del)^{-i} \leq 2\del^{p'-p} SOL_p(I_p) \ , $$ as $SOL_p(I_p) \geq \rho_p\sum_{i=0}^{{y}} (1+\del)^{-i}(\Pi_p^i)^2/2$ (this is the sum of $\Gamma$-values of running the jobs of each density separately).

We get that for any $1 \leq p < p' \leq q$, $C_{p,p'}\leq 2\del^{p'-p}(SOL_p(I_p)+SOL_{p'}(I_{p'})$. Moreover, for any $1 \leq p \leq q$, $\sum_{p'>p} \del^{p'-p} = \sum_{u=1}^{q-p} \del^u \leq \sum_{u=1}^{q-1} \del^u$, and $\sum_{p''<p} \del^{p-p''} = \sum_{u=1}^{p-1} \del^u \leq \sum_{u=1}^{q-1} \del^u$.
We find that $$SOL(A_{\zeta})- \sum_{p=1}^q SOL_p(I_p) \leq 2\sum_{1\leq p<p'\leq q} \del^{p'-p} (SOL_p(I_p)+SOL_{p'}(I_{p'})) \leq $$ $$ 4 \sum_{p=1}^{q} SOL_p(I_p) \sum_{u=1}^{q-1} \del^u \leq \frac{4\del}{1-\del}  \sum_{p=1}^q SOL_p(I_p) \leq 8\del  \sum_{p=1}^q SOL_p(I_p)\ . $$
We get $SOL(A_{\zeta}) \leq (1+8\del)\sum_{p=1}^q SOL_p(I_p) \leq (1+8\del)(1+\nu)\sum_{p=1}^q OPT(I_p)\leq (1+8\del)(1+\nu)OPT(A_{\zeta})$.
\end{proof}

In the next section (Section \ref{brsimp}) we design a $(1+\del)$-approximation algorithm
for the bounded ratio problem, that gives a $(1+\del)(1+8\del)$-approximation algorithm for $A^{\zeta}$. Since $(1+\del)(1+8\del)<1+10\del$, for an appropriate choice of $\zeta$, we get an $(1+22\del)$-approximation algorithm for $A'$, and a $(1+48\del)$-approximation algorithm for $A$. Thus, letting $\del=\frac{\eps}{48}$ will give a $(1+\eps)$-approximation
algorithm for the general problem. The algorithm of the next
section is applied at most $n$ times for each choice of $\zeta$, i.e., at most $\frac{n}{\del^6}$ times in total.

\subsection{An EPTAS for the bounded ratio problem}\label{brsimp}
Let $I$ be a bounded instance such that all densities are in
$[1,\rho=\rho(\delta)=(1+\del)^y]$, where $y$ is a function of
$\del$ (${y}$ and $\xi$ are as defined in the previous section). By scaling (and possibly increasing the interval of densities), any valid input can be transformed into this form.
We have the following properties. First, we have
$y=\frac{\xi}{\del^{\ell+1}}-\xi -1 \geq \frac{\xi}{\del^{\ell}}>
400$ (since $\frac{1}{\del^{\ell+1}}-\frac{1}{\del^{\ell}}-1 =
\frac{1-\del-\del^{\ell+1}}{\del^{\ell+1}}>1 > \frac 1{\xi} $ as
$2\del^{\ell+1}+\del<1$) and $y+2<y+\xi+1 \leq \frac{\xi}{\del^6} \leq
\frac{1}{\del^{9}}$. We also have $(1+\del)^{{y}} \geq
(1+\del)^{\frac{\xi}{\del^{\ell}}} \geq
\left( \frac{1}{\del^{\ell}}\right)^{\frac{1}{\del^{\ell}}} \geq
(\frac{1}{\del^3})^{\frac{1}{\del^3}} >  \frac{1}{\del^{1500}}$
and $(1+\del)^y \leq (1+\del)^{\xi/\del^{\ell+1}} \leq
(\frac{1}{\del^{\ell+1}})^{\frac 1{\del^{\ell+1}}} \leq
(\frac{1}{\del})^{6/\del^6} \leq (\frac{1}{\del})^{1/\del^{7}-12}$,
and $y+1 \leq \del^2(1+\del)^y$ since $y+1 \leq
\frac{1}{\del^{9}}$ while $(1+\del)^y \geq \frac{1}{\del^{1500}}$.

Since the density of every job is in $[1,(1+\del)^y]$, every job
$j\in I$ has $p_j=(1+\del)^{k}$, $w_j=(1+\del)^{k'}$, for some
integers $k,k'$ such that $0 \leq k'-k \leq y$. Let
$\gamma=\frac{\del^{12}}{(1+\del)^y}$. We have $\gamma \geq
\del^{1/\del^7}$ and $\gamma \leq \del^{1512}$. Let $\I$ denote
the set of values $i$ such that there is at least one job of size
$(1+\del)^i$. Clearly, $|\I| \leq n$. Let $n_{r,i}$ denote the
number of jobs of size $(1+\del)^i$ and density $(1+\del)^r$. By
scaling the machine speeds (that are integer powers of $1+\del$)
and possibly reordering the machines, let the speeds of machines
be $s_1 \geq s_2 \cdots \geq s_m$, where without loss of
generality, $s_1=1$, and for $i\geq 2$, $s_i=(1+\del)^{k_i}$, for
some non-positive integer $k_i\leq 0$.

Recall that the work of a machine is the total size of jobs
assigned to it. We would like to assume that job sizes are scaled
such that the work of the most loaded machine (the most loaded
machine in terms of work) of speed $1$ is in $[1/(1+\del),1)$. We
will now show that this is possible. For a job $j$, and $1 \leq b
\leq \frac n{\del}$, let $D_{j,b}$ be the interval
$[(1+\del)^{b-1}p_j,(1+\del)^b p_j)$, where
$(1+\del)^{\frac{n}{\del}} \geq n+1$.

\begin{claim}
For any solution, there exist a job $1\leq j' \leq n$ and an integer $1 \leq b' \leq \frac{n}{\del}$, such that the work of the most loaded machine of speed $1$ is in $D_{j,b}$.
\end{claim}
\begin{proof}
Consider the most loaded machine of speed $1$ (breaking ties arbitrarily). Let $j'$ be a largest job assigned to this machine. The work of the machine is at least $p_{j'}$. Since there are at most $n$ jobs assigned to this machine, each of size at most $p_{j'}$, its work is at most $n\cdot p_{j'}$. We have $[p_{j'},n\cdot p_{j'}]\subseteq [p_{j'},(n+1)p_{j'})\subseteq \bigcup_{1 \leq b \leq \frac{n}{\del}} D_{j',b}$, thus, there exists a value $b'$ as required.
\end{proof}

For every choice of a pair $j,b$ ($\frac{n^2}{\del}$ choices in
total), we scale (i.e., divide) job sizes by $(1+\del)^b p_j$
(which is an integer power of $1+\del$), and apply the algorithm
described later in this section (this algorithm enforces the
existence of a fastest machine with a suitable load). By the
claim, given an optimal solution $OPT$, for at least one choice of
$j,b$ the assumption regarding the work of the most loaded machine
of speed $1$ holds as a result of the scaling. We will pick the
best solution, whose objective function value cannot be larger
than that of the solution obtained for the correct choice of
$j,b$. In what follows we analyze the properties of the correct
choice of $j,b$ for a given optimal solution $OPT$ after the
scaling. In what follows the job sizes are according to the
scaling.

Let $U$ be a threshold on job sizes. Let $\hat{I}$ be the subset
of jobs assigned to machine $i$ in some solution. The $U$-cost of
machine $i$ in this solution is defined as the total weighted
completion time of jobs whose sizes are at least $U$, plus the
$\Gamma$-values of jobs whose sizes are below $U$. The cost for
machine $i$ is therefore its $U$-cost plus $\frac{1}{2s_i}\sum_{j
\in \hat{I}, p_j <U } w_j\cdot p_j$. Obviously, the $U$-cost never
exceeds the cost, for any value of $U$. Additionally, similarly to
the cost, the $U$-cost for a fixed value of $U$ is monotone in the
sense that removing a job from the machine decreases its $U$-cost,
and thus, if we consider a block of jobs, decreasing the total
size of jobs in a block also decreases the $U$-cost.

We use the following functions of $\del$: $f(\del)$ and $g(\del)$,
both integral, non-decreasing, and negative. A machine $i$ is
called {\it slow} if its speed is at most $(1+\del)^{f(\del)}$
(and otherwise we say that it is not slow or that it is fast). We
say that a machine is lightly loaded if its work does not exceed
$(1+\del)^{g(\del)}$ (and otherwise it is heavily loaded).  Let
$f(\del)=g(\del)-\frac{2}{{\del}^2}$ and
$g(\del)=-(\frac{1}{\del})^{\frac{1}{\del^{300}}}$. Consider a
fixed optimal solution $OPT$ (for the input considered in this
section, given the scaling of machine speeds and job sizes).

\begin{lemma}
No machine has load strictly above $\frac{2}{\del}$ in $OPT$, and this is an upper bound on the work of any machine as well.
\end{lemma}
\begin{proof}
Assume by contradiction that machine $i$ has a job $j$ of
completion time strictly above $\frac{2}{\del}$. Since machines of
speed $1$ have loads of at most $1$, we have $s_i \leq
\frac{1}{1+\del}$ (as $i$ cannot have speed $1$). We move $j$ to
run last on a machine of speed $1$, and compare its previous
completion time with its new completion time. If $p_j \leq  s_i
\cdot \frac{2}{\del}$, then running job $j$ on a machine of speed
$1$ would result in a completion time of at most $1+p_j \leq
1+\frac 2{\del(1+\del)}< \frac 2{\del}$, which holds since it is
equivalent to $2+\del+\del^2 < 2+2\del$ holding for all $\del<1$.
 Thus in this case $j$ has a smaller completion time, contradicting the optimality of the original solution.
If $p_j > s_i \cdot \frac{2}{\del}$, then the load of machine $i$ is at least $\frac{p_j}{s_i}$, and $1+p_j < \frac{p_j}{s_i}$ holds for any $\del<1$. The last inequality holds since it is equivalent to $p_j>\frac{s_i}{1-s_i}$, and thus for proving it, it is sufficient to show $s_i \cdot \frac{2}{\del} \geq \frac{s_i}{1-s_i}$, or equivalently, $s_i \leq 1-\frac{\del}{2}$, which holds for any $\del<1$ since $s_i \leq \frac{1}{1+\del}$.
Since no speed exceeds $1$, the work of a machine does not exceed its load, and thus every machine has a work of at most $\frac{2}{\del}$ as well.
\end{proof}

\begin{corollary}
In $OPT$ every slow machine is
lightly loaded.
\end{corollary}
\begin{proof}
Any machine has load of at most $\frac{2}{\del}$, and thus work of
at most $s_i \cdot \frac 2{\del}$. A slow machine $i$ has work of
at most $s_i \cdot \frac {2}{\del} \leq  \frac
{2(1+\del)^{f(\del)}}{\del} =  \frac
{2(1+\del)^{g(\del)-2/{\del^2}}}{\del}
    \leq (1+\del)^{g(\del)}$, since $(1+\del)^{2/{\del^2}}>\frac{2}{\del}$.
\end{proof}

Next, we define machine configurations. A configuration is a
vector that defines the schedule of one machine (a set of jobs
assigned to it in terms of the types of jobs, that is, a job is
specified by its size and weight but not by its identity), and the
set of all configurations will be denoted by $\C$. For a
configuration $C \in \C$, the first component is an integer
$j_1(C) \in \mathbb{Z}$, such that the total work of the machine
is in $((1+\del)^{j_1(C)-2},(1+\del)^{j_1(C)}]$.  The second
component is a non-positive integer $j_2(C)$ such that the speed
of the machine is $(1+\del)^{j_2(C)}$. The third component is an
integer $j_3(C) \in \mathbb{Z}$, such that $j_3(C)=j_1(C)-j_2(C)$,
and therefore the completion time (or load) of the machine is in
$((1+\del)^{j_3(C)-2},(1+\del)^{j_3(C)}]$. Recall that
$\gamma=\frac{\del^{12}}{(1+\del)^y}$. For $0 \leq r \leq y$, and
$i \leq j_1(C)$ such that ${\gamma}(1+\del)^{j_1(C)-1} \leq
(1+\del)^i \leq (1+\del)^{j_1(C)}$ (i.e., $j_1(C)-1+\lceil
\log_{1+\del}\gamma \rceil\leq i \leq j_1(C)$), there is an
integer component $n_{r,i}(C) \geq 0$ stating how many jobs of
size $(1+\del)^i$ and density $(1+\del)^r$ are assigned to this
machine. These jobs are called {\it large} (large for
configuration $C$, or large for a machine that has configuration
$C$). We let $n_{r,i}(C)=0$ for other (smaller or larger) values
of $i$ (these are constants that are not a part of the
configuration). There are additional components for other jobs
assigned to a machine whose configuration is $C$. These jobs
(which are not taken into account in the components of the form
$n_{r,i}(C)$) must have smaller values of $i$, as jobs with larger
values of $i$  (jobs of sizes above $(1+\del)^{j_1(C)}$) cannot be
assigned to a machine that has configuration $C$ (since they are
too large, given the upper bound on the work of the machine).
These remaining jobs are called small jobs for $C$, or small jobs
for a machine whose schedule is according to configuration $C$.
For every $r$, there is an integer component $t_r(C) \geq 0$ such
that the total size of small jobs (of sizes in
$(0,{\gamma}(1+\del)^{j_1(C)-1})$) of density $(1+\del)^r$
assigned to a machine with configuration $C$ is in
$((t_r(C)-1){\gamma}(1+\del)^{j_1(C)-1},t_r(C){\gamma}(1+\del)^{j_1(C)-1}
]$.  We have $t_r(C)=0$ if and only if there are no such jobs in a
machine with this configuration. Recall that $\I$ is the set of
indices $i$ such that there is at least one job of size
$(1+\del)^i$.  A configuration $C$ is valid if the total size of
jobs is sufficiently close to its required work, and its load does
not exceed $\frac{2}{\del}$. Formally, letting $\I(C)=\{i \in \I :
j_1(C)-1+\lceil \log_{1+\del}\gamma \rceil\leq i \leq j_1(C)\}$ it
is required that $\sum_{0 \leq r \leq y, i \in \I(C)} (1+\del)^i
n_{r,i}(C)+\sum_{r=0}^y t_r(C)({\gamma}(1+\del)^{j_1(C)-1}) \in
((1+\del)^{j_1(C)-2},(1+\del)^{j_1(C)}]$. We also let $\I'(C)=\{i
\in \I : i \leq j_1(C)-2+\lceil \log_{1+\del}\gamma \rceil\}$, and
require $(1+\del)^{j_3(C)-2}<\frac{2}{\del}$. In the remainder of
this section, given $C$, we will use
$U(C)={\gamma}(1+\del)^{j_1(C)-1}$ for computing the $U$-cost of a
machine with configuration $C$.

The number of components is therefore at most
$3+(y+1)(\log_{1+\del}\frac{1}{{\gamma}}+3)$. Since
$\log_{1+\del}\frac{1}{\gamma}=y+\log_{1+\del} \frac
{1}{\del^{12}} \leq y+\frac{1}{\del^{13}}$, and $y+1 \leq
\frac{1}{\del^9}$, we have $3+(y+1)(\log_{1+\del}\frac{1}{\gamma}
+3) \leq
3+\frac{1}{\del^9}(\frac{1}{\del^{9}}+\frac{1}{\del^{13}}+3)<\frac{1}{\del^{23}}$.
Consider a fixed pair of values $j_1$ and $j_2$, and let
$j_3=j_1-j_2$. Consider the set of configurations $C$ such that
$j_1(C)=j_1$ and $j_2(C)=j_2$ (and $j_3(C)=j_3$). For any $0 \leq
r \leq y$,  $t_r(C) \geq 0$ is an integer such that
$t_r(C){\gamma}(1+\del)^{j_1-1} \leq (1+\del)^{j_1}$, i.e.,
$t_r(C) \leq
\frac{1+\del}{{\gamma}}=\frac{(1+\del)^{y+1}}{\del^{12}}$. Since
the total size of large items is at most $(1+\del)^{j_1}$ while
the size of such an item is at least ${\gamma}(1+\del)^{j_1-1}$,
we find $0 \leq n_{r,i}(C) \leq \frac{1+\del}{{\gamma}}$. We get
$\frac{1+\del}{\gamma}+1 = \frac{(1+\del)^{y+1}}{\del^{12}}+1 <
\frac{(1+\del)^{y+2}}{\del^{12}}$ (since $0<\del\leq \frac 18$),
fixing $j_1$ and $j_2$, and using $y+2 \leq \frac{1}{\del^9}$ and
$1+\del <\frac 1{\del}$, the number of different configurations
(with given values $j_1,j_2,j_3$) is at most
$(\frac{(1+\del)^{\frac{1}{\del^{9}}}}{\del^{12}})^{\frac{1}{\del^{23}}}<(\frac{1}{\del})^{1/\del^{230}}$,
which is a constant (a function of $\del$).

Next, we find the number of pairs $j_1,j_2$. The values $j_2$
correspond to actual speeds of machines. Thus, the number of
possible values for $j_2$ is obviously at most $m$, which ensures
that there are configurations for each possible speed. The values
$j_1$ are defined as follows. For every subset of jobs $\bar{I}$,
let $\psi_{\bar{I}}=\lceil \log_{1+\del} \sum_{j \in \bar{I}}
p_j\rceil$. The values $\psi_{\bar{I}}$ act as first components of configurations (and there are no other
options). We bound the number of intervals of the form
$I_{\jmath}=((1+\del)^{\jmath-1},(1+\del)^{\jmath}]$ such that the
actual work of a machine (which must be the total size of a subset
of jobs) can belong to such an interval. The work of a machine
whose largest job is $j$ is in $[p_j,n\cdot p_j]$. Every interval
$[p_j,n\cdot p_j]$ overlaps with at most $1+\log_{1+\del} n$
intervals of the form $I_{\jmath}$. Thus, the total number of
values that $j_1$ can have is at most
$n(1+\frac{n}{\del})=O(\frac{n^2}{\del})$, and we moreover only
allow such values that $(1+\del)^{j_3-2}<\frac{2}{\del}$. We find
that the number of options for $j_1,j_2,j_3$ is
$O(\frac{n^2m}{\del})$. Let
$J=\{\sigma_1,\sigma_2,\ldots,\sigma_{\kappa}\}$ be the set of all
$\kappa$ different speeds such that $1=\sigma_1 > \sigma_2 >
\cdots > \sigma_{\kappa}$, and let $N_{i}$ be the number of
machines of speed $\sigma_i$.

Consider the following mathematical program $\Pi$. The goal of
$\Pi$ is to determine a partial schedule via machine
configurations. For every configuration $C$, $X_C$ is a variable
stating how many machines have this configuration. Obviously, the
number of used configurations whose second component is $\sigma_i$
cannot exceed $N_i$. For every triple of density, size, and a
configuration (density $(1+\del)^r$, size $(1+\del)^i$, and
configuration $C$), there is a variable $Y_{r,i,C}$ corresponding
to the  number of jobs of this size and density that are assigned
to machines whose configurations are $C$, as small jobs for this
configuration. The variable $Y_{r,i,C}$ may be positive only if a
job of size $(1+\del)^i$ is small for configuration $C$, i.e., $i
\leq j_1(C)-2+\lceil \log_{1+\del}\gamma \rceil$, and in all other
cases we set $Y_{r,i,C}=0$ to be a constant rather than a
variable. A configuration $C$ has a cost denoted by $cost(C)$
associated with it, which is the $U$-cost for
$U(C)={\gamma}(1+\del)^{j_1(C)-1}$. Recall that for the purpose of
calculating the $U$-cost of a machine, a list of its large jobs is
needed, but for its jobs of size below $U$ (i.e., small jobs), the
only needed property of the subset of small jobs (for
configuration $C$) of each one of the densities is their total
size, and the exact list of such small jobs is not needed. For
large jobs, the exact identities of jobs are not needed as well,
but a list of densities and sizes is needed. Thus, the $U$-cost
for configuration $C$ is calculated assuming that the machine has
exactly $n_{r,i}(C)$ jobs of size $(1+\del)^i$ and weight
$(1+\del)^{i+r}$ (i.e., density $(1+\del)^r$) assigned to it, and
the total size of small jobs (of sizes in
$(0,{\gamma}(1+\del)^{j_1(C)-1})$)  with densities equal to
$(1+\del)^r$ is exactly $t_r(C)\cdot{\gamma}(1+\del)^{j_1(C)-1}$.
The objective of $\Pi$ is to minimize the sum of costs of
configurations plus the missing parts of the costs of jobs that
are assigned as small. For that, for a job of size $(1+\del)^i$
and density $(1+\del)^r$ that is assigned to a machine of speed
$s$ as a small job for its configuration, an additive term of
$\frac{(1+\del)^{r+2i}}{2s}$ is incurred (this is the difference
between the actual cost for this job, and its $\Gamma$-value, which is the part of the cost
already included in the $U$-cost). This last term only depends on
the speed of the machine that runs the job rather than the
specific machine. Thus, for each $r$ and $i$ we add
$(1+\del)^{r+2i} \sum_{C \in \C}
\frac{Y_{r,i,C}}{2(1+\del)^{j_2(C)}}$ to the cost of
configurations to get the total cost of the schedule, where
$\sum_{C \in \C} Y_{r,i,C}$ is the number of jobs of size
$(1+\del)^i$ and density $(1+\del)^r$ that are assigned as small
jobs for their configurations.

Condition (\ref{2nd}) ensures that the number of used machines for each speed does not exceed the existing number of such machines.
Condition (\ref{3rd}) states that every job is assigned (either it is a large job of some machine or a small job) and condition (\ref{4th}) considers jobs of density $(1+\del)^r$ that are assigned as small to machines scheduled according to configuration $C$, and verifies that sufficient space is allocated for them if the space for them is slightly extended. Condition (\ref{5th}) ensures that indeed there is a machine of the maximum speed that has a work that is close to $1$, that is, above $\frac{1}{(1+\del)^3}$ and at most $1$ (the condition that the work is in $[1/(1+\del),1)$ is slightly relaxed).

\begin{eqnarray}
\nonumber && \min \sum_{C\in {\C}} cost(C)X_C +\sum_{r=0}^y\sum_{i \in \I}(1+\del)^{r+2i} \sum_{C \in \C} \frac{Y_{r,i,C}}{2(1+\del)^{j_2(C)}}  {\mbox  \ \  \ \ \  \ \ \  \ \      s.t. \  \ \ \  \ \ \  \ \ \  \ \ \  \ \ }\\ \nonumber \\
\label{2nd} && \sum_{C \in {\C} : (1+\del)^{j_2(C)}=\sigma_i} X_C \leq N_i  \ \ \ \ \ \ \ \ \ \ \ \ \ \ \ \ \ \ \ \ \ \ \ \ \ \ \ \ \ \  \ \ \ \ \ \ \ \ \ \ \ \ \ \ \ \ \  \ \ \ \ \ \ \ \ \ \ \ \ \ \ \ \ \ \forall \sigma_i\in J
\\ \nonumber \\ && \label{3rd} \sum_{C\in\C}n_{r,i}(C)X_C+\sum_{C \in \C} Y_{r,i,C}=n_{r,i} {\mbox \  \  \ \  \ \ \  \ \ \  \ \ \  \ \ \  \ \ \  \ \  \ \  \ \  \ \ \  \ \ \  \ \ \  \ \ \  \ \  \  \ \ \  \ \ \  \ \ \  \ \ } \forall 0\leq r \leq y,i\in {\I}
\\ \nonumber \\
&& \label{4th} \sum_{i\in \I'(C)} (1+\del)^i  Y_{r,i,C} \leq
(t_r(C)+1){\gamma}(1+\del)^{j_1(C)-1} X_C   {\mbox\     \ \ \  \ \
\  \ \ \  \ \ } \forall 0\leq r \leq y, C \in {\C}
\\
&& \label{5th} \sum_{C \in \C : j_1(C) \in \{-1,0\}, j_2(C)=0} X_C \geq 1
\\ \nonumber \\
&& \label{6th} Y_{r,i,C}\geq 0  {\mbox \  \ \ \  \ \ \  \ \ \  \ \ \  \ \ \  \ \ } \forall C \in \C, 0\leq r \leq y,i\in \I'(C)\\
&& \label{7th} X_C \geq 0 {\mbox \  \ \ \  \ \ \  \ \ \  \ \ \  \ \ \  \ \ \  \ \ } \forall C \in {\C}
\end{eqnarray}

Observe that we distinguish between large and small jobs for a machine based on the configuration and not solely by its speed.  This feature is needed since our problem does not allow the use of the dual-approximation method.
We
define the set of heavy configurations $\C_H$ as
$\C_H=\{C \in \C : j_1(C) > g(\del\}$, (note that the
definition of a heavy configuration or a heavily loaded machine depends only on the work,
and it is independent of the speed) and the
complement set of light configurations is $\C_L =\C
\setminus \C_H$. We see $\Pi$ as a mixed-integer linear program. All {\it variables} $Y_{r,i,C}$
may be fractional.  The variables of configurations corresponding
to slow machines and variables of light configurations of fast machines may be fractional,
whereas the variables of configurations corresponding to fast machines
and heavy configurations must be integral.  Recall that for every pair of speed and work
(i.e, a pair $j_1,j_2$), the number of different valid
configurations is constant (as a function of $\del$). The number
of different speeds of fast machines is at most $(-f(\del))$. Recall that the
work of a heavily loaded machine of speed $s$ is in $((1+\del)^{g(\del)},\frac {2s}{\del}]$. For a fixed slow speed $s\leq (1+\del)^{f(\del)}$, consider configuration $C$ that satisfies $j_2(C) =\log_{1+\del}s \leq f(\del)$.
By the condition on $j_3(C)$, the configuration $C$ has $j_1(C)-j_2(C)-1=j_3(C)-1<\log_{1+\del}{\frac{2}{\del}}+1<\frac{2}{\del^2}$, as $(1+\del)^{2/\del^2-1}>1+\frac 2{\del}-\del>\frac 2{\del}$. If $C$ is heavy, then $j_1(C) \geq g(\del)+1=f(\del)+\frac{2}{{\del}^2}+1$, and since $j_2(C) \leq f(\del)$, we reach a contradiction. Thus, any configuration for a slow machine is light. The number of different values of $j_1(C)$ such that $C$ is a heavy configuration is at most $-g(\del)+\log_{1+\del}\frac{2}{\del}+3 \leq \frac{2}{\del^2}+2-g(\del)=-f(\del)+2$, and $j_2(C)$ of a heavy configuration $C$ must be a fast speed (so there are at most $-f(\del)$ values for it).  As the number of integral
variables is constant (as a function of $\del$), an optimal solution can be found in
polynomial time \cite{Len83,Kan83}. We will first compare the cost of the optimal schedule $OPT$ to the
cost of an optimal solution of $\Pi$, and then we will show how to
obtain an actual schedule given a solution of $\Pi$, such the cost of the schedule is larger only by a factor of at most $1+\del$ than the objective function value of the solution to $\Pi$. Let
$(X^*,Y^*)$ denote an optimal solution to the mixed-integer linear
program, and let $Z^*$ be its objective value. Let $Z^{OPT}$
denote the cost of $OPT$.

\begin{theorem}
If $(X^*,Y^*)$ is an optimal solution of $\Pi$, then $Z^* \leq Z^{OPT}$.
\end{theorem}
\begin{proof}
We define an integral solution for $\Pi$ that is based on $OPT$.
Given machine $1 \leq \lambda \leq m$, we define its configuration
$C$ as follows. Given the work $W$ of the machine, let
$j_1=j_1(C)=\lceil \log_{1+\del} W \rceil$. Given the speed
$s_{\lambda}$, let $j_2=j_2(C)=\log_{1+\del} s_{\lambda}$. Let
$j_3=j_3(C)=j_1-j_2$. Recall that for every possible subset of
input jobs whose total size does not exceed
$\frac{2s_{\lambda}}{\del}$, the integer $j'$, such that the
interval $((1+\del)^{j'-1},(1+\del)^{j'}]$ contains their total
size, is one of the options for the first component of
configurations unless
$(1+\del)^{j_3-2} \geq \frac{2}{\del}$, which does not hold since
in $OPT$ we have $W \leq
\frac{2s_{\lambda}}{\del}$ and $(1+\del)^{j_3-2} \geq \frac{2}{\del}$ implies $(1+\del)^{j_1-2} \geq \frac{2s_{\lambda}}{\del} \geq W > (1+\del)^{j_1-1}$, a contradiction. Thus, the value $j_1(C)$ as defined
above is a valid first component for a configuration, and the
second component $j_2(C)$ is valid since $(1+\del)^{j_2(C)} \in
J$. Let $\I(C)$ be defined as above ($\I(C)$ contains all $i$ such
that $j_1(C)-1+\lceil \log_{1+\del}\gamma \rceil\leq i \leq
j_1(C)$). For any $i \in \I(C)$ and $0\leq r \leq y$, the value
$n_{r,i}(C)$ is defined to be the exact number of jobs of density
$(1+\del)^r$ and size $(1+\del)^i$ assigned to $\lambda$ in $OPT$. The
remaining jobs assigned to $\lambda$ are obviously small (as there
cannot be a job of size above $W$ assigned to $\lambda$, and $(1+\del)^{j_1}\geq W$). The
total size of the remaining jobs assigned to $\lambda$ is
$W-\sum_{0 \leq r \leq y, i \in \I(C)} (1+\del)^i n_{r,i}(C) \leq
(1+\del)^{j_1(C)}$, and the size of each such job is below
${\gamma}(1+\del)^{j_1(C)-1}<{\gamma}W$. For every $0\leq r \leq
y$, let $W_r$ denote the total size of (small) jobs whose density
is $(1+\del)^r$ and sizes below ${\gamma}(1+\del)^{j_1(C)-1}$ that are assigned to $\lambda$. Let
$t'_r(\lambda)=\frac{W_r}{{\gamma}(1+\del)^{j_1(C)-1}}$ and
$t_r(C)=\lfloor t'_r(\lambda)\rfloor$ (and therefore  $t_r(C) \leq t'_r(\lambda) <
t_r(C)+1$). The values $n_{r,i}$ and $t_r(C)$ are non-negative
integers that do not exceed
$\frac{W}{{\gamma}(1+\del)^{j_1(C)-1}}\leq
\frac{1+\del}{{\gamma}}$. Thus, the components of the configuration vector are such that if $C$ it is valid with respect to the approximate total size of items, then $C \in \C$ .
We now show that the resulting
configuration is indeed a valid configuration, by calculating
$\sum_{0 \leq r \leq y, i \in \I(C)} (1+\del)^i
n_{r,i}(C)+\sum_{r=0}^y t_r(C)({\gamma}(1+\del)^{j_1(C)-1})$. We
have $$\sum_{0 \leq r \leq y, i \in \I(C)} (1+\del)^i
n_{r,i}(C)+\sum_{r=0}^y t'_r(\lambda)({\gamma}(1+\del)^{j_1(C)-1})=W \ ,
$$ and $0 \leq \sum_{r=0}^y
(t'_r(\lambda)-t_r(C))({\gamma}(1+\del)^{j_1(C)-1})\leq
(y+1)({\gamma}(1+\del)^{j_1(C)-1})$. Using $y+1 \leq
\frac{1}{\del^9}$, $\gamma \leq \del^{1512}$, the last expression is at most
${\del^{1503}}(1+\del)^{j_1(C)-1}$. We get $$(1+\del)^{j_1(C)}
\geq W \geq \sum_{0 \leq r \leq y, i \in \I(C)} (1+\del)^i
n_{r,i}(C)+\sum_{r=0}^y t_r(C)({\gamma}(1+\del)^{j_1(C)-1}) $$
$$\geq (1+\del)^{j_1(C)-1}(1-\del^{1503})> (1+\del)^{j_1(C)-2}
\ ,
$$ since $(1-\del^{1503})(1+\del) > 1$ for $\del \leq \frac 18$.
After defining the configuration of each machine, we count the number of machines with each configuration $C'$, and
let the variable $X_{C'}$ be the number of machines with configuration $C'$.
The
values $Y_{r,i,C}$ are simply defined by counting the numbers of corresponding jobs.
Moreover, the objective function value is no larger than the cost of the
schedule since the blocks of small jobs may have smaller sizes in the configurations compared to the original schedule (as the values $t_r(C)$ are rounded down versions of the values $t'_r(\lambda)$), and the total sizes of jobs in these blocks were rounded down in the construction of the corresponding configurations (from $t'_r(\lambda)$ to $t_r(C)$ for density $r$ and configuration $C$).
For the machine of speed $1$ whose work is in $[\frac
1{1+\del},1)$, its configuration $\hat{C}$ has $j_1({\hat{C}})=0$ or $j_1({\hat{C}})=-1$,
and $j_2({\hat{C}})=0$, so $x_{\hat{C}}\geq 1$, and condition
(\ref{5th}) is satisfied since all variables $X_{C'}$ are
non-negative.  Thus, for each speed $s=\sigma_i$, the number of
configurations $C'$ where $j_2(C')=\log_{1+\del} s$ is exactly
$N_i$, and condition (\ref{2nd}) is satisfied. Condition
(\ref{3rd}) is satisfied since each job is either counted in some $Y_{r,i,C}$ if it is small for its configuration, or it is counted since for its configuration $C$, $n_{r,i}(C) \cdot X_C$ is exactly the number of jobs of size $(1+\del)^i$ and density $(1+\del)^r$ that are scheduled as large jobs on machines whose configuration is $C$.
Finally, condition (\ref{4th}) is satisfied
since the total size of jobs of density $(1+\del)^r$ assigned as
small jobs to machine $\lambda$ whose configuration is $C$ is exactly
$t'_r(\lambda)  \gamma  (1+\del)^{j_1(C)-1} $, and $t'_r(\lambda) \leq
t_r(C)+1$, and thus it is at most $(t_r(C)+1)\cdot \gamma(1+\del)^{j_1(C)-1}$.  Summing up the last inequality over all machines with configuration $C$ results in the bound of condition (\ref{4th}).
\end{proof}

\begin{theorem}
There exists a schedule whose cost is at most $(1+\del)Z^*$, and such a schedule can be constructed in polynomial time from $(X^*,Y^*)$.
\end{theorem}
\begin{proof}
We start with constructing an alternative set of values of the variables and bounding the resulting cost from above. We will later convert this alternative solution (that may violate some of the constraints of $\Pi$) into a schedule.
For every $C \in \C$, let $X'_C=\lfloor X^*_C \rfloor$. If $X^*_C=0$, then  we set $X'_C=0$. In this case $Y^*_{r,i,C}=0$ must hold for all $r,i$ by (\ref{4th}), and we set $Y'_{r,i,C}=0$. If $X^*>0$, we consider two cases. If $C \in
C_H$, then the variables of configurations are integral, and therefore $X'_C=X^*_C$. In this case we set
$Y'_{r,i,C}=\lceil Y^*_{r,i,C} \rceil$ for any $0 \leq r \leq y$,
$i \in \I'(C)$. Otherwise, in the case $C \in C_L$, let
$Y'_{r,i,C}=\left\lfloor \frac{X'_C}{X^*_C} \cdot Y^*_{r,i,C} \right\rfloor$.
Using these definitions, in the case $Y'_{r,i,C}>0$, if $C \in C_H$, then we have $Y'_{r,i,C} \leq
Y^*_{r,i,C}+1$, and otherwise we have $Y'_{r,i,C} \leq \frac{X'_C}{X_C^*}Y^*_{r,i,C}$.
Moreover, if $C\in C_L$, then we have $Y'_{r,i,C} >
\frac{X'_C}{X^*_C} \cdot Y^*_{r,i,C} -1 >
Y^*_{r,i,c}-\frac{Y^*_{r,i,C}}{X^*_C}-1$, by $X^*_C < X'_C+1$ (that holds even if $X^*_C=0$). Using condition (\ref{4th}),
for every $C \in C_L$ such that $X^*_C>0$ and $0 \leq r \leq y$, we have $$\sum_{i\in \I'(C)} (1+\del)^i
Y'_{r,i,C} \leq \sum_{i\in \I'(C)} (1+\del)^i \frac{X'_C}{X^*_C}
Y^*_{r,i,C} \leq \frac{X'_C}{X^*_C}
(t_r(C)+1){\gamma}(1+\del)^{j_1(C)-1} X^*_C
$$ $$=(t_r(C)+1){\gamma}(1+\del)^{j_1(C)-1} X'_C \ . $$ Similarly, for
every $C \in C_H$  such that $X^*_C>0$ and $0\leq r \leq y$, we have $$\sum_{i\in \I'(C)} (1+\del)^i
Y'_{r,i,C} \leq \sum_{i\in \I'(C)} (1+\del)^i (Y^*_{r,i,C}+1) \leq
(t_r(C)+1){\gamma}(1+\del)^{j_1(C)-1} X^*_C+\sum_{i\in \I'(C)}
(1+\del)^i $$ $$\leq
(t_r(C)+2+\frac{2}{\del}){\gamma}(1+\del)^{j_1(C)-1} X'_C$$ as
$X'_C=X^*_C \geq 1$, and $\sum_{i\in \I'(C)} (1+\del)^i \leq
{\gamma}(1+\del)^{j_1(C)-1}\sum_{i=0}^{\infty}
\frac{1}{(1+\del)^i}={\gamma}(1+\del)^{j_1(C)-1} \cdot
\frac{1+\del}{\del}$. If $X^*_C=0$, then the corresponding condition (the first one if $C \in C_L$, and the second one if $C \in C_H$) hold trivially, as all variables are equal to zero.

Since $X'_C \leq X^*_C$ for $C \in \C$, and $Y'_{r,i,C} \leq
Y^*_{r,i,C}$ for all $C\in C_L$, $0 \leq r \leq y$, and $i\in I'(C)$, the objective function value
for these variables is at most $Z^*$ plus $$\sum_{r=0}^y\sum_{i
\in \I}(1+\del)^{r+2i} \sum_{C \in C_H: i \in \I'(C)}
\frac{Y'_{r,i,C}-Y_{r,i,C}^*}{2(1+\del)^{j_2(C)}} \leq \sum_{r=0}^y\sum_{C \in C_H}\sum_{i
\in \I'(C)}(1+\del)^{r+2i}
\frac{X^*_C}{2(1+\del)^{j_2(C)}} \ , $$ since for $C \in C_H$ it holds that
$Y'_{r,i,C}=0$ if $X^*_C=0$, and otherwise $Y'_{r,i,C}
-Y^*_{r,i,C}=1 \leq X^*_C$. Next, we show that $\sum_{r=0}^y\sum_{i \in
\I'(C)}(1+\del)^{r+2i}  \frac{1}{2(1+\del)^{j_2(C)}} \leq \del^{100}
cost(C)$ for any $C \in C_H$, and thus we will conclude that the objective function value of the set of
variables $X'_C$, $Y'_{r,i,C}$ is at most $(1+\del^{100})Z^*$. Note that the solution that we consider is not a feasible solution for $\Pi$, as (for example) condition (\ref{3rd}) does not necessarily hold as it is possible that
 $\sum_{C\in\C}n_{r,i}(C)X'_C+\sum_{C \in \C} Y'_{r,i,C} \neq n_{r,i}$, still we can bound its objective function value using the objective function value of the optimal (and feasible) solution. We have $cost(C) \geq \frac{(1+\del)^{2j_1(C)-4}}{2(1+\del)^{j_2(C)}}$ (by Lemma \ref{lboncost}, since $cost(C)$ is computed for jobs of total size above $(1+\del)^{j_1(C)-2}$, and the jobs densities are no smaller than $1$).
Let $i'=\max \I'(C)$. We have $(1+\del)^{2i'} \leq (\gamma
(1+\del)^{j_1(C)-1})^2$ and  $\sum_{i \leq i'} (1+\del)^{2i} <
(1+\del)^{2i'+2}/\del$. Additionally,  $\sum_{r=0}^{y} (1+\del)^{r}
< (1+\del)^{y+1}/\del$. For a given $C \in C_H$, we get $$\sum_{r=0}^y\sum_{i \in
\I'(C)}(1+\del)^{r+2i}  \frac{1}{2(1+\del)^{j_2(C)}} \leq
\frac{(1+\del)^{y+2i'+3}}  {2\del^2(1+\del)^{j_2(C)}} \leq
\frac{(1+\del)^{y+3}(\gamma (1+\del)^{j_1(C)-1})^2}
{2\del^2(1+\del)^{j_2(C)}} $$ $$= \frac{(1+\del)^{y+3}
\frac{\del^{24}}{(1+\del)^{2y}} (1+\del)^{2j_1(C)-2}}
{2\del^2(1+\del)^{j_2(C)}} \leq cost(C) \cdot \del^{22}(1+\del)^{5-y}
\leq \del^{100} cost(C)$$ (since $(1+\del)^y \geq
\frac{1}{\del^{1500}}$). The first increase in the cost  due to the transformation from $(X^*,Y^*)$ to $(X',Y')$ is therefore by an
additive factor of at most $\del^{100} Z^*$.

We let $n'_{r,i}=n_{r,i}-\sum_{C\in\C}n_{r,i}(C)X'_C-\sum_{C \in \C} Y'_{r,i,C}$ and if $n'_{r,i}>0$, then say that $n'_{r,i}$ is the number of unassigned jobs of size $(1+\del)^i$ and density $(1+\del)^r$. These jobs will be called {\it unassigned jobs}, as they will remain unassigned also after the first assignment step.
The first assignment step will be to assign jobs according to the
configurations for which $X'_C>0$, such that there will be $X'_C$
machines whose sets of large jobs will be as required (based on the definition of $C$). We will
then assign $Y'_{r,i,C}$ jobs of size $(1+\del)^i$ and density
$(1+\del)^r$ to the machines whose configuration is $C$ (these are
small jobs for $C$). In order to ensure that all jobs can be
scheduled, for each machine that is assigned the configuration
$C$, additionally to total size of at most
$t_r(C){\gamma}(1+\del)^{j_1(C)-1}$ jobs of density $(1+\del)^r$
that are small jobs for configuration $C$ that the machine can
contain, such jobs of total size at most
$\frac{3{\gamma}}{\del}(1+\del)^{j_1(C)-1}$ are allocated to this machine. These
jobs are called {\it additional jobs}, and we will calculate the
increase in the total cost as a result. The unassigned jobs will
all be assigned to a machine of speed $1$ whose configuration $C'$
has $j_1(C')=0$ or $j_1(C')=-1$. We will compute an upper bound on the total size
of these jobs that will allow us to find an upper bound on the
increase in the cost.

Consider the machines whose configuration is $C$. For every $0\leq
r \leq y$, create $X'_C$ {\it bins} of size
$t_r(C){\gamma}(1+\del)^{j_1(C)-1}$ (these bins are called bins of the first kind), and $X'_C$ bins of size
$\frac{3{\gamma}}{\del}(1+\del)^{j_1(C)-1}$ (these bins are called bins of the second kind); if $t_r(C)=0$, then we introduce
only the second kind of bins. The additional jobs are those that are packed into bins of the second kind. We define the allocation of jobs to machines by packing them as items into these bins. If the the number of items of a certain type assigned using this packing process exceeds the existing number of jobs, then some of the spaces allocated in this process for items will not receive jobs.
For every $i \in \I'(C)$, pack
$Y'_{r,i,C}$ {\it items} of size $(1+\del)^i$ into these bins
using First Fit. We show that all items are packed. Assume by
contradiction that this is not the case. Since the size of each
item is below ${\gamma}(1+\del)^{j_1(C)-1}$ and there is an item
that cannot be packed, each bin of the first kind (if such a bin exists)
is occupied by at least $(t_r(C)-1){\gamma}(1+\del)^{j_1(C)-1}$,
and each bin of the second kind is occupied by at least
$(\frac{3}{\del}-1){\gamma}(1+\del)^{j_1(C)-1}$. We find that the total size of
items of density $(1+\del)^r$ that are to be packed into these
bins as small, which is equal to $\sum_{i \in \I'(C)} (1+\del)^i
Y'_{r,i,C}$ is above $(t_r(C)+\frac{3}{\del}-2){\gamma}(1+\del)^{j_1(C)-1}X'_C$,
contradicting the upper bound that we proved on $\sum_{i \in \I'(C)} (1+\del)^i Y'_{r,i,C}$, since $\frac{3}{\del}-2 > 2+ \frac 2{\del}$ (since $\del \leq \frac{1}{8}$). The increase in the cost
of each machine (since the total size of small jobs of density $(1+\del)^r$ for each $0\leq r\leq y$ becomes larger) can be upper bounded as follows. A possible
schedule is obtained by assigning the jobs of the bins of the
first kind as a block of size at most
$t_r(C){\gamma}(1+\del)^{j_1(C)-1}$, and the jobs of the bins
of the second kind as a block of the last jobs assigned to the
machine (this is possibly not an optimal ordering). For a job $j$
that is small for its configuration the difference between its
cost and $\Gamma_j$ is already included in the objective function
value (since it is assigned as a small job), and thus we compute the total $\Gamma$-values of the added blocks.
Instead of considering these blocks (for different values of $r$) separately, we see it as one block assuming that all densities are equal to $(1+\del)^y$ (this assumption cannot reduce the cost). We
assume without loss of generality that $j_2(C)=0$. Let $W= \sum_{0 \leq r \leq y, i \in \I(C)} (1+\del)^i
n_{r,i}(C)+\sum_{r=0}^y t_r(C)({\gamma}(1+\del)^{j_1(C)-1})$. The total
size of the blocks of small jobs added to this machine is at most
$\frac{3(y+1){\gamma}}{\del}(1+\del)^{j_1(C)-1}$, and thus, by Claim \ref{gammaval}, the sum of
$\Gamma$-values for the all these blocks is at most
$(1+\del)^y(W+3(y+1){\gamma}(1+\del)^{j_1(C)-1}/(2\del))(3(y+1){\gamma}(1+\del)^{j_1(C)-1}/\del)$.
By Lemma \ref{lboncost}, we have $cost(C) \geq W^2/2$, and $W \geq
(1+\del)^{j_1(C)-2}$, and get that the total $\Gamma$-values of the added blocks is at most
$$\frac{(1+\del)^y(W+3(y+1){\gamma}(1+\del)^{j_1(C)-1}/(2\del))(3(y+1){\gamma}(1+\del)^{j_1(C)-1})}{\del}$$ $$\leq
3(y+1)\del^{11}W(1+{1.5(y+1)(1+\del)\gamma})(1+\del)W$$
$$\leq 3(1+\del)\del^2 \cdot W^2 (1+\del^{1400}) < 7\del^2
\cdot W^2/2$$ as $\del \leq 1/8$, $y+1 \leq \frac{1}{\del^9}$, and
$\gamma \leq \del^{1512}$. Thus, the second increase in the cost
is by an additive factor of at most $7\del^2\cdot Z^*$.

Recall that $n'_{r,i}$ is the number of unassigned jobs of size
$(1+\del)^i$ and density $(1+\del)^r$, and for $C \in
C_H$, $X'_C=X^*_C$ and $Y'_{r,i,C}\geq Y^*_{r,i,C}$. Using (\ref{3rd}), we have
$$n'_{r,i} = n_{r,i}-\sum_{C\in C_H}n_{r,i}(C)X'_C - \sum_{C \in C_H
: i \in \I'(C)} Y'_{r,i,C} -\sum_{C\in C_L}n_{r,i}(C)X'_C - \sum_{C \in C_L
: i \in \I'(C)} Y'_{r,i,C}
$$

$$
\leq n_{r,i}-\sum_{C\in C_H}n_{r,i}(C)X^*_C - \sum_{C \in C_H
: i \in \I'(C)} Y^*_{r,i,C} -
\sum_{C\in
C_L}(n_{r,i}(C)(X^*_C-1))$$ $$-\sum_{C \in C_L : i \in \I'(C), X^*_C>0}
(Y^*_{r,i,C}-\frac{Y^*_{r,i,C}}{X^*_C}-1)=\sum_{C\in
C_L}n_{r,i}(C)+\sum_{C \in C_L : i \in \I'(C), X^*_C>0}
(\frac{Y^*_{r,i,C}}{X^*_C}+1) \ . $$ The total size of unassigned
items is at most $$\sum_{r=0}^y \sum_{C\in C_L} \sum_{i \in \I(C)}
(1+\del)^i n_{r,i}(C)+\sum_{r=0}^y \sum_{C \in C_L : X^*_C>0}\sum_{ i \in
\I'(C)} (1+\del)^i (\frac{Y^*_{r,i,C}}{X^*_C}+1) \ . $$ Using
$\sum_{i\in \I'(C)} (1+\del)^i  \frac{Y^*_{r,i,C}}{X^*_C} \leq
(t_r(C)+1){\gamma}(1+\del)^{j_1(C)-1}$ for $C \in C_L$ such that $X^*_C>0$, the total
size of unassigned jobs is at most
 $$\sum_{r=0}^y \sum_{C\in C_L} \left((t_r(C)+1){\gamma}(1+\del)^{j_1(C)-1}+\sum_{i \in \I(C)} (1+\del)^i n_{r,i}(C)\right)+\sum_{r=0}^y \sum_{C
\in C_L}\sum_{ i \in \I'(C)} (1+\del)^i \ . $$
Using $\sum_{0 \leq r \leq y, i \in \I(C)} (1+\del)^i n_{r,i}(C)+\sum_{r=0}^y t_r(C)({\gamma}(1+\del)^{j_1(C)-1}) \leq (1+\del)^{j_1(C)}$
and $\sum_{ i \in \I'(C)} (1+\del)^i \leq \gamma(1+\del)^{j_1(C)} /\del$ for any $C \in \C$,
the total size of unassigned jobs is at most $$\sum_{C \in C_L}  (1+\del)^{j_1(C)} + \sum_{r=0}^y \sum_{C\in C_L} {\gamma}(1+\del)^{j_1(C)-1}+\sum_{r=0}^y \sum_{C
\in C_L} \gamma(1+\del)^{j_1(C)} /\del <  \sum_{C \in C_L} 2(1+\del)^{j_1(C)} \ . $$ since $(1+(y+1)\gamma(1/\del+1))\leq 1+ \frac{1}{\del^9} \cdot \del^{1500} \cdot \frac{2}{\del}  <2$.

We will show in what follows that the definition of $g(\del)$ is such that the total size of unassigned jobs is indeed at most $2\gamma$.
The jobs are assigned to a machine with a configuration $C'$ such that $j_1(C') \in \{0,-1\}$, its completion time of this machine before any modifications is at most $1$, and after the modifications it is at most $1+\frac {3}{\del} \gamma (y+1)<1+\del^{1400}$ (since ${3} \gamma (y+1)/\del\leq \frac{3}{\del} \cdot \del^{1500} \frac 1{\del^9} < \del^{1400}$).
Even if the density of each such job is $(1+\del)^{y}$, and the unassigned jobs are now assigned as a block on the most loaded machine of speed $1$ starting at time $1+\del^{1400}$, the sum of their weighted completion times will be at most $2\gamma(1+\del)^y(1+\del^{1400}+2\gamma)<2\del^{12}(1+\del^{1399})<3\del^{12}$, while the value of $cost(C')$ is at least $(\frac {1}{1+\del})^6/2 > \frac 15$, so the increase in the cost is at most by an additive factor of $15\del^{12}\cdot Z^*$. The total cost will be at most $Z^*(1+\del^{100}+7\del^2+15\del^{12})<(1+\del)Z^*$.

Consider a speed $s$. If $s$ is a fast speed, then for every $j_1$
such that $j_1 \leq g(\del)$
 there are at most $(\frac{1}{\del})^{\frac{1}{\del^{230}}}$ light configurations (and the others are heavy), and we get $\sum_{C \in C_L : j_2(C)=s} (1+\del)^{j_1(C)} \leq (\frac{1}{\del})^{\frac{1}{\del^{230}}}(1+\del)^{g(\del)}\sum_{i=0}^{\infty} \frac{1}{(1+\del)^i} \leq (1+\del)^{g(\del)+1}(\frac{1}{\del})^{\frac{1}{\del^{231}}}$.
If $s$ is a slow speed, all the configurations $C$ with
$j_2(C)=\log_{1+\del} s$ are light, but every such configuration
$C$ satisfies $(1+\del)^{j_1(C)-2} \leq \frac{2s}{\del}$, and we
get that the corresponding sum is at most  $
\frac{2s(1+\del)^2}{\del}(\frac{1}{\del})^{\frac{1}{\del^{230}}}$.
Taking the sum over all $s$ such that $s$ is a slow speed we get
$$
\frac{2(1+\del)^{f(\del)+3}}{\del^2}\left(\frac{1}{\del}\right)^{\frac{1}{\del^{230}}}=
\frac{2(1+\del)^{g(\del)-\frac{2}{\del^2}+3}}{\del^2}\left(\frac{1}{\del}\right)^{\frac{1}{\del^{230}}}
\leq
\frac{(1+\del)^{g(\del)+3}}{\del}\left(\frac{1}{\del}\right)^{\frac{1}{\del^{230}}}
\ , $$ since $(1+\del)^{2/\del^2} \geq 2/\del$. In total, taking into account all $-f(\del)=-g(\del)+2/\del^2$ fast speeds and all slow speeds, we find
$\sum_{C \in C_L} 2(1+\del)^{j_1(C)} \leq
(2/\del^2-g(\del))(1+\del)^{g(\del)}(\frac{1}{\del})^{\frac{1}{\del^{232}}}$.
Recall that $g(\del)=-(\frac{1}{\del})^{\frac{1}{\del^{300}}}$.
Let $G(\del)=-g(\del)$. We find $G(\del)\geq
(\frac{1}{\del})^{\frac{1}{\del^{299}}}\cdot \frac{1}{\del^4}$.
Using
$(1+\del)^\frac{1}{\del^4}=((1+\del)^\frac{1}{\del^2})^{1/\del^2}
\geq (\frac 1{\del})^{1/\del^2}$, we get $(1+\del)^{G(\del)} \geq
((\frac 1{\del})^{1/\del^2})^{1/\del^{299}}=(\frac
1{\del})^{1/\del^{301}}=G(\del)^{1/\del} \geq G(\del)^8$. Since
$2/\del^2 < G(\del) < (\frac 1{\del}-1) G(\del)$, and since
$\gamma \geq \del^{\frac1{\del^7}}$, it is sufficient to show
$G(\del)^7 \geq (\frac 1{\del})^{\frac 1{\del^{235}}}$, which
holds for our choice of $g(\del)$.
\end{proof}

In summary, in this section we provided an
$(1+\del)$-approximation algorithm for the bounded ratio problem.
The algorithm tests all possible intervals $D_{j,b}$, and for each
one it constructs the MILP $\Pi$, finds an optimal solution,
and if a solution exists, converts it into a schedule. The running
time of this last conversion is $O(m+n)$, and there are
$O(\frac{n^2}{\del})$ intervals. The number of configurations is
at most $n^2m(\frac{1}{\del})^{1/\del^{231}}$, which is an upper
bound on the number of variables $X_C$. The number of variables
$Y_{r,i,C}$ is therefore at most
$n^3m(\frac{1}{\del})^{1/\del^{232}}$ since the number of values
for $i$ is at most $n$ and the number of values of $r$ is at most
$\frac{1}{{\del}^{9}}$. The number of variables $X_C$ that are
integral variables is at most
$(-f(\del))(-f(\del)+1)(\frac{1}{\del})^{1/\del^{230}}$, and since
$-f(\del)=-g(\del)+\frac2{\del^2}=(\frac{1}{\del})^{\frac{1}{\del^{300}}}+\frac{2}{\del^2}$,
so $(-f(\del)+1)^2 \leq (\frac{1}{\del})^{\frac{1}{\del^{301}}}$,
and $(-f(\del))(-f(\del)+1)(\frac{1}{\del})^{1/\del^{230}} \leq
(\frac{1}{\del})^{1/\del^{301}}$. The number of constraints
(excluding non-negativity constraints for all variables) is
$O(m+n^2m(\frac{1}{\del})^{1/\del^{231}})$.

\section{An EPTAS with release dates\label{sec:release_dates}}
We generalize the EPTAS of the previous section for the case with release dates. In this case, in the original instance $A$, a job $j$ has a time $\rho_j \geq 0$ associated with it, such that $j$ cannot be executed before time $\rho_j$. A solution or schedule is not only a partition of the jobs
 among the machines, but it also states the completion time of each
 job, which is sufficient as the schedule is non-preemptive (an
 alternative way to define a schedule is by defining the starting
 times, but we will use the former option).
A schedule (or a valid schedule) must comply with the property that no machine runs more than one job at a time, except for times when one job completes its processing and another one starts its processing. We will assume in what follows that the time slot that a job runs on a machine is half open, that is, if $j$ runs on machine $i$, its time slot is $[C_j-\frac{a_j}{v_i},C_j)$. Note that an optimal schedule may have idle times.
Using scaling (sorting the machines and dividing all speeds and sizes by the largest machine speed), we assume that $1 = v_1\geq v_2 \geq \cdots \geq
v_m>0$. Let $0<\del \leq 1/36$ be an accuracy factor,
that is a function of $\eps$ (where $\del<\eps$), and  such that
$\frac{1}{\del}$ is an integer.

\subsection{Standard rounding steps}
As a first step, we convert $A$ into $A'$.
Let $a_{\min}=\min_{1 \leq j \leq n} a_j$.
Let $\beta=\del^2 a_{\min}$. The sets of jobs and machines are
the same as in $A$. Let $r'_j=(1+\del)^{\lceil \log_{1+\del}
(\rho_j+\beta) \rceil}$ be the re;ease date of job $j$. Note that $\rho_j+\beta >0$, so $r'_j$ is well-defined for any $j$ (even if $\rho_j=0$). Let
$s_i=(1+\del)^{\lfloor \log_{1+\del} v_i \rfloor}$ be the speed
of machine $i$ in instance $A'$. Let $w_j=(1+\del)^{\lceil
\log_{1+\del} \omega_j \rceil}$, $p_j=(1+\del)^{\lceil
\log_{1+\del} a_j \rceil}$ be the weight and size
(respectively) of job $j$ in instance $A'$.
Using the new notation, for a given schedule, the completion
time of $j$ is still denoted by $C_j$, and its weighted
completion time is $w_j\cdot C_j$. Note that $1=s_1 \geq s_2 >
\cdots \geq s_m>0$. Let $\cO$ and $\cO'$ denote optimal solutions for $A$ and $A'$
respectively.

\begin{claim}
\label{easydirection}
Every solution $SOL$ for $A'$ is also a solution for $A$, and $SOL(A)\leq SOL(A')$.
\end{claim}
\begin{proof}
Recall that a solution is defined by an allocation of jobs to
machines and to completion times. Thus, $C_j$ is the completion time of job $j$ no matter which input is considered ($A$ or $A'$), and since $w_j \geq \omega_j$ for
any job $j$, it is sufficient to show that $SOL$ is a valid
solution for $A$.

For job $j$ assigned to machine $i$ in $SOL$, let
$t$ be the starting time of $j$ in $SOL$ when the input $A$ is considered,
and let $t'$ be the starting time of $j$ in $SOL$ when the input $A'$ is considered. We will show that $t \geq t'$. Since $\rho_j \leq r'_j $, and $t' \geq r'_j$ by the validity of $SOL$ for $A'$, we will have $t \geq t' \geq r'_j \geq \rho_j$, so $j$ starts running after its release date in $SOL$ for $A$ as well. Moreover, since $t' \geq C_{j'}$ for any job $j'$ that is assigned to run on $i$ in $SOL$, such that the completion time of $j'$ satisfies $C_{j'}<C_j$, we will have $t \geq t' \geq C_{j'}$, and thus no overlaps between jobs assigned to the same machine exist in $SOL$ for the input $A$. These two properties will imply the validity of $SOL$ for $A$.

The speeds of $i$ for $A$ and $A'$ are $v_i$ and $s_i$ respectively. The sizes of $j$ for $A$ and $A'$ are $a_j$ and $p_j$, respectively. Thus we have
$C_j=t+\frac{a_j}{v_i}$ and $C_j=t'+\frac{p_j}{s_i}$. Since $p_j \geq a_j$ and $s_i \leq v_i$, $\frac{a_j}{v_i} \leq \frac{p_j}{s_i}$, and therefore
$t \geq t'$, as required.
\end{proof}

Given a solution $SOL$ defined for a given input, let the time-augmented solution $TA(SOL,\upsilon)$ for some $\upsilon>1$ and the same sets of jobs and machines be the solution where each job is assigned to the same machine as in $SOL$, and the {\it completion time} of each job is $\upsilon$ times its completion time in $SOL$. In the next claim we show that if $SOL$ is a solution for $A$, then  $TA(SOL,\upsilon)$ is a solution for $A$ as well.
A schedule $SOL''$ for $A'$ is called {\it timely} if for every job $j$, the starting time of $j$ is at least $\del \cdot \frac{p_j}{s_i}$, where $i$ is the machine that runs $j$ in $SOL''$.
For a solution $SOL$ for $A$,  we will be able to use
$TA(SOL,\upsilon)$ as a solution for $A'$ for certain values of
$\upsilon$. The next claim specifies cases such that this is
indeed a valid schedule for $A'$, and moreover it is a timely schedule.

\begin{claim}
\label{stretchedworks}
Given a schedule $SOL$ for $A$ and $\upsilon>1$, let $SOL'=TA(SOL,\upsilon)$. The schedule $SOL'$ is a schedule for $A$ such that $SOL'(A)=\upsilon \cdot SOL(A)$.
If $\upsilon \geq (1+\del)^3$, then $SOL'$ is a timely schedule for $A'$.
\end{claim}

\begin{proof}
First, we consider $A$. Let $t$ and $t'$ denote the starting times of job $j$ in $SOL$ and $SOL'$, respectively, and let $i$ be the machine that $j$ is assigned to. Its completion time in $SOL$ is $t+\frac{a_j}{v_i}$, its completion time in $SOL'$ is $\upsilon(t+\frac{a_j}{v_i})$, and $t'=\upsilon(t+\frac{a_j}{v_i})-\frac{a_j}{v_i}>\upsilon t$ holds since $\upsilon>1$. Any job that runs on the same machine as $j$ and completes before $j$ in $SOL$ has a completion time of at most $\upsilon t$ in $SOL'$, and therefore for each machine, no overlaps between time intervals dedicated to different jobs are created. Moreover, as $SOL$ is valid, $t' > t \geq \rho_j$, so no job is started before its release date. The claim regarding the cost holds since
the cost of each job increases exactly by a factor of
$\upsilon$ (as its weight does not change).

Next, consider $A'$ and job $j$ again. Let $t$ denote its starting time of $j$ in $SOL$ for $A$, and let $t''$ denote the starting time of $j$ in $SOL'$ for $A'$. Once again let $i$ be the machine that $j$ is assigned to. We have $t''=\upsilon(t+\frac{a_j}{v_i})-\frac{p_j}{s_i}\geq \upsilon t+\upsilon\frac{p_j/(1+\del)}{s_i(1+\del)}-\frac{p_j}{s_i} \geq \upsilon t+\del\frac{p_j}{s_i}$ as $\upsilon \geq (1+\del)^3$. Since $t'' \geq \upsilon t$, as in the proof for $A$, there are no overlaps between the time intervals of jobs assigned to one machine. The release date $r'_j$ satisfies $r'_j \leq (1+\del)(\rho_j+\beta) = (1+\del)(\rho_j+\del^2 a_{\min})$. Since $a_{\min} \leq a_j \leq p_j$ and $s_i \leq 1$, $r'_j \leq (1+\del)\rho_j+(1+\del)\del^2 \frac{p_j}{s_i} < \upsilon t+\del\frac{p_j}{s_i} \leq t''$ as $t \geq \rho_j$, $\del \leq \frac{1}{36}$, and $\upsilon \geq (1+\del)^3$.
Since $t'' \geq \del\frac{p_j}{s_i}$, $SOL'$ is a timely schedule for $A'$.
\end{proof}

In what follows, when we are given a solution $SOL$ for $A$, we let $\bar{SOL}=TA(SOL,(1+\del)^3)$.

\begin{claim}\label{timely}
We have $(1+\del)^3 SOL(A)  \leq \bar{SOL}(A') \leq (1+\del)^4 SOL(A)$ and $\cO(A)  \leq \cO'(A') \leq (1+\del)^4 \cO(A)$.
\end{claim}
\begin{proof}
Recall that $\bar{SOL}$ is a solution for both $A$ and $A'$. Moreover, for each job, its completion time in $\bar{SOL}$ is the same for $A$ and $A'$. Since $\omega_j \leq w_j \leq (1+\del)\omega_j$, we get $\bar{SOL}(A)  \leq \bar{SOL}(A') \leq (1+\del) \bar{SOL}(A)$, and the first claim follows from $\bar{SOL}(A)=(1+\del)^3 SOL(A)$ (by Claim \ref{stretchedworks}).
The second part of the second claim follows from the first claim; since $\cO$ is an optimal solution for $A$,  we have $\cO'(A') \leq \bar{\cO}(A')\leq (1+\del)^4 \cO(A)$.  The property $\cO(A)  \leq \cO'(A')$ follows from Claim \ref{easydirection}, as $\cO'$ is a valid solution for $A$, and thus $\cO(A) \leq \cO'(A) \leq \cO'(A')$.
\end{proof}

We will only consider timely schedules for $A'$. Let $\cO''$ denote an optimal timely schedule for $A'$.

\begin{claim}
We have $\cO''(A') \leq (1+\del)^4 \cO'(A')$.
\end{claim}
\begin{proof}
Consider the schedule $\bar{\cO}=TA(\cO,(1+\del)^3)$. By Claim \ref{stretchedworks}, this is a timely schedule for $A'$. By Claim \ref{timely}, $\bar{\cO}(A') \leq  (1+\del)^4 \cO(A) \leq (1+\del)^4  \cO'(A')$. Thus, $\cO''(A') \leq (1+\del)^4 \cO'(A')$.
\end{proof}

\begin{corollary}\label{kk}
If a solution $SOL$ satisfies $SOL(A') \leq (1+k\del)\cO''(A')$ for some $k>0$,
then $SOL(A) \leq (1+(4k/3+12)\del) \cO(A)$.
\end{corollary}
\begin{proof}
Recall that $SOL$ is a solution for $A$ as well. We have $SOL(A) \leq SOL(A') \leq (1+k\del)\cO''(A') \leq (1+k\del)(1+\del)^4\cO'(A') \leq
(1+k\del)(1+\del)^8 \cO(A)$. We show that $(1+\del)^{\pi} \leq \frac 32 \pi \del$ if $\pi \geq 1$ is an integer such that $\del \leq \frac {1}{3\pi}$, by induction. For $\pi=1$, $1+\del \leq 1+1.5\del$ holds. For larger values of $\pi$ we have $(1+\del)^{\pi+1}=(1+\del)(1+\del)^{\pi}  \leq (1+\del)(1+1.5 \pi\del)=1+(1.5\pi+1)\del+1.5\pi\del^2\leq 1+1.5\pi\del+\del+0.5\del=1+1.5(\pi+1)\del$.
Using $\pi=8$ and $\del \leq \frac 1{36}$, we find $(1+k\del)(1+\del)^8  \leq (1+k\del)(1+12\del) \leq 1+k\del+12\del+12k\del^2 \leq (4k/3+12)\del$.
\end{proof}

In the analysis of algorithms for the input $A'$, we will use {\it pseudo-costs} rather than actual costs. The completion time of any job is strictly positive. The pseudo-cost of job $j$ whose completion time $C_j$ satisfies $C_j \in [(1+\del)^{i},(1+\del)^{i+1})$ is defined as $w_j(1+\del)^{i+1}$. Let $\cO_{PC}$ denote an optimal timely solution for $A'$ with respect to pseudo-cost (that is, the total pseudo-cost of all jobs is minimized). Let $PC(SOL(A'))$ denote the pseudo-cost of $SOL$ for the input $A'$.

\begin{claim}
We have $\frac{PC(SOL(A'))}{1+\del} \leq SOL(A') \leq PC(SOL(A')) $, and $PC(\cO_{PC}(A')) \leq PC(\cO''(A')) \leq (1+\del)\cO''(A')$. If a solution $SOL'$ satisfies $PC(SOL'(A')) \leq (1+k'\del)PC(\cO_{PC}(A'))$, then $SOL'(A) \leq (1+(2k'+14)\del) \cO(A)$.
\end{claim}
\begin{proof}
For any job $j$ its cost $w_jC_j$ in a solution $SOL$ is at least $\frac{1}{1+\del}$ fraction of its pseudo-cost and at most its pseudo-cost, which implies the first property and the inequality
 $ PC(\cO''(A')) \leq (1+\del)\cO''(A')$.
The inequality $PC(\cO_{PC}(A')) \leq PC(\cO''(A'))$ holds since $\cO_{PC}$ is an optimal timely schedule for $A'$ with respect to pseudo-cost and $\cO''$ is timely.
If $PC(SOL'(A')) \leq (1+k'\del)PC(\cO_{PC}(A'))$, then we find $SOL'(A') \leq PC(SOL'(A')) \leq (1+k'\del)PC(\cO_{PC}(A')) \leq (1+\del)(1+k'\del)\cO''(A') \leq (1+(k'+1+k'/36)\del) \cO''(A')$ (since $\del \leq \frac 1{36}$). Using Corollary \ref{kk} the last claim follows.
\end{proof}

\subsection{A simple representation of schedules}
We consider the input $A'$ and timely schedules. In what follows we only discuss pseudo-costs.  We slightly abuse  notation and use $SOL(A')$ to denote the {\it pseudo-cost} of $SOL$ for $A'$. Let $\J_{i,\ell}$ denote the time interval $[(1+\del)^{i},(1+\del)^{i+1})$ on machine $\ell$. Obviously, the total size of jobs for which both the starting time and completion time are within this interval (on machine $\ell$) is at most $s_{\ell} \cdot \del \cdot (1+\del)^i$, which is called the {\it length} of this interval. We let $\theta$ denote the smallest value of $i$ such that $(1+\del)^{\theta}$ is a release date of some job of $A'$.
Since we use pseudo-costs, we will represent schedules by stating for each job $j$ the machine that runs it, the interval where $j$ starts, and the interval that contains the completion time of the job (both being time intervals of the same machine). Note that if the completion time of $j$ is $(1+\del)^i$ (on some machine $\ell$), then we say that the completion time of $j$ belongs to $J_{i,\ell}$ even though the time slot that $j$ runs in is $[(1+\del)^i-\frac{p_j}{s_i},(1+\del)^i)$ (so $j$ is completed just before $\J_{i,\ell}$).
Alternatively, it is possible to state, for each interval, the list of jobs whose starting times are in this interval, and the list of jobs whose completion times are in this interval. The cost of the schedule can be computed by computing the total cost of all intervals. The cost of $\J_{i,\ell}$ is $(1+\del)^{i+1}$ times the total weight of jobs whose completion times are in $\J_{i,\ell}$.
Obviously, additional conditions are required for such lists to ensure that a list originates in a valid schedule (where an exact completion time is assigned to each job), and a complete schedule of the same cost (i.e., pseudo-cost) can be constructed. For example, if job $j$ has a starting time in $\J_{i_1,\ell}$ and completion time in $\J_{i_2,\ell}$, where $i_2-i_1 \geq 2$, then for any $i_1<i'<i_2$, no job can have a starting time or a completion time in $\J_{i',\ell}$. Moreover, the total size of jobs that have to run in a sequence of intervals cannot exceed their total length.
In the next lemma we formulate necessary and sufficient conditions for a list to represent an actual schedule.

\begin{lemma}
Consider the following conditions on a list.
\begin{enumerate}
\item For every job $j$, that has a starting time in $\J_{i,\ell}$ and a completion time in $\J_{i',\ell'}$, it holds that $\ell'=\ell$, $i'\geq i$, and $r'_j\leq (1+\del)^i$.
\item  For any $1\leq \ell \leq m$, $i\geq \theta$, and $\theta'$ such that $\theta \leq \theta' \leq i$ the total size of jobs for which both the completion times and starting times are in $\J_{\theta',\ell} \cup \J_{\theta'+1,\ell} \cup \cdots \cup \J_{i,\ell}$ is strictly smaller than the total length of the intervals $\J_{\theta',\ell},\J_{\theta'+1,\ell},\ldots,\J_{i,\ell}$.
\item For any $j$ whose starting time and completion time are in the intervals $\J_{i_1,\ell}$ and $\J_{i_2,\ell}$ respectively, such that $i_2>i_1$, any interval $\J_{i',\ell}$ with $i_1<i'<i_2$ has no starting times of jobs and no completion times of jobs.
\item For any $\J_{i,\ell}$, there is at most one job whose starting time is in this interval and its completion time is not, and at most one job whose completion time is in this interval, but its starting time is not.
\end{enumerate}
Every schedule has these properties, and for every list that has these properties there exists a schedule of the same cost that obeys the requirements on starting times and completion times of jobs of this list.
\end{lemma}
\begin{proof}
Given a schedule where completion times of jobs are specified, all properties hold trivially, thus it is left to show that given a list satisfying all properties, it is possible to assign a completion time (and thus a time slot) to every job.

Consider a schedule given as a list. We refer to the time intervals where the starting time and completion time of a job $j$ are listed as its listed starting time and listed completion time. We will assign each job an actual starting time and an actual completion time.
Due to the first property, the list defines a partition of jobs to machines, and therefore each machine can be considered separately. Consider machine $\ell$ and apply the following process. Find
the first time interval $\J_{i,\ell}$ (with the smallest $i$) that has at least one job whose starting time is in $\J_{i,\ell}$. Let
$i'=i$, repeat the following process, and stop after all jobs of $\ell$ have
been assigned to actual starting times and actual completion times.  If $\J_{i',\ell}$ has no jobs
listed to start in this time interval, then let $i'=i'+1$.
Otherwise, starting the time that is the maximum between $(1+\del)^{i'}$ and the last actual completion time assigned to any job, schedule the jobs having both listed starting times and listed completion times in the time interval $J_{i',\ell}$ to actual time slots consecutively in some order by assigning them actual starting times and actual completion times. If there is a job $j$ listed to start in
$\J_{i',\ell}$ and a listed completion time in $\J_{i'',\ell}$ where
$i''>i'$, assign it an actual completion time which is the maximum between
$(1+\del)^{i''}$ and the last actual completion time assigned to any job
in $\J_{i',\ell}$ (or $(1+\del)^{i'}$ if no such job exists) plus $\frac{p_j}{s_{\ell}}$, and a suitable starting time, and let $i'=i'+1$.

Each job with a listed starting time in $\J_{i,\ell}$ was assigned to an actual time slot no earlier than $(1+\del)^i$, and without any overlap (since a job is either assigned an actual starting time that is the actual completion time of the previously assigned job, or a later actual starting time).
Every job $j$ that was assigned for a given value $i'$ has a listed starting time in $\J_{i',\ell}$, and it was assigned an actual starting time of at least $(1+\del)^{i'}$. Since the listed starting time of $j$ is in $\J_{i',\ell}$, by the first condition, $r'_j\leq (1+\del)^{i'}$, and $j$
is assigned an actual starting time that is no smaller than its release date. It remains to show that any job $j$ whose listed starting time is in $\J_{i_1,\ell}$ and its listed completion time is in $\J_{i_2,\ell}$ (for $i_2\geq i_1$) was indeed assigned an actual starting time and an actual completion time in these intervals, respectively.
Assume by contradiction that the process failed to do so. Let $j'$ be the first job for which either the starting time or the completion time is not in the appropriate time intervals, and consider the first actual time that violates the requirement (i.e., if $j'$ was assigned both an actual starting time and an actual completion time that are not with accord with the listed ones, we will consider its starting time). Consider the case $i_1=i_2$, where $j$ was assigned a completion time of at least $(1+\del)^{i_1+1}$, consider the last time before this time that the machine is idle (there must be such a time as the machine is idle in $[0,(1+\del)^i)$, for some $i\geq \theta$). By construction, the last idle time ends at some time $(1+\del)^{\theta'}$ for $\theta \leq \theta'\leq i_1 $. Consider the set of jobs that were assigned to actual time slots starting $(1+\del)^{\theta'}$ and up to the actual completion time assigned to $j'$ (in particular, this set includes $j'$). The listed starting times of these jobs are in $\J_{\theta',\ell} \cup \J_{\theta'+1,\ell} \cup \cdots \cup \J_{i_1,\ell}$. Jobs with listed starting times in $\J_{\theta',\ell} \cup \J_{\theta'+1,\ell} \cup \cdots \cup \J_{i_1-1,\ell}$ have listed completion times in $\J_{\theta',\ell} \cup \J_{\theta'+1,\ell} \cup \cdots \cup \J_{i_1,\ell}$, as otherwise, by the third condition, $j'$ cannot exist. A possible job with a listed starting time in $\J_{i_1,\ell}$ and a listed completion time in a later time interval of $\ell$ is not considered by the process above up to $j'$. Thus, we find that the jobs that have both a listed starting time and a listed completion time in $\J_{\theta',\ell} \cup \J_{\theta'+1,\ell} \cup \cdots \cup \J_{i_1,\ell}$
violate the second condition, as their total size is at least the total size of the intervals $\J_{\theta',\ell},\J_{\theta'+1,\ell},\ldots,\J_{i_1,\ell}$.  Consider the case $i_1<i_2$, where $j'$ is assigned a starting time of at least $(1+\del)^{i_1+1}$. In this case $j'$ is not included in the set of jobs violating the condition, and since by the fourth condition the other jobs assigned for $i_1$ have listed completion times in $\J_{i_1,\ell}$, their assigned completion time is strictly smaller than $(1+\del)^{i_1+1}$ contradicting the assumption on the starting time of $j'$. Consider the case $i_1<i_2$, where $j'$ is assigned an actual completion time of at least $(1+\del)^{i_2+1}$. Since $j'$ was not assigned a smaller actual completion time, there is no idle time between the job that was assigned just prior to $j'$ and $j'$. Moreover, as $j'$ received an actual starting time in $\J_{i_1,\ell}$, and using the third condition, the set of jobs assigned to actual time slots starting the last idle time (denoted by $(1+\del)^{\theta'}$) and up to $j'$ have listed starting times and listed completion times in $\J_{\theta',\ell} \cup \J_{\theta'+1,\ell} \cup \cdots \cup \J_{i_2,\ell}$, contradicting the second condition, as their total size is at least the total size of these time intervals.\end{proof}

Observe that if a timely schedule is given as a list, then every schedule that is created from the list (i.e., choosing a permutation of the jobs that start and complete in a common interval) results in a timely schedule.

\subsection{Shifted schedules and time stretched schedules}
For a job $j$ of size $(1+\del)^i$, let its stretched size be
$(1+\del)^{i+1}$. Given a schedule $SOL$ for $A'$, we define a
schedule $S(SOL)$,  called the {\it shifted} schedule of the
schedule $SOL$, and job sizes are stretched. The stretched input
(consisting of the same machines and jobs, where each machine has
the same speed in both inputs, each job has the same release date
but its processing time is exactly $1+\del$ times larger) is
called $\bar{A'}$. If job $j$ has a starting time in $\J_{i_1,\ell}$
and completion time in $\J_{i_2,\ell}$ (where $i_2 \geq i_1$) in
$SOL$, we define its starting time and its completion time in
$S(SOL)$ to be in the intervals $\J_{i_1+1,\ell}$ and
$\J_{i_2+1,\ell}$, respectively.

\begin{claim}\label{stre}
If $SOL$ is a valid timely schedule, then so is $S(SOL)$, and the costs satisfy $S(SOL)(\bar{A'})=(1+\del)SOL(A')$.
\end{claim}
\begin{proof}
A job $j$ starts later in $S(SOL)$ than it does in $SOL$, and therefore the release dates are respected. The schedule $S(SOL)$ can be obtained from $SOL$ by multiplying each starting time and completion time by $1+\del$. It remains timely since both the size and the starting time are multiplied by the same factor.
\end{proof}

The schedule $S(SOL)$ can obviously be used also as a schedule for an input where the size of each job in $A'$ is stretched by some factor in $[1,1+\del]$ (for this we consider $S(SOL)$ as an assignment of a machine and completion time for each job), and it remains timely since jobs can only start later while their sizes can only decrease.  This is the case even if the stretch factors of different jobs may be different.

Given a timely schedule $\hat{S}$ for an input $\bar{A}$, we define a new schedule that we call the {\it schedule obtained from $\hat{S}$ by time stretching by a factor of $1+\del$}, as follows.
If $j$ is assigned to run (in $\hat{S}$) on machine $\ell$ during $[t,t')$, then we have $p_j=s_{\ell}(t'-t)$.
We reserve the time period $[(1+\del)t,(1+\del)t')$ on machine $\ell$ for job $j$ (where $(1+\del)t$ is the reserved starting time of $j$, and $(1+\del)t'$ is the reserved completion time of $j$), and assign $j$ to start at time $(1+\del)t+\frac{\del p_j}{2s_{\ell}}$. This time will be called the basic starting time of $j$.
Obviously, no job will start running before its release date, and the schedule remains timely.
The basic completion time of $j$ will be $(1+\del)t+\frac{\del p_j}{2s_{\ell}}+\frac{p_j}{s_{\ell}}=(1+\del)t+\frac{p_j}{s_{\ell}}(1+\del)-\frac{\del p_j}{2s_{\ell}}= (1+\del)t+(t'-t)(1+\del)-\frac{\del p_j}{2s_{\ell}}
 = (1+\del) t'-\frac{\del p_j}{2s_{\ell}}$. If $j$ originally completed during $\J_{i,\ell}$ in $\hat{S}(\bar{A})$, then both its reserved and basic completion times  will be no later than in the interval $\J_{i+1,\ell}$, so the cost increases by at most a multiplicative factor of $1+\del$.  Next, for any interval $\J_{i,\ell}$, such that there is no job whose reserved time interval contains $\J_{i,\ell}$, shift all jobs that start and end during this time interval such that they run continuously as early as possible, that is, either starting at the beginning of $\J_{i,\ell}$, or just at the basic completion time of a job that starts in an earlier time interval and completes in $\J_{i,\ell}$ (jobs are reassigned to run without idle time, ignoring the time that was reserved before jobs are started or after they are completed). Moreover, if there is a job whose reserved starting time is during $\J_{i,\ell}$, its reserved completion time is in another time interval, but it is small for $\J_{i,\ell}$ (namely, its size is smaller than $\del^{10}$ times the length of the interval), it is also shifted to start as early as possible. The set of these jobs, that we call the {\it jobs of $\J_{i,\ell}$}, can be processed in any order, as long as the length of the interval is sufficient to accommodate all of them such that their completion times are within the time interval.
The new starting times and completion times will be called actual starting times and actual completion times, respectively.  For jobs whose time slots are not modified the actual starting times and actual completion times are equal to the basic starting times and basic completion times, respectively.

 \paragraph{Analysis.}
 Let $U$ denote the total size of the jobs whose reserved starting times and reserved completion times are in $\J_{i,\ell}$. If there is a job whose reserved starting time is this time interval but its reserved completion time is not, let $U_2$ denote the total size of this job that is assigned to run during $\J_{i,\ell}$ and $U'_2$ the size that was originally reserved for this job before its basic starting time during $\J_{i,\ell}$ (letting $U_2=U'_2=0$ if no such job exists), and similarly, if there is a job with a reserved starting time in an earlier time interval on the same machine and a reserved completion time in $\J_{i,\ell}$, let $U_1$ denote the total size of this job running during $\J_{i,\ell}$, and $U'_1$ the size that was reserved for this job after its basic completion time during $\J_{i,\ell}$ (where $U_1=U'_1=0$ if no such job exists). Let $Z=s_{\ell} \del(1+\del)^i$ denote the length of the time interval $\J_{i,\ell}$.

\begin{claim}
Assume that there is no job with a reserved time interval that contains $\J_{i,\ell}$, then the total size that is available for jobs of $\J_{i,\ell}$ is at
least $U+\del^2 Z$, and the current total size of the jobs of $\J_{i,\ell}$ is at
most $U+\del^{10} Z$.
\end{claim}
\begin{proof}
The total size of the jobs of $\J_{i,\ell}$ exceeds $U$ by at most the size of one small job for this time interval, i.e., by at most $\del^{10} Z$.

Out of the total size $Z$ of $J_{i,\ell}$, the size that was not originally reserved for any job is $Z-U_1-U'_1-U_2-U'_2-(1+\del)U$.
The last calculation is valid since according to the basic starting times, total size of $(1+\del)U$ is reserved for jobs starting and ending during $\J_{i,\ell}$. We have $U'_1 \geq \frac{\del}2 U_1$, and $U'_2 \geq \frac{\del}2 U_2$ (since for $y=1,2$, if $U_y\neq 0$, then $U'_y$ is equal to $\frac {\del}2$ times the size of the job, while $U_y$ may be only a part of this size, and if $U_y=0$, the claim trivially holds). The total size occupied by jobs running during $\J_{i,\ell}$ is exactly $U_1+U_2+U$. Thus, the total remaining size is at least $Z-U_1-U_2-U$. Using $U'_1 \geq \frac {\del}2 U_1$ we find $U_1+U'_1 \geq  (1+\frac {\del}2) U_1$, and $U_1 \leq (1-\frac{\del}3)(U_1+U'_1)$, since $\del<1$. Similarly, we have $U_2 \leq (1-\frac{\del}3)(U_2+U'_2)$, and $U \leq (1-\frac{\del}3)(1+\del)U$.
Using $Z \geq (1+\del)U+U_1+U_2+U'_1+U'_2$, we have
$Z-U-U_1-U_2 \geq Z-(1-\frac{\del}3)(U_1+U'_1+U_2+U'_2+(1+\del)U) \geq \frac{\del}3 Z
 \geq \del^2 Z$. \end{proof}

Consider a timely schedule $\hat{S}$ and the solution $\bar{S}$ obtained from $\hat{S}$ by time stretching by a factor of $1+\del$.  Let $\J_{i,\ell}$ be an interval that is not contained in a reserved time period of any job, then by the last claim it contains an idle time of length at least $\del^3$ times the length of the interval, and this kind of idle time in such an interval $\J_{i,\ell}$ is said to be a {\it gap} in $\J_{i,\ell}$.  Observe that there might be other time intervals containing idle time which are not gaps (in case that the time interval is contained in a reserved period of some job).

\begin{corollary}\label{timestretch}
Given a timely schedule $\hat{S}$, the schedule $\tilde{S}$ obtained from $\hat{S}$ by time stretching by a factor of $1+\del$ can be constructed in polynomial time and the cost (pseudo-cost) of this solution is at most $1+\del$ times the cost (pseudo-cost) of $\hat{S}$, and such that for each interval $\J_{i,\ell}$ for which one of the following holds for $\tilde{S}$: the entire time interval is idle, or
there is a  job with a reserved starting time in this time interval, or there is a job with a reserved completion time in this interval, the interval contains idle time of length at least $\del^3 Z$ where $Z=s_{\ell} \del(1+\del)^i$ is the length of $\J_{i,\ell}$. Moreover, any job that is small for the interval that contains its actual starting time has a completion time in the same interval.
\end{corollary}

Consider the solution $\bar{S}$ obtained from a timely schedule $\hat{S}$ by time stretching by a factor of $1+\del$, and let $j$ be a job whose reserved starting time is in the interval $\J_{i,\ell}$, then the processing time of $j$ on machine $\ell$ is at most $\frac{(1+\del)^{i+1}}{\del}$ and its actual starting time in $\hat{S}$ is at most $(1+\del)^{i+1}$, thus its actual completion time in $\hat{S}$ is no later than $(1+\del)^{i+1}\cdot (1+\frac{1}{\del})=\frac{(1+\del)^{i+2}}{\del}$, and therefore its reserved completion time in $\bar{S}$ is at most $\frac{(1+\del)^{i+3}}{\del}$ and the end of the time interval containing the reserved completion time of this job is at most $\frac{(1+\del)^{i+4}}{\del}$.  Therefore, using $\frac{(1+\del)^{4}}{\del} \leq \frac{1}{\del^2}$ we conclude the following.

\begin{claim}\label{consec_gaps}
If  $\bar{S}$ obtained from a timely schedule $\hat{S}$ by time stretching by a factor of $1+\del$, and $\J_{i,\ell}$ has a gap, then $\hat{S}$ has a gap on machine $\ell$ during the time frame $[(1+\del)^{i+1},(1+\del)^{i}\frac{1}{\del^2})$.
\end{claim}

In what follows, given a time-stretched schedule, we refer to actual starting times and actual completion times simply as starting times and completion times, respectively.

\subsection{Job shifting procedure}
Our next goal is to create a new input from $A'$ by increasing the release date of some of the jobs while keeping all other parts of the input unchanged such that the following holds.  There is a positive constant $z$, such that for every integer value of $t$, the set of jobs released at time $(1+\del)^t$ of a common density can be scheduled on the $m$ machines to complete no later than time $z(1+\del)^t$.

First, we define {\it divisions} of job sizes as follows. Let $\Delta=\lceil \log_{1+\del} 2 \rceil$. We have $\Delta \leq \frac{1}{\del}$ as $(1+\del)^{\frac 1{\del}}>2$. For a job size $(1+\del)^i$, let $k_i$ be an integer such that $2^{k_i} \cdot (1+\del)^i \in (1,2]$. Let $k'_i=\lceil \log_{1+\del} 2^{k_i} \rceil+i$, where $k'_i \leq \Delta$ as $(1+\del)^{\Delta-i} =\frac{(1+\del)^{\Delta}}{(1+\del)^i} \geq {2}\cdot {2^{k_i-1}}=2^{k_i}$, and $k'_i \geq 1$, as $(1+\del)^{-i} <2^{k_i}$.
The division of a job of size $(1+\del)^i$ is defined to be $k'_i$, and its subdivision is $k_i$. The pseudo-size of such a job is defined to be $\pi_i=\frac{(1+\del)^{k'_i}}{2^{k_i}}$. We have $(1+\del)^{k'_i-i} \geq 2^{k_i}$ and $(1+\del)^{k'_i-i-1} < 2^{k_i}$. Thus we have $(1+\del)^i \leq \pi_i < (1+\del)^{i+1}$. The divisions $1,2,\ldots,\Delta$ form a partition of $[1,2)$. The division of a size is the part of $[1,2)$ that it belongs to when it is multiplied by an appropriate power of $2$. The subdivision is this last power of $2$.
For example, a job of size $(1+\del)^{-1}$ has the subdivision $1$, as $1<\frac{2}{1+\del}\leq 2$. Its division is $\lceil \log_{1+\del} 2\rceil -1=\Delta-1$, which means that once the size is multiplied by the appropriate power of $2$, it becomes relatively close to $2$. For a job of size $1+\del$, the subdivision is $0$, and its division is $1$. Let ${A}''$ be the instance $A'$ where the size of a job of size $(1+\del)^i$ is replaced with the pseudo-size $\pi_i$. The values $\pi_i$ of one division form a divisible sequence, and more specifically, given two such distinct values of one division, the larger one divided by the smaller one is a positive integral power of $2$.
In what follows we will schedule the jobs of ${A}''$ in some cases. Let $\tilde{\cO}$ be an optimal timely solution for ${A}''$.

\begin{corollary}
We have $\cO''(A') \leq \tilde{\cO}({A''}) \leq (1+\del) \cO''(A')$.
\end{corollary}
\begin{proof}
The first inequality holds since any timely schedule for ${A''}$ is a timely schedule for $A'$. The second inequality follows from  Claim \ref{stre}.
\end{proof}

\begin{claim}
Consider a timely schedule. If job $j$ is assigned to run on machine $\ell$ and has a starting time in $\J_{i,\ell}$, then $p_j < \frac{s_{\ell}}{\del}(1+\del)^{i+1}$. If the completion time of $j$ is in the same interval, then $p_j < s_{\ell} \del (1+\del)^{i}$.
\end{claim}
\begin{proof}
Let $t$ be the starting time of $j$. We have $t \geq \del \frac{p_j}{s_{\ell}}$. Since $t<(1+\del)^{i+1}$, the first claim follows. The second claim holds by the length of the interval.
\end{proof}
Recall that any job of size below $\del^{10} s_{\ell} \del
(1+\del)^i=\del^{11}s_{\ell}(1+\del)^i$ is defined to be small for
interval $\J_{i,\ell}$. We also say that any job of size in
$[\del^{11}s_{\ell}(1+\del)^i,s_{\ell}\del(1+\del)^i)$ is called
medium for this time interval, and job of size in
$[s_{\ell}(1+\del)^i,\frac{s_{\ell}}{\del}(1+\del)^{i+1})$ is
called large for this time interval, and larger jobs (of size at
least $\frac{s_{\ell}}{\del}(1+\del)^{i+1}$) are called huge for
the interval. Recall that a job that is huge for a given time
interval cannot start during this time interval in a timely
schedule. We say that a job is big for an interval if it is either
medium or large for it. The following definition uses the
pseudo-sizes of $j_1,j_2$ that are their sizes in ${A''}$.

\begin{definition}
A schedule for ${A''}$ is called {\it  organized} if it is
timely, and it has the following properties.
\begin{enumerate} \item Let $j$ be a job whose starting time is in $\J_{i,\ell}$. If $j$ is small for this interval, then its completion time is in the same interval.

\item Consider two jobs, $j_1,j_2$, of a given division $k$ and
the same density, such that the starting time of $j_1$ is in
interval $\J_{i_1,\ell_1}$, and the starting time of $j_2$ is in
interval $\J_{i_2,\ell_2}$, where $i_2>i_1$ or both $i_2=i_1$ and $\ell_2>\ell_1$.
Moreover, assume that
$r'_{j_2} \leq (1+\del)^{i_1}$.

\begin{enumerate}
\item If $\pi_{j_1} < \pi_{j_2}$, then  $j_2$ is not small for $\J_{i_1,\ell_1}$.

\item If $\pi_{j_1}=\pi_{j_2}$, then $r'_{j_2} \geq r'_{j_1}$.

\item If $\pi_{j_1}=\pi_{j_2}$, and
 $r'_{j_2}=r'_{j_1}$, then the index of the second job is smaller,
i.e., $j_2<j_1$.

\end{enumerate}
\end{enumerate}
\end{definition}

Note that in part 2(a) of the definition, if $i_1=i_2$, the case where $j_1$  is big for $\J_{i,\ell_1}$ while $j_2$ is small for $\J_{i,\ell_2}$, cannot occur. This definition defines fixed priorities on jobs. For jobs of equal sizes and densities, a job of smaller release date has higher priority, and out of two jobs with the same release date, the one of the higher index has a higher priority. For a fixed time interval, there are also priorities between jobs  of one division that are small for it. The priority between two jobs of one size remains the same, and additionally, a larger job has a higher priority than a smaller job (this is defined per time interval, and only for jobs of the instance that are small for it). Obviously, given a time interval $\J_{i,\ell}$, all jobs that are released at time $(1+\del)^{i+1}$ or later are irrelevant for it and have no priority.
Let $\bar{\cO}$ denote an optimal organized schedule for
${A''}$ (in particular, $\bar{\cO}$ must be timely).

\begin{lemma}
$\bar{\cO}({A''}) \leq (1+2\del) \tilde{\cO}({A''})$.
\end{lemma}
\begin{proof}
Given $\tilde{\cO}({A''})$, we define a schedule as follows.  First, we consider the solution obtained from $\tilde{\cO}({A''})$ by time stretching by a factor of $1+\del$. Thus, any job that is small for the interval that contains its starting time, the completion time of the job is  in the same interval as its starting time, and therefore  the first condition of an organized schedule holds.
We consider the set of jobs that have both starting times and completion times in $\J_{i,\ell}$, and we call them  {\it jobs of $\J_{i,\ell}$} (this set contains, in particular, all jobs with starting times in the interval that are small for this interval).
In what follows we will modify this set of jobs (possibly for multiple intervals) such that the total size of the jobs of the interval  may increase by an additive factor of $\del^8 (1+\del)^i$ (for $\J_{i,\ell}$).

Given a pair of intervals $\J_{i\ell}$ and $\J_{i',\ell'}$, we say
that $\J_{i',\ell'}$ is later than $\J_{i,\ell}$ if either $i'>i$
or both $i'=i$ and $\ell'>\ell$.  In this case, we also say that
$\J_{i,\ell}$ is earlier than $\J_{i',\ell'}$. This defines a
total order on the intervals. Now, we apply a process for every
interval, such that as a result, given a job $j$ of a division $k$
whose starting time is in an interval $\J_{i,\ell}$ and $j$ is
small for the interval (so this is a job of the interval), there
is no job $j'$ of the same division and density as $j$ that is
released at time $(1+\del)^i$ or earlier that is larger than $j$,
such that $j'$ is small for $\J_{i,\ell}$, and $j'$ starts in a later
interval $\J_{i',\ell'}$. The process will be applied for
$i=\theta,\theta+1,\ldots$ in increasing order, and the intervals
$\J_{i,\ell}$ for the different values of $\ell$ and a common
value of $i$ will be considered in an increasing order of $\ell$.
During this process for a given interval, we may modify the assignment of jobs
that are currently assigned both a starting time and a completion
times in $\J_{i,\ell}$, that are small jobs for this interval, that
is, the subset of jobs out of the the jobs of $\J_{i,\ell}$ that
are small for this interval may be partially swapped with other jobs. Recall
that if there is a job that starts in this interval but ends
later, then it is not small for this interval, and thus we do not
deal with it now (as we currently only deal with jobs that are small
for the intervals where their starting times are). When
$\J_{i,\ell}$ is treated, the cost of the solution will not
increase. The set of jobs of this interval will possibly be
modified iteratively until it reaches a new fixed set. Letting $Z$ be the
length of this interval, the total size of the final set of
jobs of the interval will be larger by at most an additive factor
of $\del^7 Z$ compared to the initial set. Thus, since the
interval has a gap in a time stretched schedule as any small job for the time interval had either a reserved starting time or a reserved completion time (or both) in this time interval, the resulting
set of jobs of the interval can be assigned starting and
completion times in this interval. After we finish dealing with a
time interval, we will not modify the set of jobs of this time
interval again when we apply the same process on other (later) time
intervals. The process is applied for every division separately
(but the different densities for one division are considered
together). In each step of the process for one time interval, we
will either decrease the number of jobs of $\J_{i,\ell}$, or we
will declare that the process is completed for a given division
and density (or both). Thus, the process is finite for each time interval and division.

For a given division $k$, in order to apply the process, we create
a list of jobs of this interval that are small for this time interval and that belong to division $k$,
sorted by non-increasing density, and for each density, the jobs
are sorted by non-increasing size. At each time, there will be a
prefix of this list whose jobs already fulfill the conditions
above (for any job, no larger job of the same division and density
that could have been assigned to the current time interval as a small job with
respect to its release date is assigned to a later time interval).
Consider the first job $j$ in the list for division $k$ that does
not fulfill the condition. Let $d$ be its density. Replace it with
the largest job $j'$ of the same division and density that is
assigned to a later interval and could be assigned to $\J_{i,\ell}$ as a small job (if there is no such job, then the process for division $k$ and this time interval is complete). Note that there is currently no job
of $\J_{i,\ell}$ of division $k$ and density $d$ that is larger
than $j$ but smaller than $j'$. Additionally, note that from this
time the only jobs of this division and density that will be moved
out of the set of jobs of $\J_{i,\ell}$ are smaller than $j'$. In
order to add $j'$ to the set of jobs of $\J_{i,\ell}$, we now
create a subset of jobs that will leave the set of jobs of
$\J_{i,\ell}$ and will be scheduled instead of $j'$. The first
such job is $j$. Jobs are added to the subset according to the sorted list of density $d$ (i.e., giving preference to larger jobs), as long as the suffix of the sorted list for density $d$ that started with $j$ is non-empty, and the total size $\pi_{j'}$ was not reached by the subset. Since the jobs sizes of one division (as we are dealing with ${A''}$, these are the pseudo-sizes defined above) form a
divisible sequence, that is, for every pair of distinct sizes, the smaller
one divides the larger one, the difference between $\pi_{j'}$ and the total size of the subset of jobs is always an integer multiple of the size of any  remaining job of density $d$ (that appeared after $j$ in the list). Thus, as long as there is still a job of density $d$ that appeared after $j$ in the sorted list, and if the size $\pi_{j'}$ was not reached yet, the result of selecting an additional job to the subset will never result in exceeding the total size $\pi_{j'}$. The selection is stopped in one of
two cases. The first case is that the total size of jobs selected for the subset
reaches exactly  $\pi_{j'}$.
The second case is that all jobs of density $d$
of division $k$ that appeared after $j$ in the sorted list (i.e., all jobs that were in the list before $j'$ was chosen, whose sizes were smaller than $\pi_{j'}$ were selected. In this case we continue to select jobs of the same division
and smaller densities (if such jobs exist) for the subset,  as long as the total size of the subset does not
exceed $\pi_{j'}$ (but in this case we may skip some jobs if the total size of the resulting subset  would exceed $\pi_{j'}$).  Once again,
the process stops either if the total size of selected jobs is
exactly $\pi_{j'}$, or if all jobs that could be selected
were already selected (here the only constraint on remaining jobs of smaller densities and the same division is that adding any such job to the subset would result in exceeding the total size $\pi_{j'}$). We have three cases in total. In the last two cases we are done dealing with density $d$ for the current division. In the first two cases, $j'$ is swapped with jobs whose total size is exactly $\pi_{j'}$, and in the last case it is swapped with jobs whose total size is smaller. Thus, the total size of the jobs of $\J_{i,\ell}$ may have increased, but the jobs that replace $j'$ have a total size no larger than $\pi_{j'}$.
Later we will show that the total size of the jobs of $\J_{i,\ell}$ remains sufficiently small after applying the process for all divisions. We now show that after applying the process for a single division, the cost did not increase (assuming that the total size of the jobs of the interval indeed remains small enough).  We will show that the
cost of the modified solution (by swapping $j'$ and the subset that we found) does not
exceed the previous cost. Let $\pi$ denote the total size of jobs
that were moved out of the set of jobs of $\J_{i,\ell}$, and let
$\pi'=\pi_{j'}$. We have $\pi \leq \pi'$. Let
$\J_{i',\ell}$ be the interval containing the completion time of $j'$. Let $w'=w_{j'}$, and let $w$ be the total weight of jobs of the selected subset. The
contribution of $j'$ and the jobs of the subset to the objective
function before the swap was $(1+\del)^{i'+1} w' + (1+\del)^{i+1} w$.
The modified contribution of these jobs is at most $(1+\del)^{i'+1} w +
(1+\del)^{i+1} w'$ (this includes the possibility that $j'$ starts in an earlier time interval and some of the jobs that are swapped with $j'$ will complete in this earlier time interval). Since $i'\geq i$, it is sufficient to show $w'\geq
w$. Since the density of a selected job is at most $d$, and the
total size of these jobs is $\pi \leq \pi'$, we find $w \leq d
\cdot \pi' =w'$ as required. Next, we compute the total excess of the time interval, i.e., the total size of all
jobs ever moved to be a part of the jobs of $\J_{i,\ell}$ compared
to the set of jobs that were removed from this set in the same step (the sum of differences $\pi'-\pi$ for all steps and all divisions). The total size
of the jobs of $\J_{i,\ell}$ increases only when in the process
above we have $\pi < \pi'$. In this case, the size of any remaining job of
division $k$ any a  sufficiently small density (below $d$) is larger than $\pi'-\pi$. We
claim that for every division, the total excess is at most
$2\del^{11}s_{\ell}(1+\del)^i$, that is, at most $2\del^{10}$
times the length of the time interval. This is obviously true for
empty divisions, while a non-empty division never becomes empty.
We say that a pair of division and density was
dealt with if at least one job of this division and density was moved into $\J_{i,\ell}$. In this case the list of jobs of $\J_{i,\ell}$ with this division and density cannot become empty (but some densities can become empty). Densities that did not exist in the set of jobs of $\J_{i,\ell}$ for a given division will never enter this set.
We claim that
throughout the process for a given division $k$, the current total excess
never exceeds twice the size of the smallest job for the division
$k$ that is currently a job of $\J_{i,\ell}$ and its density is at most the
last density that was dealt with and created an additional excess
(if no density was considered so far, the property is that the
total excess is zero). The property is true when the process
starts. Since according to the definition of the process, the size
of the smallest job for division $k$ and density $\tilde{d}$ (that is a
job of $\J_{i,\ell}$) cannot decrease throughout the process (as subsets of jobs are being swapped with larger jobs), the property
keeps holding in the case that no new excess is created. When the
first excess is created while considering jobs of density $d$,
this excess is no larger than the smallest job out of the remaining jobs of
the current division and density smaller than $d$, then we are done.
Finally, assume that the excess increases when job $j'$ is swapped
with job $j$ and the selected subset of jobs, when the process is applied for density
$d'$. Prior to this process, the excess was at most twice the
smallest job of density below $d$ (where $d$ is the density of the job considered in the previous iteration that created new excess), and thus it was at most twice
the smallest job of density below $d'$ (as $d' \leq d$). Let
$j''$ be the smallest job that is selected for the subset replacing $j'$. The difference between $\pi_{j'}$ and the total size of the selected jobs is an integer multiple of  $\pi_{j''}$ (as sizes are divisible).
Since the difference remains positive,
there are no jobs of the size of $j''$ or smaller that remain for densities
at most $d'$ (and division $k$), and the remaining smallest job
has size at least twice the size of $j''$.
Let $\bar{\jmath}$ be the
new smallest job of density at most $d'$ and division $k$. The
previous excess is at most $2\pi_{j''} \leq 2 \cdot \pi_{\bar{\jmath}}/2=\pi_{\bar{\jmath}}$, and the new
excess is at most $\pi_{\bar{\jmath}}$, proving the claim in this
case too. As there are at most $\frac 1{\del}$ divisions, the
total excess over all divisions is at most $\frac{1}{\del}
\cdot  2\del^{11}s_{\ell}(1+\del)^i \leq \del^{9}
s_{\ell}(1+\del)^i=\del^8 Z$, where $Z$ is the size of the
interval.

Finally, now that properties 1 and 2(a) hold,  we will perform a process where pairs of jobs of equal size and density are swapped. Note that this will not increase the cost or harm property 1; any two swapped jobs have exactly the same sizes and weights, so the time intervals allocated to jobs remain the same. Property 2(a) must hold without any relation to indices of jobs or release dates.
In this process, for any size and density, remove all jobs from their positions, sort them by non-decreasing release dates and by decreasing indices within each release date, and assign them into the original time slots according to the order of time intervals of these time slots. We argue that a job cannot be assigned before its release date. Consider the job whose position in the sorted list is $x$. Since all the time slots were previously used for similar jobs (of the same size and density), there are at least $x$ jobs whose release date is sufficiently early to run in the time slot that $x$ gets in the process. Since the job appears in position $x$ in the sorted list, its release date must be one of the smallest $x$ release dates. This ensures the remaining properties.
\end{proof}

We apply an additional transformation on release dates to obtain the instance $\tilde{A}$ from $A'$ where the release date of $j$ is denoted by $r_j$ (and satisfies $r_j \geq r'_j$). Other than modified release dates for some of the jobs, $\tilde{A}$ and $A'$ have the same machines and the same jobs. The instance $\tilde{A}$ may contain release dates that do not exist in $A'$, but all release dates are integer powers of $1+\del$ and the smallest release date will remain $(1+\del)^{\theta}$.
Thus, any (timely) solution for $\tilde{A}$ is a (timely) solution for $A'$ as well. In order to use a solution of $A'$ for $\tilde{A}$, one has to show that each job $j$ that has a starting time in $\J_{i,\ell}$ has $r_i \leq (1+\del)^i$ (while other aspects of feasibility follow from the feasibility for $A'$). In particular, we can use an organized schedule for $A''$ as a schedule for $A'$ and thus for $\tilde{A}$. In $\tilde{A}$, for every kind of jobs, if too many (in terms of the number or the total size) pending jobs exist at time $(1+\del)^i$, the release date of some of them is increased. This can be done when it is impossible to schedule all these jobs due to the lengths of intervals. We need to take into account jobs that have  starting times in an interval $\J_{i,\ell}$ even if their completion times are in later time intervals. Jobs with this property are called special. The total size of jobs that are not special will be at most the total lengths of time intervals. The process is applied separately and independently on every possible density.

We will define $\tilde{A}$ using a process that acts on increasing values of $i$ and at each step applies a modification on release dates of jobs whose current release date is $(1+\del)^i$.
For a given density $d$, initialize
$r_j=r'_j$ for every job $j$ of density $d$. For each possible size, create a list of preferences. In this list, jobs are ordered by non-decreasing release dates in $A'$, and within every release date, by decreasing indices.
In what follows, for any job $j$ we will refer to $r'_j$ as the initial release date of $j$, and to the
values $r_j$ as modified release dates. As these values will be changed during
a modification process that we define (they can possibly be changed a number of times for each job), when we discuss such a value, we will always refer to the current value even if it will be changed later. We apply the process for
$i=\theta,\theta+1,\ldots$. The stopping condition will be the
situation where for some $i$ there are no jobs whose modified release dates
are at least $(1+\del)^i$. If during the process for some $i$
there are no jobs with modified release date $(1+\del)^i$, but there exist
jobs with larger release dates (in this case these are both initial and modified release dates of those jobs), we skip this value of $i$ and move
to the next value. Recall that our goal is to increase some (modified) release dates of
jobs where these jobs will not be started before the next possible release date, and to break ties consistently, we do this for jobs that would not be started before the next release date in an organized schedule. Since we are interested in organized schedules that are, in particular, timely, jobs that are huge for a given interval should not be scheduled a start time in it.
For a given value of $i$, and a fixed density $d$, we consider jobs whose modified release date is $(1+\del)^i$.
Out of these jobs we will select a subset, and for every such $j$ in the subset, the current value of $r_j$ will remain $(1+\del)^i$ and will not change later. For every job $j'$ that is not  selected, we modify its modified release date into $r_{j'}=(1+\del)^{i+1}$. Initially, no job is selected.
For every time interval $\J_{i,\ell}$ (for some $1 \leq \ell \leq m$), for every size that is big (large or medium) for this interval and every density, select the unselected highest priority job $j$ such that $r_j=(1+\del)^i$.
For sizes that no such job exists obviously no job is selected. Additionally, for every size that is medium for this interval and every density, select an additional unselected $\frac{1}{\del^{10}}$ highest priority jobs of this size (selecting all such jobs if less than $\frac{1}{\del^{10}}$ such jobs exist). For every division $k$ and every density, consider the job sizes of this division that are small for this time interval in non-increasing order. For each size, unselected jobs of one size are selected according to priorities. Keep selecting unselected jobs according to priorities, moving to the next size if all jobs of one size have been selected, until the total size of selected jobs of division $k$ (and density $d$, that are small for $\J_{i,\ell}$) is at least  $s_{\ell} \del (1+\del)^i$ (and thus their total size is at most $s_{\ell} \del (1+\del)^i (1+\del^{10})$, since any such job is small for $\J_{i,\ell}$). If the total size of all these jobs is smaller than $s_{\ell} \del (1+\del)^i$, then all of them are picked.
For a given triple $d,i,\ell$, letting $Z=s_{\ell}\del(1+\del)^i$, the total size of selected jobs is as follows. There are at most $\log_{1+\del}\frac {1+\del}{\del^2}+1 \leq \frac{1}{\del^3}+2$ sizes of large jobs, each having a size of at most $\frac{1+\del}{\del^2}Z$. There are at most $\log_{1+\del}\frac 1{\del^{10}}+1 \leq \frac{1}{\del^{11}}+1$ sizes of medium jobs, each having a size of at most $Z$, and at most $\frac 1{\del}$ divisions. The total size of selected large jobs is at most $Z(\frac{1}{\del^3}+2)(\frac{1+\del}{\del^2})<\frac{Z}{\del^6}$. The total size of selected medium jobs is at most $(1+\frac{1}{\del^{10}})(\frac{1}{\del^{11}}+1)Z <\frac{Z}{\del^{22}}$. The total size of small jobs is at most $Z\cdot \frac 1{\del} (1+\del^{10})$.
The total size is therefore at most $\frac{Z}{\del^{23}}$.

\begin{lemma}
In an organized schedule for ${A''}$, every job $j$ is scheduled a starting time in a time interval that is no earlier than $r_j$. Thus, any organized solution for ${A''}$ is a valid timely solution for $\tilde{A}$ with the same cost.
\end{lemma}
\begin{proof}
Consider an organized schedule for ${A''}$ and assume by
contradiction that there exists a job $j'$ that is assigned a starting time in a time interval $\J_{i',\ell'}$ such that $r_{j'}\geq (1+\del)^{i'+1}$.
Let $\J_{i,\ell}$ be the first interval (according to the ordering of time intervals) that has a job $j$ with a starting time in it that was not selected for $\J_{i,\ell}$ or any earlier interval, such that either $j$ is big for this interval, and its priority is lower than jobs of the same size and density selected for $\J_{i,\ell}$, or $j$ is small for this time interval, and it has lower priority than small jobs of the same division and density selected for this time interval. Since $A''$ is a valid schedule, $r'_{j} \leq (1+\del)^{i}$, so it was possible to select $j$ if it is still unselected, and since $r_j>(1+\del)^i$, $j$ was not selected for this time interval or for earlier time intervals. Thus, in the process of  selection of jobs similar to $j$ for $\J_{i,\ell}$ (either jobs of the same size and density, if $j$ is big for this time interval, or jobs of the same division and density otherwise), the full number or total size of such jobs was selected (since an unselected job remains). If $j$ is large for $\J_{i,\ell}$, then there is another job $\tilde{\jmath}$ of the same size and density selected for $\J_{i,\ell}$, where $\tilde{\jmath}$ has higher priority.
As at most one such job can have a starting time in this time interval, $\tilde{\jmath}$ must have a starting time in another interval $\J_{i_1,\ell_1}$.  If $j$ is medium for $\J_{i,\ell}$, then there are $\frac{1}{\del^{10}}+1$ other jobs of the same size and density selected for $\J_{i,\ell}$, of higher priority. At most $\frac{1}{\del^{10}}+1$ jobs of this size and density can have starting times in $\J_{i,\ell}$, and therefore there is among these jobs a job $\tilde{\jmath}$ with higher priority that has a starting time in another interval  $\J_{i_1,\ell_1}$ .
Consider the case that $j$ is small for  $\J_{i,\ell}$. All jobs that are small for this interval and have starting times in it also have completion times in it, so their total size does not exceed the length of the time interval. Since there are jobs of the same division and density whose total size is at least the length of the interval that were selected for this time interval, at least one such job  $\tilde{\jmath}$ is assigned a starting time in another time interval, and  $\tilde{\jmath}$ has a higher priority than $j$ in $\J_{i,\ell}$. In this case we also let $\J_{i_1,\ell_1}$ denote this time interval where the  starting time of $\tilde{\jmath}$ is.
If $\J_{i_1,\ell_1}$ is an earlier interval. Since $\tilde{\jmath}$ is selected for $\J_{i,\ell}$ while $r'_{\jmath} \leq (1+\del)^i$ (since its starting time is in $\J_{i_1,\ell_1}$ and the schedule is valid for $A''$), we find that an earlier interval already has a job whose priority is too low (no matter whether this is its priority as a big job for that interval or a small job for that interval), contradicting the minimality of $\J_{i,\ell}$. If this is a later interval, then this contradicts the fact that the schedule is organized, as a job of a higher priority in $\J_{i,\ell}$ than $j$ has a starting time in a later time interval while $j$ has a starting time in $\J_{i,\ell}$.
\end{proof}

\begin{lemma}
\label{corYDensities}
If the input belongs to at most $\hat{y}$ densities (where $\hat{y}$ is a function of $\del$), then the jobs of modified release date $r$ can be all scheduled during $[t,t')$, where $t \geq r$, and $t' \leq t+ \frac{r\hat{y}}{\del^{22}}$.

Given a set of $\hat{y}$ densities, the set of jobs with modified release date at most $r$ and these densities, can be all scheduled during a time interval $[t,t')$ where $t>r$ and $t' \leq t+\frac{r\hat{y}(1+\del)}{\del^{23}}$.
\end{lemma}
\begin{proof}
Let $t \geq r$, and schedule the jobs whose modified release date is $r$ starting at time $t$, such that the jobs scheduled on machine $\ell$ are those that were selected for $\J_{i,\ell}$ (where $r=(1+\del)^i$). The total size of the jobs of one density is at most $\frac{Z}{\del^{23}}=\frac{s_{\ell}(1+\del)^i}{\del^{22}}=\frac{s_{\ell}r}{\del^{22}}$ so the jobs of $\hat{y}$ densities will be completed (on a machine of speed $s_{\ell}$) within a time of $\frac{r\cdot \hat{y}}{\del^{22}}$.

To prove the second claim, note that the total size of jobs of one density and modified release date at most $r$ that will run on machine $\ell$ is at most $\sum_{r'=1}^r \frac{s_{\ell} r}{\del^{22}}\leq \frac{s_{\ell}r'}{\del^{22}}\cdot \frac{1+\del}{\del}$.
\end{proof}

\begin{corollary}\label{job_shifting_cor}
Consider an instance $\tilde{A}_s$ obtained from $\tilde{A}$ by keeping jobs belonging to only $\hat{y}$ given densities, and that are released by time $R$ (for a fixed $R>0$). There exists an optimal timely schedule $\O_s$ for $\tilde{A}_s$, no job is completed after time $\frac{\hat{y}R}{\del^{25}}$.
\end{corollary}
\begin{proof}
Consider an optimal schedule $\O'_s$, and let $R'=\frac{\hat{y}R}{\del^{24}}$. Let $R' <  T \leq  (1+\del)R'$ be such that $T=(1+\del)^{\iota}$ for an integer $\iota$. Consider the jobs whose completion times are above $T$ (and thus in their pseudo-costs, their weights are multiplied by at least $T(1+\del)$). By Lemma \ref{corYDensities}, it is possible to schedule these jobs during $[T,T+\frac{R\hat{y}}{\del^{23}})$. We have $\frac{R\hat{y}}{\del^{23}}=\del R' < \del T$, and therefore it is possible to remove all these jobs from $\O'_s$, and schedule them within the time interval $[T,(1+\del)T)$. As all jobs that are still running at time $T$ are removed, before the jobs are assigned again, the time interval $[T,(1+\del)T)$ is idle on all machines. No reassigned jobs has an increased cost, and therefore the resulting schedule, $\O_s$ is optimal as well. We are done since $(1+\del)T \leq (1+\del)^2 R'=(1+\del)^2 \frac{\hat{y}R}{\del^{24}} < \frac{\hat{y}R}{\del^{25}}$.
\end{proof}

\subsection{Applying the shifting technique on the release dates of jobs}
We let $\alpha = \frac{\hat{y}}{\del^{34}}$ where $\hat{y}> \frac{1}{\del^{12}}$ is an integral function of $\del$ that will be chosen later.
For $k=0,1,\ldots ,\frac{\alpha}{\del}-1$, and $i\geq 0$ let $Q_{i,k}= \theta+ \lceil \log_{1+\del} \alpha^{k+\frac{i\alpha}{\del}} \rceil$, where we interpret $Q_{i,\frac{\alpha}{\del}}$ as $Q_{i+1,0}$, and $Q_{-1,k}=\theta$ for all $k$.
Let $\Psi_{i,k}=(1+\del)^{Q_{i,k}}$ for $i \geq -1$, $0\leq k \leq\frac{\alpha}{\delta}$.

We define a new instance $A_k$ based on $\tilde{A}$ as follows.  The instance $A_k$ has the same set of $m$ machines and $n$ jobs, but we will modify the release dates of some jobs.  For every job $j$, for which there exists an integer $i$ such that the release date of $j$ (in $\tilde{A}$) is in the time frame $D_{i,k}=\left[\Psi_{i,k},\Psi_{i,k+1}\right)$, we increase the release date of $j$ to be $\Psi_{i,k+1}$. Let $|D_{i,k}|=\Psi_{i,k+1}-\Psi_{i,k}$.

\begin{lemma}
Given a timely optimal schedule $\opt_{\tilde{A}}$ for the instance $\tilde{A}$, there exists a value of $k$ such that $A_k$ has a feasible schedule of cost at most $(1+2\del)\opt_{\tilde{A}}$.
\end{lemma}
\begin{proof}
Given $\opt_{\tilde{A}}$ (with specific time slots assigned to every job), we define a partition of the jobs into the sets $J_{ik}$ consisting of all jobs whose starting times (in $\opt_{\tilde{A}}$) is within the time frame $D_{i,k}$.  We denote  the total cost of the jobs in $J_{ik}$ (in the schedule $\opt_{\tilde{A}}$) by $T_{i,k}$.  Thus, $\opt_{\tilde{A}}=\sum_{i,k} T_{i,k}$.

For every value of $k$, we create a feasible schedule $\sol_k$ for the instance $A_k$ by modifying $\opt_{\tilde{A}}$ in the following way.  For $i=0,1,\ldots$, we add idle time of length $|D_{i,k}|$ for every machine just before the first job of $J_{ik}$ starts on that machine (shifting all further jobs to a later time by an additive factor of $|D_{i,k}|$).  Observe that this is indeed a feasible schedule for $A_k$ as for every job whose release date was increased to $\Psi_{i,k+1}$, its processing in $\opt_{\tilde{A}}$ starts no earlier than $\Psi_{i,k}$, and we added an idle time of length at least $|D_{i,k}|=\Psi_{i,k+1}-\Psi_{i,k}$, before its starting time, and thus the resulting starting time of the schedule is at least the new release date of the job.

We bound the difference $\sum_{k=0}^{\frac{\alpha}{\delta}-1} (\sol_k(A_k)-OPT_{\tilde{A}}(\tilde{A}))$. Assume that the starting time of job $j$ in $\opt_{\tilde{A}}$ is in $D_{i',k'}$. Then in all solutions $\sol_k$ (for $0\leq k \leq \frac{\alpha}{\delta}-1$) together, the total length of idle time that was added before the starting time of $j$ is the sum of values $|D_{i,k}|$ such that $i<i'$ or $i=i'$ and $k\leq k'$. This total length is at most $(1+\del)^{Q_{i,k+1}}\leq (1+\del)^{\theta+1}\cdot \alpha^{k+1+\frac{i\alpha}{\del}}$, while the completion time of $j$ in $\opt_{\tilde{A}}$ is at least $(1+\del)^{Q_{i,k}} \geq (1+\del)^{\theta}\cdot \alpha^{k+\frac{i\alpha}{\del}}$. Thus, $\sum_{k=0}^{\frac{\alpha}{\delta}-1} (\sol_k(A_k)-OPT_{\tilde{A}}(\tilde{A})) \leq (1+\del)\alpha OPT_{\tilde{A}}(\tilde{A})$, and  $\sum_{k=0}^{\frac{\alpha}{\delta}-1} \sol_k(A_k) \leq (\frac{\alpha}{\delta}+(1+\del)\alpha) OPT_{\tilde{A}}(\tilde{A}))$. By the pigeonhole principle, the value of $k$ for which $\sol_k(A_k)$ is minimal satisfies $\sol_k(A_k) \leq (1+\del(1+\del))OPT_{\tilde{A}}(\tilde{A})<(1+2\del)OPT_{\tilde{A}}(\tilde{A})$.
\end{proof}

For a fixed value of $k$, we partition the job set of $A_k$ into the following subsets. For every $i\geq 0$ we let $A_{ik}$ be the instance with  job set consisting  of the jobs whose release date (in $A_k$) is in the time interval $[\Psi_{i-1,k+1},\Psi_{i,k})$.   The following lemma follows using the fact that an optimal solution for $A_k$ gives a feasible solution for each $A_{ik}$. Let $\opt_k$ be an optimal cost of a solution to $A_k$, and let $\opt_{ik}$ be the cost of an optimal solution for $A_{ik}$.

\begin{lemma}
We have $\sum_i \opt_{ik} \leq \opt_k$.
\end{lemma}

For every value of $i,k$, all jobs in the instance $A_{ik}$ have release dates in $[\Psi_{i-1,k+1}, \frac{\Psi_{i,k}}{1+\del}]$. Recall that no job is released during the time $(\frac{\Psi_{i,k}}{1+\del},\Psi_{i,k+1})$.
We have $(1+\del)^{Q_{i,k}-1}\leq (1+\del)(1+\del)^{\theta-1}\alpha^{k+\frac{i\alpha}{\del}}$ and $(1+\del)^{Q_{i-1,k+1}}\geq (1+\del)^{\theta}\alpha^{k+1+\frac{(i-1)\alpha}{\del}} \geq (1+\del)^{Q_{i,k}-1} / \alpha^{\frac{\alpha}{\del}-1}$, and thus the number of distinct release dates in $A_{ik}$ is at most $(\frac{\alpha}{\del}-1) \cdot (\log_{1+\del} \alpha +1) \leq (\frac{\alpha}{\del}-1)(\frac{\alpha}{\del}+ 1) < \left( \frac{\alpha}{\del} \right)^2$.

In the next section we will show the existence of a polynomial time algorithm that find an approximate solution for the input $A_{ik}$, that approximates an optimal solution within a multiplicative factor of $1+\kappa\del$ for some constant value of $\kappa$, such that the output of the algorithm, denoted by $\sol_{ik}$, satisfies the following structural properties:
\begin{property}\label{prop1}
Let $\Psi$ denote  an upper bound on the maximum release date of the instance that is an integer power of $1+\del$, and let $\Psi'$ be an integer power of $1+\del$ in the interval $(\Psi \cdot \frac{\alpha}{1+\del}, \Psi \cdot \alpha(1+\del))$.
The solution $\sol_{ik}$ runs the jobs that start after $\Psi$ in non-increasing order of density, and
for every machine of speed $\sigma$ in $\sol_{ik}$ and two jobs $j,j'$ running on that machine whose completion times $C_j,C_{j'}$ satisfy $C_j-\frac{p_j}{\sigma}\geq \frac{\Psi'}{2}$, $C_{j'} \geq C_j + \Psi \cdot \frac{\hat{y}}{\del^{28}}$, then the density of $j'$ is at most the density of $j$ times $\frac{1}{(1+\del)^{\hat{y}}}$.
\end{property}

\begin{property}\label{no_large_prop}
Let $\Psi$ denote  an upper bound on the maximum release date of the instance that is an integer power of $1+\del$, and let $\Psi'$ be an integer power of $1+\del$ in the interval $(\Psi \cdot \frac{\alpha}{1+\del}, \Psi \cdot \alpha(1+\del))$.
For every machine of speed $\sigma$, there is no job $j$ assigned to this machine in $\sol_{ik}$ of size larger than $\Psi \cdot \frac{\hat{y}}{\del^{25}} \cdot \sigma$.
\end{property}

Note that the last property depends only on the allocation of jobs to machines such that for each job there is an upper bound on the index of machine that can process it.  Thus, the following holds.
\begin{remark}\label{rmk_large_prop} If $\tilde{sol}_{ik}$ satisfies Property \ref{no_large_prop} with some valid values of $\Psi$ and $\Psi'$, and we construct another schedule $\hat{sol}_{ik}$ by changing the starting times of some of the jobs and moving some jobs to lower index machines (with respect to their assigned machine in $\tilde{sol}_{ik}$), then $\hat{sol}_{ik}$ also satisfies Property \ref{no_large_prop} with the same values of $\Psi$ and $\Psi'$.
\end{remark}

In the later parts of the analysis we will use also another auxiliary property as follows.
\begin{property}\label{prop2}
Let $\Psi$ denote  an upper bound on the maximum release date of the instance that is an integer power of $1+\del$, and let $\Psi'$ be an integer power of $1+\del$ in the interval $(\Psi \cdot \frac{\alpha}{1+\del}, \Psi \cdot \alpha(1+\del))$.
The solution
$\sol_{ik}$ runs the jobs that start after $\Psi$ in non-increasing order of density, and
for every machine of speed $\sigma$ in $\sol_{ik}$ and two jobs $j,j'$ running on that machine whose completion times $C_j,C_{j'}$ satisfy $C_j-\frac{p_j}{\sigma}\geq \frac{\Psi'}{4}$, $C_{j'} \geq C_j + \Psi \cdot \frac{\hat{y}}{\del^{27}}$, then the density of $j'$ is at most the density of $j$ times $\frac{1}{(1+\del)^{\hat{y}}}$.
\end{property}

We further denote by $\sol_{ik}$ the cost of the solution $\sol_{ik}$.  Then, using $\sol_{ik}\leq (1+\kappa \del) \opt_{ik}$ for every $i,k$, we get that $\sum_i \sol_{ik} \leq (1+\kappa\del)\sum_i \opt_{ik} \leq (1+\kappa\del) \opt_k\leq (1+\kappa\del)(1+2\del) \opt_{\tilde{A}}$.

\begin{lemma}
Consider a particular machine and timely schedules $\sol_{ik}$ for every value of $i$ that satisfy properties \ref{prop1} and \ref{no_large_prop} with $\Psi=\Psi_{i,k}$.  Let $\s_i$ be the set of jobs assigned to this machine by  $\sol_{ik}$.  Then, there is a polynomial time algorithm that constructs  (for each machine) a feasible schedule of the jobs in $\bigcup_{i} \s_i$ on this machine with a total cost of at most $(1+\del)^2 \sum_i \sol_{ik}$.
\end{lemma}

\begin{proof}
We consider one specific machine and without loss of generality we assume that the speed of the machine is $1$. We first modify $\sol_{ik}$ by removing idle time periods after $\frac{\Psi_{i,k}}{1+\del}$ and denote the resulting schedule by $\sol''_{ik}$. Observe that $\sol''_{ik}$ satisfies properties \ref{prop1} and \ref{no_large_prop} with $\Psi=\Psi_{i,k}$ and $\Psi'=\Psi_{i,k+1}$ where Property \ref{prop1} holds because if two jobs satisfy the condition (on the difference between their completion times) for $\sol''_{ik}$ then they also satisfy the condition for $\sol_{ik}$, and the claim regarding Property \ref{no_large_prop} holds by Remark \ref{rmk_large_prop}. Then, consider the schedule $\sol'_{ik}$ that is the schedule obtained from $\sol''_{ik}$ by time stretching by a factor of $1+\del$.  Using Corollary \ref{timestretch}, we conclude that for every time interval, if $\sol'_{ik}$ contains a gap, then it has idle time of at least $\del^3$ times the length of the interval.  Moreover, the total costs of the solutions $\sol'_{ik}$ and $\sol_{ik}$ satisfy $\sum_i \sol'_{ik} \leq (1+\del)\sum_{i} \sol_{ik}$.

We next modify the schedule $\sol'_{ik}$ for every $i$ as follows. We start with constructing an initial schedule of some of the jobs.  For every $i$ we say that the time frame $[\Psi_{i-1,k+1},\Psi_{i,k+1})$ {\it belongs} to $i$, and
   the jobs of $A_{ik}$ with completion times (according to $\sol'_{ik}$) within the time frame that belongs to $i$ are scheduled according to the schedule $\sol'_{ik}$.  Due to the release dates constraints, their starting times are also during this time frame, and the other jobs of $A_{ik}$ have later completion times in $\sol'_{ik}$, and we shortly describe the schedule of these jobs. This completes the description of the initial schedule and the subset of jobs that are scheduled according to it.
    Given the initial schedule we define gaps of this schedule as follows.  For every time interval $\J_{\iota,\ell}$ that belongs to the time frame of $i$, we say that a gap of $\sol'_{ik}$ during $\J_{\iota,\ell}$ is a gap of the initial schedule.

   Consider a fixed value of $i$ for which the instance $A_{ik}$ has jobs that are not scheduled in the initial schedule. That is, there is at least one job of $A_{ik}$ whose completion time in $\sol'_{ik}$ is after time $\Psi_{i,k+1}$. Beginning with the staring time in $\sol''_{ik}$  $\mu$ of the first such job,  define subsets of $\s_i$ as follows.  For any integer $\ell \geq 1$, let $\s_{i,\ell}$ be the set of jobs of $A_{ik}$ whose starting time in $\sol''_{ik}$ is in the time frame $[\mu + (\ell-1) \Psi_{i,k+1} , \mu +\ell\Psi_{i,k+1} )$.  Observe that by Property \ref{no_large_prop} which $\sol''_{ik}$ satisfies (and thus $\sol'_{ik}$ also satisfies using Remark \ref{rmk_large_prop}), running this set of jobs on the machine requires a time interval of at most $\Psi_{i,k+1}(1+\del)  \leq 2 \Psi_{i,k+1}$.

For every $i$, we allocate space to the job sets $\s_{i,\ell}$ as follows: For every time interval $[(1+\del)^t,(1+\del)^{t+1})$ between two consecutive integer powers of $1+\del$ for which the initial schedule  has a gap and $(1+\del)^t \geq  \frac{\Psi_{i,k+1}}{\del^5}$, we allocate a set of jobs $\s_{i,\ell}$ for one value of $\ell$ in an increasing order of $\ell$.  That is, the first such gap receives $\s_{i,1}$, the second receives $\s_{i,2}$, etc.  We apply this procedure for every value of $i$. This completes the description of the resulting schedule where all jobs are assigned.  Next, we establish the feasibility of the schedule and bound its cost.

First note that the total processing times (on this machine) of jobs which are allocated to one specific time interval $[(1+\del)^t,(1+\del)^{t+1})$ is at most $$\sum_{i: \Psi_{i,k+1}/(\del^5) \leq (1+\del)^t} 2\Psi_{i,k+1} .$$  Denote by $\hat{\imath}$ the maximum value of $i$ for which $\Psi_{i,k+1}/(\del^5) \leq (1+\del)^t$, then this total processing time is at most $3  \Psi_{\hat{\imath}, k+1}$ because $\Psi_{i,k+1} \leq \del\Psi_{i+1,k+1}$ and thus $\sum_{i:i\leq \hat{\imath}} \Psi_{i,k+1} \leq \frac{3}{2(1+\del)} \Psi_{\hat{\imath},k+1}$. The resulting schedule is feasible by Corollary \ref{timestretch}, because
 $3 \Psi_{\hat{\imath}, k+1}\leq 3\del^5(1+\del)^t\leq \del^4(1+\del)^t$
 is at most $\del^3$ times the length of the time interval $[(1+\del)^t,(1+\del)^{t+1})$, which is $\del(1+\del)^t$.

 We next bound the increase of the cost due to the delay of the jobs in $\s_{i,\ell}$ for all $\ell$ with respect to the cost of $\sol'_{ik}$.
We let $\s_{i,0}$ be the set of jobs whose starting times in $\sol'_{ik}$ is in the time interval $[\Psi_{i,k+1}/2, 3\Psi_{i,k+1}/4)$ (note that every job assigned to this machine has size at most $\del \Psi_{i,k+1}$ and thus $\s_{i,0}$ is non-empty).   We next argue that for every value of $\ell \geq 1$ the density of the maximum density job in $\s_{i,\ell}$ is at most $(1+\del)^{\hat{y}\ell}$ times the minimum density of a job in $\s_{i,0}$.  We assume that the jobs in $\s_{i,\ell}$ are ordered in a non-decreasing order of their densities.  For $\ell\geq 0$ we denote by $j_{\ell}$ the first job in $\s_{i,\ell}$.  For $\ell=1$, observe that the completion time of $j_0$ is at most $3\Psi_{i,k+1}/4+ \del \Psi_{i,k+1} \leq \frac{4}{5}  \Psi_{i,k+1}$, and the completion time of $j_1$ is at least $\Psi_{i,k+1} \geq \frac{4}{5}  \Psi_{i,k+1} + \Psi \cdot \frac{\hat{y}}{\del^{28}}$, and
the claim follows by Property \ref{prop1}.  Next, assume that the claim holds for $\ell-1$. To prove the claim for $\ell$, it suffices to show that the density of $j_{\ell}$ is at most $(1+\del)^{\hat{y}}$ times the density of $j_{\ell-1}$.  This last claim holds because the completion time (in $\sol''_{ik}$) of $j_{\ell-1}$ is at most $\mu + (\ell-1) \Psi_{i,k+1} + 2 \del  \Psi_{i,k+1}$, and the completion time of $j_{\ell}$ is at least $\mu + \ell \Psi_{i,k+1} > \mu + (\ell-1) \Psi_{i,k+1} + 2 \del  \Psi_{i,k+1} + \Psi_{i,k} \cdot \frac{\hat{y}}{\del^{28}}$ and the claim follows by Property \ref{prop1}.
Then, for every $\ell \geq 1$ and every pair of jobs $j\in \s_{i,0}$ and $j'\in \s_{i,\ell}$, we conclude that the density of $j'$ is at most the density of $j$ times $\frac{1}{(1+\del)^{\hat{y}\ell}}$.

   Since there are no release dates of jobs from $\s_i$ after $\Psi_{i,k}$ and since $\s_{i,1}\neq \emptyset$, the solution $\sol''_{ik}$ does not have idle time during the time interval $[\Psi_{i,k}, \Psi_{i,k+1})$, and thus the total processing times of the jobs in $\s_{i,0}$ is at least $\frac{\Psi_{i,k+1}}{4}-\del\Psi_{i,k+1} \geq \frac{\Psi_{i,k+1}}{5}$ due to a possible job completing in this time frame but starting earlier.  On the other hand the total size of the jobs in $\s_{i,\ell}$ is at most $(1+\del)\Psi_{i,k+1}$.

Moreover, note that the delay of the set of jobs $\s_{i,\ell}$ is by an additional (additive) time interval of length at most $\Psi_{i,k+1} \cdot \frac{1}{\del^{8\ell}}$ with respect to their schedule in $\sol'_{ik}$.

Let $d_i$ be the minimum density of a job in $\s_{i,0}$.
Denote by $c_{i,\ell}$ the additional cost of delaying the set of jobs $s_{i,\ell}$, and by $c_{i,0}$ the cost of $\s_{i,0}$ in the schedule $\sol'_{ik}$, then $c_{i,0} \geq \frac{\Psi_{i,k+1}}{5} \cdot d_i \cdot \frac{\Psi_{i,k+1}}{2}$. Denote by $t_{i,\ell}$ the total processing time of the jobs in $\s_{i,\ell}$, for $\ell\geq 0$.  Then, for $\hat{y}\geq \frac{1}{\del^{12}}$ we have $(1+\del)^{\hat{y}} \geq (1+\del)^{1/(\del^{12})} \geq \frac{1}{\del^{11}}$, and therefore $\frac{1}{\del^8(1+\del)^{\hat{y}}-1} \leq \del^{2}$. Thus, the following holds.
\begin{eqnarray*}
\sum_{\ell=1}^{\infty} c_{i,\ell} &\leq& \sum_{\ell=1}^{\infty} t_{i,\ell}\cdot d_i \cdot \frac{1}{(1+\del)^{\hat{y}\ell}} \cdot \frac{\Psi_{i,k+1}}{\del^{8\ell}} \\ &\leq& \sum_{\ell=1}^{\infty} (1+\del)\Psi_{i,k+1} \cdot d_i \cdot \frac{1}{(1+\del)^{\hat{y}\ell}} \cdot \frac{\Psi_{i,k+1}}{\del^{8\ell}}\\&\leq&
10 \cdot c_{i,0} \cdot (1+\del)\cdot  \sum_{\ell=1}^{\infty} \frac{1}{(1+\del)^{\hat{y}\ell}} \cdot \frac{1}{\del^{8\ell}}\\ &=& 10 \cdot c_{i,0} \cdot (1+\del)\cdot \frac{1}{\del^8(1+\del)^{\hat{y}}-1}\\
&\leq& 10 \cdot c_{i,0} \cdot (1+\del)\cdot \del^2 \leq \del c_{i,0} .
\end{eqnarray*}
\end{proof}

\begin{corollary}
If there exists an algorithm with time complexity $T(n,m,\frac{1}{\del})$ such that for a given pair $i$ and $k$, it constructs a schedule $\sol_{ik}$ satisfying properties \ref{prop1} and \ref{no_large_prop} for $\Psi=\Psi_{i,k}$ and $\Psi'=\Psi_{i,k+1}$ such that $\sol_{ik}\leq (1+\kappa \del)\opt_{ik}$, then there exists an algorithm with time complexity $n\frac{\alpha}{\del} T(n,m,\frac{1}{\del})$ that approximates $\tilde{A}$ within a factor of $(1+\kappa \del)(1+2\del)(1+\del)^2$.
\end{corollary}

\subsection{A guessing step on the structure of the optimal schedule}
In what follows, we consider one specific instance $\hat{A}=A_{i\tilde{\kappa}}$,
and we scale all release dates and job sizes so that the minimum
release date of a job in $\hat{A}$ is $1$, and the maximum release
date of a job in $\hat{A}$ is at most $\Psi=\frac{\Psi_{i\tilde{\kappa}}}{\Psi_{i-1,\tilde{\kappa}+1}}$ (and thus
$\Psi \leq \alpha^{\frac{\alpha}{\del}-1}$), and we let $\Psi'= \frac{\Psi_{i,\tilde{\kappa}+1}}{\Psi_{i-1,\tilde{\kappa}+1}}$.

Given a schedule $\sol$ for
$\hat{A}$, we define machine colors as follows.  The color of
machine $i$ in $\sol$ is the subset of the time intervals of the
form $[(1+\del)^t,(1+\del)^{t+1})$ such that $\sol$ has a starting
time or some idle time during this interval on machine $i$, and $t$ is
an integer in the interval $[0,\lceil \log_{1+\del}
\alpha^{\frac{\alpha}{\del}-1}\rceil+1]$.  Observe that the number
of possible colors for a machine is
$O((\alpha)^{\frac{\alpha}{\del^2}-1})$.  We
say that the color of machine $i$ is pink (or $i$ is a pink
machine) if its color is the set of all time intervals within
$[0,\lceil \log_{1+\del} \alpha^{\frac{\alpha}{\del}-1}\rceil+1]$.  In what follows, we assume that the number of machines in the input is at least $\frac{2}{\del^7}+3$.  Observe that if this assumption does not hold, then the PTAS of Afrati et al. \cite{Afrati+99} for a constant number of unrelated machines is actually an EPTAS and we use that algorithm instead.
The {\it palette} of schedule $\sol$ is the vector with
$\frac{1}{\del^7}+1$ components consisting of the colors of the
machines with smallest indices (that is, the fastest machines,
breaking ties in favor of smaller indices), where the $i$-th
component of the palette is the color of machine $i$.

Given a palette $\pi$ and an integer $\tau$, we denote by $(1+\del)^{\tau}\pi$ the palette defined as follows.  For $\tau=1$, let $\sol_{\pi}$ be a schedule with palette $\pi$, then $(1+\del)\pi$ is the palette of the shifted schedule of $\sol_{\pi}$, and for larger values of $\tau$, we define the palette by recursion using $(1+\del)^{\tau} \pi=(1+\del) (1+\del)^{\tau-1}\pi$.

\begin{claim}
The number of possibilities for palettes of schedules is $O\left((\left(\alpha\right)^{\frac{\alpha}{\del^2}-1})^{\frac{1}{\del^7}+1}\right)$.
\end{claim}

We next modify the optimal schedule and show that there is a near optimal schedule whose palette has a pink machine.
\begin{lemma}\label{pinkmachinelem}
Let $\sol'$ be a timely schedule for instance $\hat{A}$ whose palette is $\pi'$.  Then, there is another timely schedule $\sol$ (with palette $\pi$) that has a pink machine and such that $\sol(\hat{A})\leq (1+\del)^2 \sol'(\hat{A})$.  Moreover, if $\sol'$ satisfies with $\Psi$ and $\Psi'$ properties \ref{prop2} and \ref{no_large_prop}, then $\sol$ satisfies properties \ref{prop1} and \ref{no_large_prop} with the same values of $\Psi$ and $\Psi'$.
\end{lemma}
\begin{proof}
First, we consider the schedule $\sol''$ obtained from $\sol'$ by time stretching by a factor of $1+\del$.  Let $v$ be the index of a machine such that $2\leq v\leq \frac{1}{\del^7}+1$ and machine $v$ is assigned a set of jobs (in $\sol''$) of minimum total weight.  Then, the total weight of the jobs assigned to machine $v$ is at most $\del^7$ times the total weight of the jobs in $\hat{A}$.  We next modify the solution $\sol''$ by moving to the first machine all the jobs assigned to machine $v$ starting no later than $\Psi$  while increasing the completion time of each such job by a multiplicative factor of at most $\frac{1}{\del^6}$.  The schedule of all other jobs remains as it is in $\sol''$.  This will prove the claim where $\sol$ is the resulting schedule.

For every job $j$ that we move which is assigned to complete in $\sol''$ during the time interval $[(1+\del)^t,(1+\del)^{t+1})$ on machine $v$, we schedule $j$ on the last gap in $\sol''$ on machine $1$ such that the completion time of $j$ will be at most $\frac{1}{\del^6} \cdot (1+\del)^{t+1}$.  We note that if $\sol''$ has a time interval $[(1+\del)^{t'},(1+\del)^{t'+1})$ with a gap on machine $1$, then it is assigned jobs from machine $v$ whose completion times are in the time frame $[\del^6 (1+\del)^{t'+1}, \del^4 (1+\del)^{t'+1})$ (on machine $v$), and thus all these jobs are completed (in $\sol''$) by time $\del^4(1+\del)^{t'}$. Therefore, all these jobs fit into the length of the idle time period within the time interval $[(1+\del)^{t'},(1+\del)^{t'+1})$ (using Corollary \ref{timestretch} and the fact that the length of this interval is $\del(1+\del)^{t'}$).

We modify the solution $\sol$ such that for every machine, the jobs starting after the largest release date are processed according to a non-decreasing order of their density (this modification does not increase the cost of the solution).
To prove the last part of the claim, note that since $\sol'$ satisfies Property \ref{no_large_prop}, so does $\sol$ (using Remark \ref{rmk_large_prop}).  Next, consider a pair of jobs $j$ and $j'$ satisfying the conditions of Property \ref{prop2} in $\sol'$. If they are assigned in $\sol'$ to a machine that is not the first machine, then they clearly satisfy the conditions of Property \ref{prop1} and the claim holds.  Thus, consider the case in which $\sol'$ assigns both $j$ and $j'$ to the first machine.  Note that the total size of jobs that were moved from machine $v$ to machine $1$ is at most $\frac{\Psi}{\del}$ (because all of them started on machine $v$ not later than $\Psi$, and thus end not later than $\frac{\Psi}{\del}$ and the first machine is not slower than $v$). Denote by $start(j)$ and $start(j')$ the starting times of $j$ and $j'$ (in $\sol$), respectively.

We find that the last starting time of a job $\hat{\jmath}'$ that starts not later than $j'$ and it was originally assigned to machine $1$ in $\sol''$ is at least $start(j') - \frac{\Psi}{\del} - \Psi \cdot \frac{\hat{y}}{\del^{25}} \geq start(j') - 2\Psi \cdot \frac{\hat{y}}{\del^{25}}$, the first starting time of a job $\hat{\jmath}$ that starts not earlier than $j$ and it was originally assigned to machine $1$ in $\sol''$ is at most $start(j) + \frac{\Psi}{\del} + \Psi \cdot \frac{\hat{y}}{\del^{25}} \leq start(j)+ 2\Psi \cdot \frac{\hat{y}}{\del^{25}}$.  Since $j$ and $j'$ satisfy the condition on the difference between their starting times as in Property \ref{prop2} and $4\Psi \cdot \frac{\hat{y}}{\del^{25}} \leq \Psi \cdot \frac{\hat{y}}{\del^{28}}- \Psi \cdot \frac{\hat{y}}{\del^{27}}$, the starting times of $\hat{\jmath}$ and $\hat{\jmath}'$ satisfy the condition on their difference as in Property \ref{prop1}. The claim follows because the density of $j$ is not smaller than the density of $\hat{\jmath}$ (as the jobs are processed in non-increasing order of their densities) which is at least $(1+\del)^{\hat{y}}$ times the density of $\hat{\jmath}'$ (because $\sol''$ satisfies Property \ref{prop2}) and the density of $\hat{\jmath}'$ is at least the density of $j'$.
\end{proof}

The algorithm enumerates all possible palettes such that the palette has a pink machine.  We will denote by $(\hat{A},\pi,v)$ the instance obtained from $\hat{A}$ by requiring that there exists a non-negative integer $\tau$ for which the schedule has a palette $(1+\del)^{\tau}\pi$ and a pink machine $v$ (where $2\leq v \leq \frac{1}{\del^7}$ and $\pi$ is the palette of the near-optimal solution to $\hat{A}$ whose existence is established in Lemma \ref{pinkmachinelem}).  We conclude the following.

\begin{corollary}
 Assume that there is an algorithm with time complexity $\hat{T}(n,m,\frac{1}{\del})$ that approximates the instance $(\hat{A},\pi,v)$ with an approximation ratio $1+\hat{\kappa} \del$ such that the returned solution satisfies properties \ref{prop2} and \ref{no_large_prop}, then there is an algorithm with time complexity  $O\left(\left(\alpha\right)^{(\frac{\alpha}{\del^2}-1)(\frac{1}{\del^7}+1)}\cdot \hat{T}(n,m,\frac{1}{\del})\right)$ that approximates instance $\hat{A}$ with an approximation ratio $(1+\del)^2\cdot (1+\hat{\kappa}\del)$, and hence an algorithm with time complexity $O\left(n\frac{\alpha}{\del} \left( \frac{\alpha}{\del}\right)^{(\frac{\alpha}{\del}-1)(\frac{1}{\del^7}+1)}\cdot \hat{T}(n,m,\frac{1}{\del})\right)$ that approximates $\tilde{A}$ with an approximation ratio $(1+2\del)(1+\del)^4\cdot (1+\hat{\kappa}\del)$.
\end{corollary}

\subsection{Applying the shifting technique on the densities of jobs}
Let $\xi= \lceil \ell \log_{1+\del}\frac 1{\del} \rceil$ for a
fixed integer $\ell =25$ (and thus $\frac{1}{\del^{25}} = \frac{1}{\del^{\ell}} \leq
(1+\del)^{\xi} <
\frac{1+\del}{\del^{\ell}}<\frac{2}{\del^{\ell}}<\frac{1}{\del^{\ell+1}}= \frac{1}{\del^{26}}$. Since $(1+\del)^{\frac{1}{\del^2}}>\frac{1}{\del}$, we have $\xi \leq \frac{\ell}{\del^2}<\frac{1}{\del^3}$.
For an integer
$c \in \mathbb{Z}$, let $\Omega_c=\{c\xi+1,\ldots,(c+1)\xi\}$.

Let $0 \leq \zeta\leq \frac{1}{\del^{\ell+1}}-1$. We define the
instance $A_{\zeta}$ by modifying the weights of jobs in $\hat{A}$ (we still consider only solutions that have palette $(1+\del)^{\tau}\pi$ and pink machine $v$). For
a job $j$, if for some integer $\nabla$, $\log_{1+\del} \frac{w_j}{p_j}
\in \Omega_{\nabla/\del^{\ell+1}+\zeta}$, then $w^{\zeta}_j=w_j \cdot
(1+\del)^{\xi}$, and otherwise $w^{\zeta}_j=w_j$. In the first
case, we have $w^{\zeta}_j = (1+\del)^\xi w_j \leq
\frac{1+\del}{\del^{\ell}} w_j$. Let $\cO^{\zeta}$ be an optimal
solution for $A_{\zeta}$. As the set of jobs and machines is the
same in $\hat{A}$ and $A_{\zeta}$ (and the only property of a job that can differ in the two instances is its weight), the sets of feasible (timely) solutions
for the two instances are the same.  Next, we bound the increase
of the cost due to the transformation from $\hat{A}$ to $A_{\zeta}$. As a result, no job  $j \in A_{\zeta}$ has a density such that $\log_{1+\del} \frac{w^{\zeta}_j}{p_j} \in \Omega_{\nabla/\del^{\ell+1}+\zeta}$ for any integer $\nabla$. Any value $(1+\del)^{\beta}$ where $\beta \in   \Omega_{\nabla/\del^{\ell+1}+\zeta}$ for an integer $\nabla$ is called a forbidden density for $\zeta$.  The following claim is similar to Claim \ref{brshiftnorelease} used for the problem without release dates.

\begin{claim}
Given a solution $SOL$, any $0 \leq {\zeta} \leq \frac{1}{\del^{\ell+1}}-1$ satisfies $SOL(\hat{A}) \leq
SOL(A_{{\zeta}})$, and there exists a value $0 \leq \bar{\zeta} \leq \frac{1}{\del^{\ell+1}}-1$ such that we have $SOL(A_{\bar{\zeta}}) \leq  (1+2\del) SOL(\hat{A})$. Additionally, there exists a value $0 \leq \zeta' \leq \frac{1}{\del^{\ell+1}}-1$
such that $\cO'(\hat{A}) \leq \cO^{\zeta'}(A_{\zeta'})  \leq  (1+2\del) \cO'(\hat{A})$, and if a solution $SOL_1$ satisfies $SOL_1(A_{\zeta'}) \leq (1+k'\del))\cO^{\zeta'}(A_{\zeta'})$ for some $k'>0$, then
$SOL_1(\hat{A}) \leq (1+(2k'+2)\del) \cO'(\hat{A})$.
\end{claim}
\begin{proof}
We start with the first claim. Since for any $\zeta$, the weight
of a job in $A_{\zeta}$ is no smaller than the weight of the
corresponding job in $\hat{A}$, $SOL(\hat{A})\leq SOL(A_{\zeta})$ holds for any $\zeta$.
Seeing
$SOL$ as a solution for $\hat{A}$ (that is, $SOL$ is an assignment of a machine and a completion time for each job), let $SOL_{\zeta}$ denote the total
weighted completion time in $SOL$ (for $\hat{A}$) of jobs for which
$w_j\neq w^{\zeta}_j$, that is, the contribution of these jobs to
the objective function value for the instance $\hat{A}$. We have
$SOL(\hat{A})=\sum_{\zeta=0}^{\frac 1{\del^{\ell+1}}-1} SOL_{\zeta}$,
and $SOL(A_{\zeta})\leq \sum_{0 \leq \eta \leq \frac
1{\del^{\ell+1}}-1, \eta \neq \zeta} SOL_{\eta}+
(1+\del)^{\xi}SOL_{\zeta} \leq
SOL(\hat{A})+(1+\del)^{\xi}SOL_{\zeta}$.

Let $\bar{\zeta}$ be such that
$SOL_{\bar{\zeta}}$ is minimal. Then, $SOL_{\bar{\zeta}} \leq
{\del^{\ell+1}}SOL(\hat{A})$. We get $SOL(A_{\bar{\zeta}}) \leq
(1+\del^{\ell+1}(1+\del)^{\xi}) SOL(\hat{A})  \leq (1+2\del)SOL(\hat{A}) $.

The second part will follow from the first one.  Let $SOL'=\cO'$. By the second claim of the first part, there exists a value $\zeta'$ such that $\cO'(A_{\zeta'}) \leq (1+2\del)\cO'(\hat{A})$. Since $\cO'$ and
$\cO^{\zeta'}$ are optimal solutions for $\hat{A}$ and $A_{\zeta'}$
respectively, we have $\cO'(\hat{A}) \leq \cO^{\zeta'}(\hat{A})$ and
$\cO^{\zeta'}(A_{\zeta'}) \leq \cO'(A_{\zeta'})$.
Letting
$SOL''=\cO^{\zeta'}$ we get (using the first part of the first claim) $\cO^{\zeta'}(\hat{A}) \leq
\cO^{\zeta'}(A_{\zeta'})$, which proves $\cO'(\hat{A}) \leq \cO^{\zeta'}(A_{\zeta'})  \leq  (1+2\del) \cO'(\hat{A})$. We get $SOL_1(\hat{A}) \leq SOL_1(A_{\zeta'})\leq (1+k'\del)O^{\zeta'}(A_{\zeta'}) \leq (1+2\del)(1+k'\del) \cO'(\hat{A}) = (1+(k'+2)\del+2k'\del^2) \cO'(\hat{A})\leq (1+(2k'+2)\del) \cO'(\hat{A})$.
\end{proof}

Thus, we apply the following algorithm for every value of $\zeta$.  Namely, it suffices to show how to approximate the optimal solution for $A_{\zeta}$ with palette $\pi$ and pink machine $v$ where the algorithm results in solutions with palette $(1+\del)^{\tau}\pi$ for a fixed  positive integer value of $\tau$.  A solution of minimum cost (as a solution for $\hat{A}$) among all the solutions obtained for different values of $\zeta$ will be the output of the algorithm for approximating $\hat{A}$.  In the next section we will show an algorithm that receives an input consisting of a subset of the jobs of $A_{\zeta}$ where the ratio between the maximum density and the minimum density is at most $(1+\del)^y$ where $y=\frac{\xi}{\del^{\ell+1}}-\xi-1$ (and we let $\hat{y}=\frac{\xi}{\del^{\ell+1}}$ denote the number of distinct  densities of these jobs in the instance $\hat{A}$), and outputs a solution of cost at most $(1+\del)^4$ times the cost of an optimal solution for this instance (recall that these are pseudo-costs).  We will ensure that these solutions have palette $(1+\del)^2\pi$ and pink machine $v$.  Observe that this solution satisfies that every machine completes its processing no later than
  $\frac{\hat{y}\alpha^{\frac{\alpha}{\del}-1}}{\del^{25}}$ where $\alpha^{\frac{\alpha}{\del}-1}$ is the maximum release date.
  In what follows, we show that this procedure can be used to approximate an optimal solution for $A_{\zeta}$ as well.

Given a specific value of $\zeta$, we partition the job set of $A_{\zeta}$ into subsets such that the ratio between the maximum density and the minimum density of jobs of a common subset is upper bounded by $(1+\del)^y$ and for every two jobs $j$ and $j'$ of distinct subsets such that the density of $j$ is larger than the density of $j'$, we have that the density of $j$ is at least $(1+\del)^{\xi}$ times the density of $j'$.  Formally we define instances $\ak$ for every integer value of $k$ in which the set of machines is the same as in $A_{\zeta}$ and the job set of $\ak$ is the subset of jobs of $A_{\zeta}$ with densities in the interval $[(1+\del)^{\left( k/\del^{\ell+1}+\zeta+1\right) \xi+1}, (1+\del)^{\left( (k+1)/\del^{\ell+1}+\zeta\right) \xi}]$.  We will use the following characterization of the partition of the job set into the job sets of the instances $\ak$.
\begin{property}
Consider a pair of jobs $j$ and $j'$ for which (in $A_{\zeta}$) the density of $j$ is at least the density of $j'$.  Then, if $j$ and $j'$ are jobs of a common instance $\ak$, then the density of $j$ is at most $(1+\del)^y$ times the density of $j'$.  However, if $j$ belongs to the job set of $\ak$ and $j'$ belongs to the job set of $\akk$ such that $k>k'$, then the density of $j$ is at least $(1+\del)^{\xi(k-k')}$ times that of $j'$.
\end{property}

We denote by $\solk$ an approximated solution for $\ak$ such that $\solk$ has palette $(1+\del)^2\pi$ and pink machine $v$.  We will show in the next section how to compute such a schedule.  Next we explain how to combine the different solutions (one for each integer value of $k$ such that the job set of $\ak$ is non-empty).  The goal of the next step is to avoid a situation in which there is a machine $u$ such that in a specific time is assigned a very high density job in one of the solutions $\solk$ that is very short, and in another solution it is assigned a very long job that has low density.  Combining two such schedules is problematic, and we will avoid this (this is also the reason for introducing the notion of palette).  To avoid this situation, we introduce the following definition.

\begin{definition}
Consider a schedule $\sol$.  A time interval $\J_{t,u}$ where $(1+\del)^t \leq \alpha^{\frac{\alpha}{\del}-1}$  is called sparse for $\sol$, if there is a job with a starting time in $\J_{t,u}$ but the total size of jobs with starting times in $\J_{t,u}$ is at most $s_u \cdot \del^5(1+\del)^t$ (where $s_u$ is the speed of $u$).  A machine $u$ is called sparse, if there is a time interval on machine $u$ which is sparse.
\end{definition}

\begin{lemma}\label{nosparselemma}
Consider a schedule $\solk$ for the instance $\ak$ with a palette $(1+\del)^2\pi$ and pink machine $v$.  Then, there is another schedule $\solka$ for instance $\ak$ with palette $\pi'=(1+\del)^4\pi$  and pink machine $v$ such that $\solka(\ak) \leq (1+\del)\solk(\ak)$, $\solka$ does not have a sparse machine of index at least $\frac{1}{\del^7} +2$, and for every time interval $[(1+\del)^t,(1+\del)^{t+1})$, $\solka$ has at most one machine for which this time interval is sparse.  Moreover, there exists a polynomial time algorithm that constructs $\solka$ from $\solk$.
\end{lemma}
\begin{proof}
We first create a solution $\solkb$ from $\solk$ by time stretching by a factor of $1+\del$ (this ensures that any sparse interval has a gap) and then applying the following process.  As long as there exists a time interval $I_t=[(1+\del)^t,(1+\del)^{t+1})$ that is sparse (even after some modifications in earlier iterations of the algorithm have been applied) on two machines $u_1$ and $u_2$ such that $u_2> u_1$, then we move the set of jobs that starts during $\J_{t,u_2}$ to run on machine $u_1$ starting immediately after the completion of the last job that is originally processed in $\J_{t,u_1}$ (that is, the last job that starts and completes its processing during this time interval).  Since in the solution obtained by time stretching there is a gap of length at least $s_{u_1}\cdot \del^4(1+\del)^t$ during $\J_{t,u_1}$ (due to sparseness, even if a part of the gap was already used to receive jobs in an earlier iteration, the gap is more than half empty, and can accommodate an additional size of at least $s_{u_1}\cdot \del^5(1+\del)^t$), and since $u_1$ is not slower than $u_2$, we can complete processing all jobs that start on $u_2$ during this time interval (that is, this modification does not increase the pseudo-cost of the solution obtained from $\solk$ by time stretching).  We denote the resulting schedule by $\solkb$ and observe that $\solkb(\ak) \leq (1+\del)\solk(\ak)$.

We next apply once again time stretching by a factor of $1+\del$.  We observe that machine $v$ is pink and its schedule is obtained from the final $\solkb$ by time stretching by a factor of $1+\del$.  Thus, for every time interval $\J_{t,v}$, there is a gap of length at least $s_v \cdot \del^4(1+\del)^t$.  For every value of $t$, if there exists a machine $u$ for which $\J_{t,u}$ is sparse in $\solkb$ and $u \geq  \frac{1}{\del^7} +2$, then we move the set of jobs that start during $\J_{t,u}$ to run on machine $v$ during $\J_{t,v}$ (they will all be completed within a time interval which is shorter than the length of the gap during $\J_{t,v}$ because $v$ is not slower than $u$).  In the resulting schedule denoted by $\solka$ there are no sparse machines of index at least $\frac{1}{\del^7} +2$, the palette of $\solka$ is  $(1+\del)^4\pi$, and $\solka(\ak) \leq (1+\del)^2\solk(\ak)$.
\end{proof}

Observe that we can always assume that solution $\solka$
satisfies that the completion time of any machine is at most
$\alpha^{\frac{\alpha}{\del}-1}\cdot \frac{\hat{y}}{\del^{25}}$
(by Corollary \ref{job_shifting_cor}). For every value of $k$, we
let $\solkab$ be the solution obtained from $\solka$ by time
stretching by a factor of $1+\del$, and thus $\solkab$ has palette
$(1+\del)^5\pi$.  We modify the schedule $\solkab$ by first changing the starting times of the jobs whose starting time and completion time belong to a common time interval so as delaying the processing of these jobs as much as possible.  We apply an additional modification in which if a time interval $\J_{i,u}$ starting strictly after the maximum release date of a job is sparse, then its preceding interval (on the same machine) has a total processing time of at least $6\del^5(1+\del)^{i-1}$ (otherwise all the jobs whose starting times are in $\J_{i,u}$ are moved to start in $\J_{i-1,u}$, and we repeat the process as long as there exists such a sparse but non-empty interval starting strictly after the maximum release date).

We consider the solutions $\solkab$ for all values of $k$ on one
specific machine $u$ and without loss of generality assume that
$s_u=1$. Recall that $\solkab$ was obtained from a timely schedule by time stretching by a factor of $1+\del$, and thus for every time interval that is not contained in the reserved period of a job in $\ak$, there is a gap of length at least $\del^3$ times the length of the interval.  Throughout the process of creating a combined solution,
 we guarantee that $\widehat{\sol}$ is a feasible schedule of
all the jobs which the solutions $\solkkab$ schedule on machine
$u$ for values of $k'$ which are at least the current value of $k$
(which is the index of the iteration).

Based on the palette, i.e., on the color of $u$ if $u\leq \frac{1}{\del^7}+1$, we conclude that
there is no interval $\J_{i,u}$ for which there is an index $k$
for which $\solkab$ assigns jobs to a sparse interval on
$\J_{i,u}$ and another index $k'$ for which $\solkkab$ assigns a job $j$ a reserved starting time
smaller than $(1+\del)^i$ and a reserved completion time at
least $(1+\del)^{i+1}$. Thus, for every interval $\J_{i,u}$, we have
one of the following two cases: either there might be indices of
$k$ for which the interval $\J_{i,u}$ is sparse for $\solkab$, and
in this case for every value of $k'$, $\solkkab$ has a gap during
$\J_{i,u}$, or there are no such sparse intervals in $\J_{i,u}$.
Note that if $u\leq \frac{1}{\del^7}+1$ and $(1+\del)^i$ is strictly larger than the maximum release date, then we do not know which case holds before processing the solutions and observing for the first time either a sparse interval or a job whose reserved period contains $\J_{i,u}$.
If it is possible that the first case holds (i.e., it is not forbidden by the color of $u$), we allocate space within the gap of $\J_{i,u}$ for sparse intervals of
$\J_{i,u}$ (of different solutions), that we call {\it sparse-gap}, whose size is set to
$2\del^5(1+\del)^i$, and we ensure that the total size of jobs assigned to
$\J_{i,u}$ by $\widehat{\sol}$ and were assigned to $\J_{i,u}$ as
part of sparse intervals by any solution $\solkab$ will be at most
$2\del^5(1+\del)^i$.

We will also have another type of gaps for each interval
$\J_{i,u}$ which we call {\it postpone-gap} of size
$\del^5(1+\del)^i$, that will be used for the assignment of jobs
that the solutions $\solkab$ (for all values of $k$) assign to
earlier time intervals, and we decided to postpone to this time
interval.  We will ensure that the total size of jobs assigned to
this type of gap of $\J_{i,u}$ will be at most $\del^5(1+\del)^t$.

We will ensure that the positions of the two types of gaps in a
common time interval are consecutive, but the exact starting time
may be changed as we schedule other jobs to start or complete
during $\J_{i,u}$.

The starting times of jobs in an interval $\J_{i,u}$ will satisfy the following properties.  If there are jobs of $\ak$ that are scheduled to start during $\J_{i,u}$ and are not assigned to one of the gaps (either to the postpone-gap of $\J_{i,u}$ or to the sparse-gap of $\J_{i,u}$), then their starting time in $\widehat{\sol}$ is exactly as it is in $\solkab$.  Other jobs may be assigned to one of the gaps, but in this case they will start and complete during the corresponding gap whose positions will be fixed later.

A time interval $\J_{i,u}$ can be free (and all time intervals are
initialized to be free), or taken, and taken time intervals are
assigned a pair $(k',t)$ based on an interval $\J_{t,u}$ and an instance
$\akk$. In this case, we will charge the increase of the cost due to
postponing jobs which were supposed to start during $\J_{i,u}$ in
any further solution $\solkkkab$ ($k''<k'$) to a postpone-gap in later time
intervals.  We will keep the invariant that the set of free
intervals are a suffix of the list of intervals (ordered from
earliest to latest).  Taken time intervals are defined in one of two cases.  Either there are jobs assigned to the postpone-gap of the interval or a later interval, and in this case for the analysis we will declare an {\it intermediate pair} for this interval, or the total size of jobs that are assigned to start during $\J_{i,u}$ exceeds $\del^5(1+\del)^i$ (or similarly for a later interval) and in this case no intermediate pair is declared for this interval.

We start with an empty schedule $\widehat{\sol}$, where every time
interval is free.  For every time interval, we decide if we have a
sparse-gap in it or not (in future iterations we may decide to
remove a sparse-gap but this may happen only if no job is assigned
to this sparse-gap). That is, for $u> \frac{1}{\del^7}+1$, we decide to have a sparse-gap for every time interval (some of them cannot be used but we will find out this information as we combine jobs to the schedule $\widehat{\sol}$), and for $u \leq \frac{1}{\del^7}+1$, we use the information of the color of $u$ to decide which time intervals will have sparse-gap (every time interval that is not contained in a reserved period of a job).
The value of $k$ is set to the largest value
for which $\solkab$ processes at least one job on machine $u$.

Consider an iteration $k$ in which we incorporate the set of jobs
scheduled by $\solkab$ on machine $u$ into the schedule
$\widehat{\sol}$.  We process the time intervals $\J_{t,u}$ in
decreasing order of $t$. Each such value of $t$ corresponds to a
sub-iteration. Let $J_{k,t}$ be the set of jobs that $\solkab$
starts during $\J_{t,u}$.

First, assume that the interval $\J_{t,u}$ is free or is taken by a pair $(t',k)$ for $t'>t$. We will
schedule the jobs $J_{k,t}$ to start during the time interval
$\J_{t,u}$ as follows.

\begin{itemize}
\item Consider the case in which the total size of $J_{k,t}$ is at
most $\del^5(1+\del)^t$, then we schedule the jobs $J_{k,t}$ to
start during the sparse-gap of this time interval. If (after the jobs $J_{k,t}$ are assigned) the total
size of jobs assigned to this sparse-gap is larger than
$\del^5(1+\del)^t$ and $\J_{t,u}$ is free, then we declare that $\J_{t,u}$ and all free intervals $\J_{t',u}$ (for $t'<t$) are taken and
we assign those intervals the pair $(k,t)$ without an
intermediate pair.

\item{} Assume that the total size of $J_{k,t}$ is larger than
$\del^5(1+\del)^t$. In this case we observe that the interval $\J_{t+1,u}$  is either free or taken by the pair $(t',k)$ that has the same value of $k$ as in the current iteration (and similarly for later intervals). We assign starting times to the jobs of $J_{k,t}$ as they are in $\solkab$.  Observe that this gives a feasible schedule as there is a gap of length $\del^4(1+\del)^t$ in the schedule $\solkab$ during $\J_{t,u}$ and this is sufficient to run the other jobs that are currently scheduled to this time interval in $\widehat{\sol}$.
If $\J_{t,u}$ is free, then we declare
 $\J_{t,u}$ as taken and assign it the pair $(k,t)$. If $\J_{t,u}$ is declared taken in this sub-iteration,
then  we also declare all the free intervals $\J_{t',u}$ (for
$t'<t$) as taken and assign them the pair $(k,t)$ without an
intermediate pair.
Moreover, if one of the jobs of $J_{k,t}$ is scheduled to complete
(in $\solkab$) at time at least $(1+\del)^{t+1}$, then we already
know that the set of intervals that belong to its reserved time
period (in $\solkab$) were not assigned any job, and all these
time intervals are declared as taken and assigned the pair $(k,t)$
without an intermediate pair.   Furthermore, consider the job $j$ of $J_{k,t}$ that we decide to schedule last, and assume that it completes in a different time interval $\J_{t',u}$.  Then, if $\J_{t',u}$ is free, and the total size of $j$ that we decided to process during $\J_{t',u}$ is at least $\del^5(1+\del)^{t'}$, then we declare the interval $\J_{t',u}$ as taken and assign it the pair $(k,t)$ without an intermediate pair.
\end{itemize}

Next, assume that the interval $\J_{t,u}$ is taken and its assigned
pair is $(k',t')$ where $k'>k$. Then, all remaining jobs of $\ak$ that were
not assigned by previous sub-iterations (not only of one time interval) are assigned to the last
postpone-gap before time $\frac{1}{\del^{10(k'-k)}}(1+\del)^{t'}$. In this case
we say that these jobs were postponed to this postpone-gap with
difference $k'-k$, and we also say that these postponed jobs of $\ak$
{\it charge} the pair $(k',t')$ with difference $k'-k$. If the
time interval containing this postpone-gap or any earlier time
interval were free, then we declare all of these free time
intervals as taken and assign these time intervals the pair
$(k',t')$ via the intermediate pair $(k,t)$.

Note that each postpone-gap gets jobs of at most one sub-iteration
as afterwards its time interval is declared taken.  That is, in any further iteration, it will not be assigned more jobs.

We note that whenever the algorithm tries to schedule (without postponing) a job $j$ whose reserved time period in its corresponding schedule $\solkab$ contains an interval $\J_{i,u}$ then the sparse-gap of $\J_{i,u}$ is empty.  This holds for all machines with index at least $\frac{1}{\del^7}+2$ as no interval has a sparse gap,  for $u \leq \frac{1}{\del^7}+1$ and time interval that starts no later than the maximum release date, the claim holds by the palette. Last, consider a later time interval on $u$. Assume by contradiction that the interval $\J_{i,u}$ was assigned some jobs to its sparse-gap. These jobs were assigned to the sparse-gap of $\J_{i,u}$ in a previous iteration of the algorithm (as $\solkab$ is feasible and the entire interval $\J_{i,u}$ is reserved for $j$ in $\solkab$).  In such a previous iteration $k'$, when we add jobs to the sparse-gap of $\J_{i,u}$ we must try to schedule the jobs that are processed by $u$ during $\J_{i-1,u}$ in the solution $\solkkab$.  These jobs are either added to the solution or postponed, and in both cases this makes the interval $\J_{i-1,u}$ taken.  This fact ensures that the interval in which $j$ starts is also taken before iteration $k$ starts, and this contradicts the assumption that the algorithm tries to
schedule $j$ without postponing it.

Furthermore, a free interval $\J_{i,u}$ may be assigned a set of jobs to its sparse-gap of total size at most $\del^5(1+\del)^i$ and perhaps one partial job such that the size that is processed during $\J_{i,u}$ is below $\del^5(1+\del)^i$, and the processing of these jobs will not contradict the schedule of any set of jobs starting in the interval $\J_{i,u}$ in any of the schedules $\solkab$ since the gap in $\solkab$ appears as early as possible in this interval.

Observe that if there exists a set of jobs that charge a pair $(k',t')$, then the
total size of jobs that $\widehat{\sol}$ schedules during
$\J_{t',u}$ that were jobs of instances $\akkk$ for $k''\geq k$
and were scheduled by $\solkkkab$ to $\J_{t',u}$ is at least
$\del^5(1+\del)^{t'}$.

\begin{claim}\label{postpone_job_size1}
For every pair $(k',t')$, the set of jobs that are postponed with
difference $\mu=k'-k$ that charge the pair $(k',t')$ has total
size at most $(1+\del)^{t'}\cdot \frac{1}{\del^{10(\mu-1)+3}}$.
\end{claim}
\begin{proof}
We first argue that if there are jobs of $\ak$ that are postponed and charge
$(k',t')$ without an intermediate pair, then these postponed
jobs of $\ak$ have total size at most $(1+\del)^{t'}\cdot
\frac{1}{\del^{3}}$.
 Since the set of jobs that
$\widehat{\sol}$ starts during $\J_{k',t'}$ is completed no later
than $\frac{(1+\del)^{t'+1}}{\del}$ as it is timely, and since the
jobs of $\ak$ which are postponed are not postponed via an
intermediate pair,  the jobs of $\ak$ that are postponed start (in $\solkab$) no later than the end of the time
interval containing $\frac{(1+\del)^{t'+1}}{\del}$ and thus are
completed by $\frac{(1+\del)^{t'+2}}{\del^2}\leq
\frac{(1+\del)^{t'}}{\del^3}$, and thus the claim follows.

Next, consider the case where the jobs of $\ak$ are postponed and
charge $(k',t')$ via the intermediate pair $(\hat{k},\hat{t})$.
Then, $k'>\hat{k}>k$ and using the postponing rule, the postponed jobs of ${A^{(\hat{k})}}$ are postponed to
a postpone-gap not later than
$\frac{1}{\del^{10(k'-\hat{k})}}(1+\del)^{t'}$. Therefore, the set of
jobs that $\ak$ which are postponed via the intermediate pair
$(\hat{k},\hat{t})$ are starting in $\solkab$ not later than the
time interval that contains this postpone-gap, and thus not later
than $\frac{1}{\del^{10(k'-\hat{k})}}(1+\del)^{t'+1}$, and thus
complete (in $\solkab$) not later than
$\frac{1}{\del^{10(k'-\hat{k})+1}}(1+\del)^{t'+1}\leq
(1+\del)^{t'}\cdot \frac{1}{\del^{10(\mu-1)+3}}$, where the
inequality holds using $k'-\hat{k} \leq \mu-1$, and the claim
follows.
\end{proof}

\begin{claim}\label{postpone_job_size2}
For every value of $t''$, the total size of the set of jobs that are postponed to a
common postpone-gap during  $\J_{t'',u}$ is at most
$\del^{10}(1+\del)^{t''}$. Thus, the resulting schedule
$\widehat{\sol}$ is feasible.
\end{claim}
\begin{proof}
Let $k$ be such that jobs of $\solkab$ were postponed to be
scheduled during $\J_{t'',u}$.  That is, there exist values of
$k'$ and $t'$ such that the set of jobs of $\ak$ of total length
$\Pi_{k}$ are postponed charging the pair $(k',t')$. Note that
each postpone-gap gets jobs of at most one sub-iteration as
afterwards its time interval is declared taken (and the next assignment of jobs will take place in a new iteration).  Thus, the values
of $k,k',t'$ are unique for a given time interval $\J_{t'',u}$
whose postpone-gap gets jobs.  We let $\mu=k'-k$, then by Claim
\ref{postpone_job_size1}, we conclude that $\Pi_{k} \leq
(1+\del)^{t'}\cdot \frac{1}{\del^{10(\mu-1)+3}}$. Since
$\widehat{\sol}$ is timely, the last postpone-gap before time
$\frac{1}{\del^{10\mu}}(1+\del)^{t'}$ is starting not earlier than
$\frac{1}{\del^{10\mu-1}}(1+\del)^{t'-2}$, and thus its time interval
is starting after time $T=\frac{1}{\del^{10\mu-2}}(1+\del)^{t'}$.
To prove the claim, we need to show that $\Pi_k \leq \del^5 T$,
and thus it suffices to show that $(1+\del)^{t'}\cdot
\frac{1}{\del^{10(\mu-1)+3}} = (1+\del)^{t'}\cdot
\frac{1}{\del^{10\mu-7}} = (1+\del)^{t'} \del^5
\frac{1}{\del^{10\mu-2}} = \del^5 T$, which holds for all $\mu$ as $\del\leq
1$.
\end{proof}

\begin{claim}
The total pseudo-cost of machine $u$ in the resulting schedule
satisfies $\widehat{\sol} \leq (1+\del) \sum_k \solkab$ that is
$1+\del$ times the total pseudo-cost of the jobs assigned to
machine $u$ in all the solutions $\solkab$.
\end{claim}
\begin{proof}
We consider one specific pair $(k',t')$ for which the job set
$J_{k',t'}$ assigned to start during $\J_{t',u}$ in $\widehat{\sol}$ are not postponed jobs, and let $c^{k',t'}$ be the
total pseudo-cost of $J_{k',t'}$ . We also consider the sets of jobs
charging this pair. These are postponed jobs of smaller values of
$k$, and the total pseudo-cost of jobs of $\ak$ that charge the
pair $(k',t')$ (in the solution $\widehat{\sol}$) is denoted by
$c_{k}$, and we denote by $J_k$ the set of jobs of $\ak$ that are
postponed jobs (and charge the pair $(k',t')$).

In order to prove the claim it suffices to show that $\sum_{k<k'}
c_k \leq \del c^{k',t'}$.  We know by Claim
\ref{postpone_job_size1}, that the total size of the jobs of $J_k$
 is at most $(1+\del)^{t'}\cdot
\frac{1}{\del^{10(k'-k-1)+3}}$, and we concluded above that the
total size of the jobs of $J_{k',t'}$ is at least
$\del^5\cdot(1+\del)^{t'}$ and this is at least $\del^{2+10(k'-k)}$
times the total size of the jobs of $J_k$.

On the other hand, the minimum density of a job in $J_{k',t'}$ is
at least $(1+\del)^{\xi(k'-k)}\geq (\frac{1}{\del})^{25(k'-k)}$
times the maximum density of a job of $J_{k}$.

Therefore, the total pseudo-cost of the jobs of all the sets $J_k$
is at most
\begin{eqnarray*}
\sum_{k<k'} c_k &\leq & c^{k',t'} \cdot \sum_{k<k'}
\del^{25(k'-k)} \cdot \frac{1}{\del^{2+10(k'-k)}} \cdot
\frac{1}{\del^{10(k'-k)}}\\
&\leq& c^{k',t'} \cdot \frac{1}{\del^2} \sum_{\mu=1}^{\infty}
\del^{5\mu}
\leq c^{k',t'} \cdot \frac{\del^3}{1-\del^{5}} \leq \del
c^{k',t'} .
\end{eqnarray*}
\end{proof}

\begin{claim}
The solution obtained from $\widehat{\sol}$ by permuting the jobs
starting after time $\alpha^{\frac{\alpha}{\del}-1}$ to be
processed according to non-increasing order of their densities is
a solution of total pseudo-cost of at most $(1+\del) \cdot
\sum_{k} \solkab$ and it satisfies properties \ref{prop2} and
\ref{no_large_prop}.  Furthermore, the palette of $\widehat{\sol}$
is $(1+\del)^3 \pi$.
\end{claim}
The claim follows because $\solka$ (and therefore also $\solkab$)
satisfies that the completion time of any machine is at most
$\alpha^{\frac{\alpha}{\del}-1}\cdot \frac{\hat{y}}{\del^{25}}$,
and the fact that in $\widehat{\sol}$ the allocation of jobs to
machines is the same as in the collection of solutions $\{
\solka\}_k$ (this is the same allocation as in $\{ \solkab\}_k$).

\subsection{EPTAS for bounded inputs (with release dates)}
In this section, we show how to approximate (within a factor of $(1+\del)^4$) any input of the form $\ak$ with palette $\pi$ and pink machine $v$.  Specifically, we assume that the jobs of the instance have densities that are integer powers of $1+\del$ in the interval $[1,(1+\del)^y]$, the release date of every job is an integer power of $1+\del$ in the interval $[1,R]$ where $R=\alpha^{\frac{\alpha}{\del}-1}$, and the size of every job, the weight of every job, and the speed of every machine, are integer powers of $1+\del$.  Moreover, we know that there exists a near optimal schedule in which every machine completes process all jobs assigned to it by time $L=\alpha^{\frac{\alpha}{\del}-1}\cdot \frac{\hat{y}}{\del^{25}}$.  Observe that $y,R,L$ are functions of $\del$.

We apply the following preprocessing of the instance.  For every non-negative integer value of $t$ such that $(1+\del)^t \leq R$, we denote the set of jobs released at time $(1+\del)^t$ by $\J_t$, and by $\hat{m}_t$ the machine of smallest index (according to the palette $\pi$) which has a gap or starting time of a job during $I_t=[(1+\del)^t,(1+\del)^{t+1})$ (observe that $\hat{m}_t\leq v$ because $v$ is pink).  If the set of jobs $\J_t$ can be scheduled on machine $\hat{m}_t$ and their total processing time on this machine $\hat{m}_t$ is at most $\del^4(1+\del)^t$ (that is, their total size is at most $\del^3 s_{\hat{m}_t}$ times the length of the interval), then we will schedule these jobs during the time interval $I_t$ on machine $\hat{m}_t$, and remove these jobs from the instance.  We apply this preprocessing for every value of $t$, and denote by $\J$ the resulting job set.

\begin{lemma}
Let $\sol$ be a feasible solution for the resulting instance (with job set $\J$) whose total pseudo-cost is at most $1+\kappa\del$ times the total pseudo-cost of an optimal solution $\opt_{\J}$ for that instance, then the solution $\sol'$ resulting from $\sol$ by first scheduling the removed jobs as suggested by the preprocessing and afterwards apply time stretching by a factor of $1+\del$ on the resulting solution,  is a feasible solution (with a palette obtained from $(1+\del)\pi$) whose cost is at most $(1+\del)(1+\kappa\del)$ times the total pseudo-cost of an optimal solution for the original instance.
\end{lemma}
\begin{proof}
Denote  the optimal solution for the original instance with respect to the total pseudo-cost by $\opt$.  Let $\J'$ be the set of jobs not in $\J$, that is, the set of jobs that we removed during the preprocessing, and let $\Delta$ be the total pseudo-cost of the jobs in $\J'$ according to $\opt$.  Then since $\opt$ is a feasible schedule, the set of jobs in $\J' \cap \J_t$ has total pseudo-cost in $\opt$ which is at least $(1+\del)^{t+1} \cdot \sum_{j\in \J' \cap \J_t} w_j$, and in $\sol'$ this total pseudo-cost is $(1+\del)^{t+2} \cdot \sum_{j\in \J' \cap \J_t} w_j$.  As for the other jobs, note that $\opt$ gives a feasible schedule for the instance with job set $\J$, whose total pseudo-cost is at least $\opt - \sum_t (1+\del)^{t+1} \cdot \sum_{j\in \J' \cap \J_t} w_j$, and thus the total pseudo-cost of $\opt_{\J}$ is at most this bound.  Thus the total pseudo-cost of $\sol'$ is at most $(1+\del)(1+\kappa\del) \cdot \left( \opt - \sum_t (1+\del)^{t+1} \cdot \sum_{j\in \J' \cap \J_t} w_j \right) + \sum_t (1+\del)^{t+2} \cdot \sum_{j\in \J' \cap \J_t} w_j \leq (1+\del)(1+\kappa\del)\opt$, and the claim follows.
\end{proof}

Let $\gamma=\frac{\del^{20}}{(1+\del)^y (y+1)(\log_{1+\del}R+1)}$.  We define a job $j$ of size $(1+\del)^i$ to be huge on machine $u$ of speed $s$ if $\frac{(1+\del)^i}{s} \geq L$ (these jobs cannot be assigned to $u$), $j$ is large on $u$ if $\frac{(1+\del)^i}{s} \in [ \gamma, L)$, and otherwise it is small.  Observe that every machine $u$ is assigned at most $\frac{L}{\gamma}$ large jobs, and this bound is a function of $\del$.

We define a type of job $j$ to be the vector consisting of its size, weight, and release date.  We say that a job type is large for machine $u$ if a job of this type is large for $u$.  Note that the number of job types that are large for $u$ is at most $(\lceil \log_{1+\del}\frac{L}{\gamma} \rceil+1) \cdot (y+1) (1+\log_{1+\del} R)$.  We let $\J_{i,r,t}$ denote the set of jobs of  size $(1+\del)^i$, density $(1+\del)^r$ (and thus weight $(1+\del)^{i+r}$) and release date $(1+\del)^t$, and let $n_{i,r,t}=|\J_{i,r,t}|$ be the number of such jobs.

We next define machine types as follows.  Each machine of index $u\leq \frac{1}{\del^7} +1$ has its own type, the set of machines of indices at least $\frac{1}{\del^7} +2$ is partitioned into machine types, such that two machines of indices at least $\frac{1}{\del^7} +2$ have the same type if they have a common speed.  Let $T$ denote the set of machines types, and for $\sigma\in T$ we let $m_{\sigma}$ denote the number of machines of type $\sigma$, and by $s_{\sigma}$ the common speed of machines of type $\sigma$.

Let $\hat{s}$ be the speed of the machine of index $\frac{1}{\del^7} +2$ (this is a fastest machine of index at least $\frac{1}{\del^7} +2$).  We say that a machine type $\sigma$ is slow (and every machine of this type is a slow machine) if $s_{\sigma} < B \hat{s}$ for a constant $B$ depending on $\del$ that we will define later.  A machine type $\sigma$ is fast if it is not slow.  We let $T_{fast}$ be the set of fast machine types, and $T_{slow}=T\setminus T_{fast}$ be the set of slow machine types.  Then $|T_{fast}|\leq \frac{1}{\del^7} +2+ \lceil \log_{1+\del} \frac{1}{B} \rceil$,  which is again a constant for every constant value of $B$.

We define a machine configuration as a vector that defines a schedule of one machine, and we will denote by $\C$ the set of all configurations.  For a configuration $C$, the first component $s(C)$ is an integer such that the speed of every machine with configuration $C$ is $(1+\del)^{s(C)}$, and the second component is the machine type $\sigma(C)\in T$ denoting that every machine with configuration $C$ has machine type $\sigma(C)$.  For every $0 \leq r\leq y$, $\log_{1+\del} \gamma (1+\del)^{s(C)} \leq i \leq \log_{1+\del} L(1+\del)^{s(C)}$, $0\leq t \leq \log_{1+\del}R$, and $t \leq i' \leq \log_{1+\del} L$, we have a component $N_C(r,i,i',t)$ denoting the number of jobs of density $(1+\del)^r$, size $(1+\del)^i$ and release date $(1+\del)^t$ that are scheduled to start during the time interval $[(1+\del)^{i'},(1+\del)^{i'+1})$ as large jobs for this machine.  Moreover, for every $0 \leq r\leq y$, $0\leq t \leq \log_{1+\del}R$, and $t\leq i' \leq \log_{1+\del} L$, we have a component $n_C(r,i',t)$ expressing that the total size of jobs of density $(1+\del)^r$ and release date $(1+\del)^t$ that are scheduled to start during the time interval $[(1+\del)^{i'},(1+\del)^{i'+1})$ as small jobs for this machine is in the interval $((n_C(r,i',t)-1)\gamma(1+\del)^{s(C)}, n_C(r,i',t)\gamma(1+\del)^{s(C)}]$.

We process the set $\C$ and remove  infeasible configurations from it.  More precisely, first if the second component is a machine type corresponding to one of the machines with index at most $\frac{1}{\del^7} +1$, we verify that the configuration does not contradict the color of this machine (as indicated by the palette), that is, if a job starts in a time interval, then the color of the machine does not forbid this.  Next, we consider configurations with second component corresponding to machine types of machines with indices at least $\frac{1}{\del^7} +2$, and if the configuration corresponds to sparse machine, then we discard this configuration (as we showed in Lemma \ref{nosparselemma} that there exists a near optimal solution where no sparse machine with index at least $\frac{1}{\del^7} +2$ exists). Here we relax this condition and apply the following deletion rule.  If there exists $i'$ for which $0<\sum_{r,i,t} N_C(r,i,i',t)(1+\del)^i+\sum_{r,t} n_C(r,i',t)\gamma(1+\del)^{s(C)}\leq \del^5(1+\del)^{i'}(1+\del)^{s(C)}$, then we delete $C$ from the list of configurations.
 Finally, for all machine types, we consider each configuration $C$ and try to schedule jobs in each time interval without creating violations.  To do so, we process the time intervals with increasing index of $t$, and whenever we reach a time interval for which the configuration $C$ defines a set of jobs to be started during the interval (as small or large jobs), we schedule the jobs in a non-decreasing order of their sizes, and we allocate total size of $(n_C(r,i',t)-1)\gamma(1+\del)^{s(C)}$ to small jobs of density  $(1+\del)^r$ and release date $(1+\del)^t$  to be scheduled during the time interval $[(1+\del)^{i'},(1+\del)^{i'+1})$.  If a violation occurs, i.e., if we try to schedule a set of jobs to start in a given time interval (after the last job of previous time intervals has completed), and the last such job for a given time interval does not start in its time interval, then the configuration is said to be infeasible and we remove it from $\C$.  Otherwise, we constructed a {\it virtual schedule} for a configuration $C$ of a set of jobs (not necessarily jobs that exist in $\J$) and we compute the total pseudo-cost of  this virtual schedule, which we will denote by $cost(C)$.  We further remove  all configurations for which the virtual schedule is not timely from $\C$.  Note that an empty set of jobs  gives a feasible configuration for every machine type (recall that this does not correspond to a sparse machine).

Let $\C_{\sigma}$ be the set of configurations for which the second component is $\sigma$.  Then, for every $\sigma$, we have that $|\C_{\sigma}|$ is at most $$ D=\left( \frac{L}{\gamma} \right)^{(y+1)[(\log_{1+\del} \frac{L}{\gamma}) +3][(\log_{1+\del} {L}+2](\log_{1+\del}R+1)}$$ which is constant for every value of $\sigma$.  We will use the value of $B$ as $B=\frac{\del^6}{L(1+\del)^2 D}$ which is indeed a fixed constant as we declared.   We denote by $\C_{fast}=\cup_{\sigma\in T_{fast}} \C_{\sigma}$, and thus $|\C_{fast}|$ is a constant term  (a function of $\del$).

Next, we define a mixed-integer linear program $\Pi$.  The decision variables are as follows.  For every configuration $C\in \C$, we have a variable $X_C$ denoting the number of machines with configuration $C$.  We will require $X_C$ to be integer if $C\in \C_{fast}$, and otherwise we will allow $X_C$ to be fractional.  The other set of decision variables are $Y_{C,r,i,i',t}$, denoting the number of jobs of density $(1+\del)^r$, size $(1+\del)^i$ and release date $(1+\del)^t$ that are scheduled to start during the time interval $[(1+\del)^{i'},(1+\del)^{i'+1})$ as small jobs on a machine with configuration $C$.  The last set of decision variables is allowed to be fractional.  The variable $Y_{C,r,i,i',t}$ exists only if a job of size $(1+\del)^i$ is small for a machine with configuration $C$ (that is for a machine with speed $s(C)$), and $t\leq i'$.  Using these decision variables, the mathematical program $\Pi$ is defined as follows.

\begin{eqnarray}
\min& \sum_{C\in \C} cost(C)X_C & \nonumber\\
s.t.& \sum_{C\in \C_{\sigma}} X_C = m_{\sigma} & \forall\  \sigma\in T \label{rel_cons1}\\
& \sum_{i' \geq t} \left( \sum_C N_C(r,i,i',t) X_C + \sum_C Y_{C,r,i,i',t}\right) = n_{i,r,t} & \forall\  0\leq r \leq y,\nonumber \\ & & \forall\  \log_{1+\del} \gamma (1+\del)^{s(C)} \leq i \leq \nonumber \\ && \leq \log_{1+\del} L(1+\del)^{s(C)},\nonumber \\ & & \forall\  0\leq t \leq \log_{1+\del}R \label{rel_cons2}\\
&\sum_i (1+\del)^i Y_{C,r,i,i',t} \leq n_C(r,i',t)\gamma(1+\del)^{s(C)} X_C& \forall\  C\in \C, \forall\  0\leq r\leq y, \nonumber\\
\\ && \forall\  0\leq t \leq \log_{1+\del}R,\nonumber \\
& & \forall\  t\leq i'\leq \log_{1+\del} L  \label{rel_cons3}\\
&X_C, Y_{C,r,i,i',t} \geq 0& \forall\  C\in \C, \forall\  0\leq r\leq y, \nonumber\\
&&\forall\  \log_{1+\del} \gamma (1+\del)^{s(C)} \leq i\leq \nonumber \\ && \leq \log_{1+\del} L(1+\del)^{s(C)},\nonumber \\&&\forall\  0\leq t \leq \log_{1+\del}R\nonumber \\&& \forall\  t \leq i' \leq \log_{1+\del} L\nonumber
\end{eqnarray}
Condition (\ref{rel_cons1}) ensures that we assign at most $m_{\sigma}$ machines to configurations with type $\sigma$. Condition (\ref{rel_cons2}) ensures that every job of size $(1+\del)^i$, density $(1+\del)^r$ and release date $(1+\del)^t$ is scheduled either as a large job or as a small job.  Finally, Condition (\ref{rel_cons3}) ensures that we do not try to schedule small jobs with total size that exceeds the total size of small jobs (with a given density and release date) that is allowed by the definition of machine configurations.
We denote by $(X^*,Y^*)$ an optimal solution to the mixed-integer linear program $\Pi$ whose cost is $Z^*$.

\begin{lemma}
Let $\sol$ be a feasible timely schedule satisfying the palette $\pi$ with no sparse machine of index at least $\frac{1}{\del^7}+2$ to the input consisting of the job set $\J$ of total pseudo-cost $\sol$ such that $\sol$ does not have a sparse machine of index at least $\frac{1}{\del^7} +2$.  Then $Z^*\leq \sol$.
\end{lemma}
\begin{proof}
Based on $\sol$, we will define a feasible integer solution to $\Pi$ whose cost as a solution to $\Pi$ is at most $(1+\del)\sol$.  For every machine $\lambda$, we define a configuration of $\lambda$ which we denote by $C_{\lambda}$ according to $\sol$ as follows.  the first component $s(C_{\lambda})$ is defined such that the speed of machine $\lambda$ is $(1+\del)^{s(C_{\lambda})}$, and the second component is the type of $\lambda$.  For every $0 \leq r\leq y$, $\log_{1+\del} \gamma (1+\del)^{s(C)} \leq i \leq \log_{1+\del} L(1+\del)^{s(C)}$, $0\leq t \leq \log_{1+\del}R$, and $t \leq i' \leq \log_{1+\del} L$, the component $N_{C_{\lambda}}(r,i,i',t)$ is the number of jobs of density $(1+\del)^r$,  size $(1+\del)^i$ and release date $(1+\del)^t$ that $\sol$ schedules to start during the time interval $[(1+\del)^{i'},(1+\del)^{i'+1})$ as large jobs for $\lambda$.  Finally, for every $0 \leq r\leq y$, $0\leq t \leq \log_{1+\del}R$, and $t\leq i' \leq \log_{1+\del} L$, the component $n_{C_{\lambda}}(r,i',t)$ is calculated by first computing the total size of jobs of density $(1+\del)^r$ and release date $(1+\del)^t$ that are scheduled to start during the time interval $[(1+\del)^{i'},(1+\del)^{i'+1})$ as small jobs for this machine and dividing the result by $\gamma(1+\del)^{s(C)}$ and then we round up the result to the next integer. Observe that the configuration $C_{\lambda}$ is a feasible configuration as it satisfies the palette, and for every machine $\lambda$, the virtual schedule starts in each interval a set of jobs of smaller total size compared to the set of jobs started in this interval according to $\sol$.  Now, we set the variables $X_C$ to be the number of machines whose configuration according to $\sol$ is $C$, and we compute $Y_{C,r,i,i',t}$ by counting the number of small jobs of each type that the solution $\sol$ schedules on machines with configuration $C$ in each time interval.  Clearly constraints (\ref{rel_cons1}) are satisfied because every machine has a configuration calculated to it, constraints (\ref{rel_cons2}) hold because every job is scheduled by $\sol$ either as a large job or as a small job, and (\ref{rel_cons3}) holds because when we calculate the value of $n_{C_{\lambda}}(r,i',t)$ we round up the total size assigned to this machine and time interval.

Regarding the cost of machine $\lambda$ in $\sol$ with respect to $cost(C_{\lambda})$, for every time interval $\J_{t,\lambda}$, the total weight of jobs included in $cost(C_{\lambda})$ is no larger than the total weight for this interval in $\sol$.
\end{proof}
The correctness of the scheme for the bounded instance is established using the following lemma.

\begin{lemma}
Let $(X^*,Y^*)$ be an optimal solution  for the mixed-integer linear program $\Pi$ with palette $\pi$.  Then, there is a polynomial time algorithm that transforms $(X^*,Y^*)$ into a feasible solution with palette  $(1+\del)^2\pi$ for the scheduling problem, such that the cost of this solution is at most $(1+\del)^3 \sum_{C\in \C} cost(C)X^*_C$.
\end{lemma}
\begin{proof}
We define an integral vector $X'$ that specifies the number of machines of each configuration by the following rule  $X'_C=\lfloor X^*_C \rfloor$, and the residual fractional vector $X''$ is defined using $X''_C=X^*_C-X'_C$.  Then, $\sum_{C\in \C} cost(C) X^*_C = \sum_{C\in \C} cost(C) X'_C + \sum_{C\in \C} cost(C) X''_C$.  We consider a schedule of some of the jobs of $\J$ as follows.  We schedule $X'_C$ machines according to configuration $C$.  For each such machine with configuration $C$, consider every time interval $[(1+\del)^{i'}, (1+\del)^{i'+1})$.  If the configuration has jobs with  starting times during this time interval and the machine has index at most $\frac{1}{\del^7}+1$, then $\pi$ allows the start of jobs during the time interval.
We allocate (at most) $\lceil \frac{Y_{C,r,i,i',t}}{X^*_C} \rceil$ jobs of density $(1+\del)^r$, size $(1+\del)^i$, release date $(1+\del)^t$ to start during the time interval $[(1+\del)^{i'}, (1+\del)^{i'+1})$. We also allocate (at most) $N_C(r,i,i',t)$ jobs of density $(1+\del)^r$,  size $(1+\del)^i$, and release date $(1+\del)^t$ to start during the time interval $[(1+\del)^{i'},(1+\del)^{i'+1})$ on this machine.  We order the jobs starting on this machine during this time interval in a non-decreasing order of their sizes.

Observe that the total length of small jobs assigned to time interval $[(1+\del)^{i'},(1+\del)^{i'+1})$ may exceed the amount $(n_C(r,i',t)-1)\gamma(1+\del)^{s(C)}$ that we used for the virtual schedule by at most $(\log_{1+\del}R +1)(y+1)\gamma(1+\del)^{s(C)} \cdot [ 1+ \sum_{q=0}^{\infty} \frac{1}{(1+\del)^q} ] \leq (1+\del)^{(s(C)}(\log_{1+\del}R +1)(y+1) \gamma \frac{2}{\del} \leq (1+\del)^{s(C)} \frac{\del^{18}}{(1+\del)^y}$ (the last inequality holds by the definition of $\gamma$).  This excess arises due to two reasons.  The first one is the difference between $(n_C(r,i',t)-1)\gamma(1+\del)^{s(C)}$ that we use for the creation of the virtual schedule and $n_C(r,i',t)\gamma(1+\del)^{s(C)}$ that bounds (in $\Pi$) the average total size of small jobs (of the corresponding density and release dates) assigned to such time interval.  The second is due to (at most) one additional small job of each release date, size, and density (due to the ceiling operation).  This total excess bound is less than $\del^{18}(1+\del)^{s(C)}$. Therefore, we apply time stretching by a factor of $1+\del$ of the schedule of these machines and obtain a feasible schedule (since for every time interval for which we increase the total size of small jobs assigned to start during it, in the schedule obtained by time stretching has a gap of length at least $\del^3$ times the length of the time interval, and this is at least $\del^4(1+\del)^{s(C)}$ even for the shortest such time interval).

Moreover,  by increasing the total size of jobs assigned to such a machine to start during the time interval $[(1+\del)^{i'},(1+\del)^{i'+1})$ by an additive term of $(1+\del)^{s(C)} \frac{\del^{18}}{(1+\del)^y}$ which is for a non-sparse machine at most $\frac{\del^{13}}{(1+\del)^y}$ times the total size of jobs starting on such a machine during $[(1+\del)^{i'},(1+\del)^{i'+1})$ (because the machine is not sparse, and we increase the total size of small jobs starting in this time interval only if some jobs start during this time interval), and thus the total increase of the cost due to the increase total size of small jobs starting during $[(1+\del)^{i'},(1+\del)^{i'+1})$ is at most $\del^{13}$ times the cost of the jobs starting during the same time interval according to the virtual schedule.  Thus, for every non-sparse machine which we assign jobs in the current partial schedule, the total pseudo-cost of the jobs assigned to it is at most $(1+\del)(1+ \del^{13})$ times the cost of its configuration.  The bound of $\del^{13}$ times the cost of the jobs starting during the same time interval according to the virtual schedule applies for every time interval that is not sparse on a given machine (on the virtual schedule of that machine).

We next bound the increase of the cost for sparse time intervals (this may happen only for machines with index at most $\frac{1}{\del^7} +1$).  We partition the increase of the cost into two parts.  First consider the contribution of jobs with release date at most $(1+\del)^{i'-1}$.  By the optimality of $(X^*,Y^*)$ as a solution for $\Pi$, we cannot modify the solution by allocating these jobs to the previous time interval (which would require changing the   configuration of this machine and the values of $Y^*$).   Therefore, the machine is busy (according to the virtual schedule) for a period of at least $\del^5(1+\del)^{i'-1}$ during this previous time interval, and again the total contribution of the increase of the total size of small jobs due to the variables $Y_{C,r,i,i',t}$ for $i' >t$ is at most $\del^{13}$ times $cost(C)$.  Next, consider the contribution of jobs released at $(1+\del)^{i'}$ to the cost of jobs starting during $[(1+\del)^{i'},(1+\del)^{i'+1})$ for the machine where this time interval is sparse.  Observe that the total size of jobs released at $(1+\del)^{i'}$ is at least $\del^5(1+\del)^{i'} (1+\del)^{s(C)}$ (as otherwise all these jobs are removed in the preprocessing step), and thus the total increase of the cost due to the small jobs of this sparse time interval is at most $\del^{13}$ times the total cost of jobs released at $(1+\del)^{i'}$ in $(X^*,Y^*)$.  Thus, the total contribution of all the jobs released at $(1+\del)^{i'}$ that are scheduled on sparse intervals  during the time frame $[(1+\del)^{i'},(1+\del)^{i'+1})$ (on all the first $\frac{1}{\del^7}+1$ machines) is at most $\del^{5}$ times the total cost of jobs released at $(1+\del)^{i'}$ in $(X^*,Y^*)$.
We conclude that the total cost of the jobs which we schedule so far is at most $(1+\del)(1+\del^{13})^2\sum_{C\in \C} cost(C)X'_C +\del^{5}\sum_{C\in \C} cost(C)X^*_C$.

Observe that the pink machine $v$ (of speed $s_v$) has either a starting time or idle time in every time interval in the current partial schedule.  We apply once again time stretching of the current partial schedule by a factor of $1+\del$ and create a solution of these jobs of cost at most $(1+\del)^2(1+\del^{13})^2\sum_{C\in \C} cost(C)X'_C +\del^{5}(1+\del)\sum_{C\in \C} cost(C)X^*_C$.

Denote by $J''$ the set of jobs which were not assigned so far and recall that $\hat{s}=s_{\frac{1}{\del^7} +2}$ and $B\hat{s}$ is an upper bound on the speed of configurations for which $X^*_C \neq X'_C$.
Note that $(X^*,Y^*)$ assigns jobs of total size at most $(1+\del)\cdot L\cdot s$ for any machine of speed $s$. Since there are at most $D$ machines that did not receive a configuration of each type, and by our choice of $B$, the total size of the jobs of $J''$ is at most $L\cdot (1+\del) \cdot  D \cdot B \cdot \hat{s} \sum_{q=0}^{\infty}\frac{1}{(1+\del)^q} = L \cdot \frac{(1+\del)^2}{\del} \cdot D \cdot B\cdot \hat{s}\leq L\frac{(1+\del)^2}{\del} \cdot D \cdot B \cdot s_v = \del^5 s_v$, where the last equation holds because $B=\frac{\del^6}{L(1+\del)^2 D}$. Thus, every unassigned job which is released at time $(1+\del)^t$ can be scheduled on machine $v$ during the time interval $[(1+\del)^t,(1+\del)^{t+1})$ (this is so because the total size of all the jobs in $J''$ is at most $\del^5 s_v$ and thus processing a subset of these jobs can be done within the gap of machine $v$ during the time interval $[(1+\del)^t,(1+\del)^{t+1})$).  We conclude that the total pseudo-cost of the jobs in $J''$ is at most $(1+\del)\sum_{j\in J''} r_j \leq (1+\del) \sum_{C\in \C} cost(C)X''_C$, where the last inequality holds because in $X''$ all the jobs of $J''$ are scheduled (fractionally) and the completion time of each job is always at least its release date.

Thus, the total cost of the resulting schedule is at most
\begin{eqnarray*}&&(1+\del)^2(1+\del^{13})^2\sum_{C\in \C} cost(C)X'_C +(1+\del)\del^{5}\sum_{C\in \C} cost(C)X^*_C + (1+\del) \sum_{C\in \C} cost(C)X''_C\\ &\leq &(1+\del)^3\left[ \sum_{C\in \C} cost(C)X'_C + \sum_{C\in \C} cost(C)X''_C \right]\\ &=& (1+\del)^3 \sum_{C\in \C} cost(C)X^*_C.
\end{eqnarray*}
\end{proof}
\bibliographystyle{abbrv}

\end{document}